\documentclass[runningheads]{llncs}
\usepackage{graphicx}
\usepackage[utf8]{inputenc}
\usepackage{amsmath,amssymb,mathrsfs}
\usepackage{braket,physics}
\usepackage{complexity}
\usepackage[ruled,linesnumbered,noend]{algorithm2e}
\usepackage{setspace}
\usepackage[noend]{algpseudocode}
\usepackage[colorlinks=true,citecolor=blue,linkcolor=blue]{hyperref}
\usepackage[nameinlink, noabbrev, capitalize]{cleveref}
\usepackage[dvipsnames]{xcolor}

\usepackage{subcaption}

\newcommand{\Enc}{\mathsf{Enc}}
\newcommand{\Auth}{\mathsf{Auth}}
\newcommand{\Dec}{\mathsf{Dec}}
\newcommand{\Aux}{\mathsf{Aux}}
\newcommand{\Anc}{\mathsf{Anc}}
\newcommand{\pair}{\mathsf{PI}}
\newcommand{\lr}{\mathsf{LR}}

\newcommand{\td}{\mathsf{TD}}
\newcommand{\NM}{\mathsf{NM}}
\newcommand{\QNM}{\mathsf{QNM}}
\newcommand{\sh}{\mathsf{SS}}
\newcommand{\nmss}{\mathsf{NMSS}}
\newcommand{\lrss}{\mathsf{LRSS}}
\newcommand{\tdss}{\mathsf{TDSS}}

\newcommand{\share}{\mathsf{Share}}
\newcommand{\rec}{\mathsf{Rec}}
\newcommand{\Copy}{\mathsf{Copy}}
\newcommand{\hyb}{\mathsf{Hyb}}

\newcommand{\sdshare}{\mathsf{SdShare}}
\newcommand{\sdrec}{\mathsf{SdRec}}
\newcommand{\mshare}{\mathsf{MShare}}
\newcommand{\mrec}{\mathsf{MRec}}

\newcommand{\Acc}{\mathsf{Acc}}
\newcommand{\Rej}{\mathsf{Rej}}

\newcommand{\locc}{\mathsf{LOCC}}
\newcommand{\lo}{\mathsf{LO}}

\newcommand{\sI}{\mathsf{I}}

\newcommand{\tDec}{\widetilde{\mathsf{Dec}}}
\newcommand{\Id}{\mathbb{I}}

\newtheorem{fact}{Fact}[section]

\definecolor{darkred}  {rgb}{0.5,0,0}
\definecolor{darkblue} {rgb}{0,0,0.5}
\definecolor{darkgreen}{rgb}{0,0.5,0}
\hypersetup{
  urlcolor   = blue,         
  linkcolor  = darkblue,     
  citecolor  = darkgreen,    
  filecolor   = darkred       
}

\usepackage[draft,inline,marginclue,index]{fixme}
\fxsetup{theme=color,mode=multiuser}
\colorlet{fxnote}{violet} 

\FXRegisterAuthor{Thiago}{aaa}{\colorbox{violet}{\color{white}Thiago}}
\newcommand{\thiago}{\thiagonote}

\DeclareMathOperator*{\argmin}{argmin}

\newcommand{\States}{\mathcal{D}}
\newcommand{\from}{\leftarrow}

\newcommand{\RSS}{\mathsf{RSS}}
\newcommand{\Share}{\mathsf{Share}}
\newcommand{\Reconstruct}{\mathsf{Reconstruct}}

\renewcommand{\epsilon}{\varepsilon}

\SetKwInput{KwInput}{Input}

\SetKwInput{KwOutput}{Output}

\usepackage{algpseudocode}%

\newcommand{\mycomment}[1]{\textup{{\color{red}#1}}}

\providecommand{\main}{.}
\usepackage{xcolor}

\usepackage[]{amsmath,amssymb,amsfonts,latexsym,enumerate,xcolor,bbm}

\usepackage[utf8]{inputenc}
\usepackage{amsmath,amssymb}
\usepackage{graphics,graphicx,tikz,color,float}
\usepackage{mathtools}
\usepackage{amsfonts}
\usepackage{tikz}
\usepackage{comment}
\usepackage{framed}
\usepackage[section]{placeins}
\usepackage[noend]{algpseudocode}
\usepackage{xcolor}

\usepackage{thm-restate}

\newcommand{\bits}{\{0,1\}}
\newcommand{\thetatamp}{\theta^{\mathsf{tamp}}}
\newcommand{\thetasame}{\theta^{\mathsf{same}}}

\newcommand{\rhotamp}{\rho^{\mathsf{tamp}}}
\newcommand{\rhosame}{\rho^{\mathsf{same}}}

\newcommand{\Supp}{\operatorname{Supp}}
\newcommand{\rk}{\textsf{rank}}
\newcommand{\beq}{\begin{equation}}
	\newcommand{\enq}{\end{equation}}
\newcommand{\bet}{\begin{theorem}}
	\newcommand{\ent}{\end{theorem}}
\newcommand{\ten}{\textnormal}
\newcommand{\tbf}{\textbf}
\newcommand{\vr}{\mathrm{Var}}
\newcommand{\myexp}{{\mathrm{e}}}
\newcommand{\rv}[1]{\mathbf{#1}}
\newcommand{\hilbertspace}[1]{\mathcal{#1}}
\newcommand{\universe}[1]{\mathcal{#1}}
\newcommand{\bound}{\mathsf{bound}}
\newcommand{\lab}{\mathrm{lab}}
\newcommand{\girth}{\mathrm{girth}}
\newcommand{\pseudoedge}[1]{\stackrel{#1}{\longrightarrow}}
\newcommand{\myexponent}{\frac{1}{\floor{n/4}}}
\newcommand{\err}{\mathrm{err}}
\newcommand{\eps}{\varepsilon}
\newcommand{\dinfty}{\ensuremath{\mathrm{D}_{\infty}}}
\newcommand{\dinftyeps}{\ensuremath{\mathrm{D}_{\infty}^{\eps}}}
\newcommand{\dseps}[3]{\ensuremath{\mathrm{D}_s^{#3}\left(#1\|#2\right)}}
\newcommand{\dzeroeps}{\ensuremath{\mathrm{D}_{\mathrm{H}}^{\eps}}}
\newcommand{\mmod}[1]{\hspace{1mm}\ensuremath{(\text{mod }#1)}}

\newcommand{\fid}{\mathsf{F}}
\newcommand{\swp}{\ensuremath{\mathrm{Swap}}}

\newcommand{\bell}{\ensuremath{\textbf{e}}}
\newcommand{\bfk}{\ensuremath{\textbf{k}}}
\newcommand{\brc}{\ensuremath{G}}
\newcommand{\brce}{\ensuremath{F}}
\newcommand{\qbit}{\ensuremath{Q}}
\newcommand*{\cC}{\mathcal{C}}
\newcommand*{\cA}{\mathcal{A}}
\newcommand*{\cR}{\mathcal{R}}
\newcommand*{\cH}{\mathcal{H}}
\newcommand*{\cM}{\mathcal{M}}
\newcommand*{\cF}{\mathcal{F}}
\newcommand*{\cB}{\mathcal{B}}
\newcommand*{\cD}{\mathcal{D}}
\newcommand*{\cG}{\mathcal{G}}
\newcommand*{\cO}{\mathcal{O}}
\newcommand*{\cK}{\mathcal{K}}
\newcommand*{\cN}{\mathcal{N}}
\newcommand*{\chS}{\hat{\mathcal{S}}}
\newcommand*{\cT}{\mathcal{T}}
\newcommand*{\chT}{\hat{\mathcal{T}}}
\newcommand*{\cX}{\mathcal{X}}
\newcommand*{\cW}{\mathcal{W}}
\newcommand*{\cZ}{\mathcal{Z}}
\newcommand*{\cE}{\mathcal{E}}
\newcommand*{\cU}{\mathcal{U}}
\newcommand*{\bP}{\mathbf{P}}
\newcommand*{\bq}{\mathbf{q}}
\newcommand*{\Ib}{\bar{I}}
\newcommand{\Br}{\mathrm{BR}}
\newcommand{\Ext}{\mathsf{Ext}}
\newcommand{\iExt}{\mathsf{IExt}}
\newcommand{\pre}{\mathsf{Crop}}
\newcommand{\advc}{\mathsf{AdvGen}}
\newcommand{\advcb}{\mathsf{AdvCB}}
\newcommand{\ff}{\mathsf{FF}}
\newcommand{\ecc}{\mathsf{ECC}}
\newcommand{\Q}{\mathcal{Q}}
\newcommand{\supp}{\mathrm{supp}}
\newcommand{\suppress}[1]{}
\newcommand{\drawn}{\leftarrow}
\newcommand{\defeq}{\ensuremath{ \stackrel{\mathrm{def}}{=} }}
\newcommand{\F}{\mathbb{F}}
\newcommand{\B}{\mathrm{B}}
\newcommand{\Pur}{\mathrm{P}}
\newcommand {\br} [1] {\ensuremath{ \left( #1 \right) }}
\newcommand {\cbr} [1] {\ensuremath{ \left\lbrace #1 \right\rbrace }}
\newcommand {\minusspace} {\: \! \!}
\newcommand {\smallspace} {\: \!}
\newcommand {\fn} [2] {\ensuremath{ #1 \minusspace \br{ #2 } }}
\newcommand {\ball} [2] {\fn{\mathcal{B}^{#1}}{#2}}
\newcommand {\relent} [2] {\fn{\mathrm{D}}{#1 \middle\| #2}}
\newcommand {\relentalpha} [3] {\fn{\mathrm{D}_{#3}}{#1 \middle\| #2}}
\newcommand {\dmax} [2] {\fn{\mathrm{D}_{\max}}{#1 \middle\| #2}}
\newcommand {\dmaxeps} [3] {\fn{\mathrm{D}_{\max}^{#3}}{#1 \middle\| #2}}
\newcommand {\mutinf} [2] {\fn{\mathrm{I}}{#1 \smallspace : \smallspace #2}}
\newcommand {\imax}{\ensuremath{\mathrm{I}_{\max}}}
\newcommand {\imaxeps}[3]{\ensuremath{\mathrm{I}^{#3}_{\max}(#1:#2)}}
\newcommand {\condmutinf} [3] {\mutinf{#1}{#2 \smallspace \middle\vert \smallspace #3}}
\newcommand {\hminone} [1] {\fn{ \mathrm{H }_{\min}}{#1}}
\newcommand {\hminn} [2] {\fn{ \mathrm{ \tilde{H} }_{\min}}{#1 \middle | #2}}
\newcommand {\hmin} [2] {\fn{ \mathrm{H }_{\min}}{#1 \middle | #2}}
\newcommand {\hmineps} [3] {\fn{\mathrm{H}^{#3}_{\min}}{#1 \middle | #2}}
\newcommand {\hmax} [2] {\fn{\mathrm{H}_{\max}}{#1 \middle | #2}}
\newcommand {\hmaxeps} [3] {\fn{\mathrm{H}^{#3}_{\max}}{#1 \middle | #2}}
\newcommand {\dheps} [3] {\ensuremath{\mathrm{D}_{\mathrm{H}}^{#3}\left(#1 \| #2\right)}}
\newcommand {\id} {\ensuremath{\mathbb{I}}}
\newcommand {\entsp}[2]{\ensuremath{\Delta_{#2}\left(#1\right)}}
\newcommand {\chnl}[1]{\ensuremath{\cN_{A\to B}{\left(#1\right)}}}
\newcommand {\wal}{\ensuremath{W^{alice}}}
\newcommand {\wbob}{\ensuremath{W^{bob}}}
\newcommand {\Hmin}{\mathrm{H}_{\min}}

\DeclareMathOperator{\QIC}{QIC}
\DeclareMathOperator{\dom}{{dom}}
\newcommand{\QCC}{\mathrm{QCC}}
\newcommand{\nmextc}{\mathsf{2nmext \mhyphen c}}
\newcommand{\newnmextc}{\mathsf{new \mhyphen 2nmext \mhyphen c}}
\newcommand{\nmextq}{\mathsf{2nmext \mhyphen q}}
\newcommand{\newnmextq}{\mathsf{new \mhyphen 2nmext \mhyphen q}}
\newcommand{\IExt}{\mathsf{IExt}}
\newcommand{\Trev}{\mathsf{Trev}}

\newcommand*{\eff}{ \tilde{\textbf{eff}}}
\newcommand*{\ef}{\textbf{eff}}

\newcommand*{\acc}{\mathsf{acc}}
\newcommand*{\rej}{\mathsf{rej}}

\newcommand{\samp}{\mathsf{Samp}}
\newcommand{\ind}{\mathcal{S}}
\newcommand{\indbar}{\overline{\mathcal{S}}}
\usepackage{afterpage}

\newcommand{\auth}{\mathrm{Auth}}
\newcommand{\ver}{\mathrm{Ver}}

\newcommand{\epspriv}{\eps_{\mathsf{priv}}}
\newcommand{\epsnm}{\eps_{\mathsf{nm}}}
\newcommand{\epslk}{\eps_{\mathsf{leak}}}

\newcommand*{\sm}{\mathsf{same}}
\newcommand*{\qpas}{\mathsf{qpa\mhyphen state}}
\newcommand*{\qmas}{\mathsf{qma\mhyphen state}}
\newcommand*{\nmas}{\mathsf{qnm\mhyphen state}}
\newcommand*{\qma}{\mathsf{qm\mhyphen adv}}
\newcommand*{\qca}{\mathsf{qc\mhyphen adv}}
\newcommand*{\qbsa}{\mathsf{qbs\mhyphen adv}}
\newcommand*{\qmra}{\mathsf{qmar\mhyphen adv}}
\newcommand*{\qia}{\mathsf{qi\mhyphen adv}}
\newcommand*{\gea}{\mathsf{ge\mhyphen adv}}
\newcommand*{\nma}{\mathsf{qnm\mhyphen adv}}
\newcommand*{\twonma}{\mathsf{2qnm\mhyphen adv}}
\newcommand*{\nmext}{\mathsf{nmExt}}
\newcommand*{\tnmext}{\mathsf{t\mhyphen nmExt}}

\newcommand*{\laext}{\mathsf{laExt}}
\newcommand{\oldmac}{\mathsf{MAC}}
\newcommand{\mac}{\mathsf{laMAC}}
\newcommand{\epr}{\mathsf{EPR}}
\newcommand{\swap}{\mathsf{Rep}}
\newcommand{\X}{\mathcal{X}}
\newcommand{\Y}{\mathcal{Y}}
\newcommand{\Z}{\mathcal{Z}}

\newcommand*{\cV}{\mathcal{V}}
\newcommand*{\sfal}{\mathsf{Alice}}
\newcommand*{\sfbo}{\mathsf{Bob}}
\newcommand*{\cY}{\mathcal{Y}}
\newcommand*{\hy}{\hat{y}}
\newcommand*{\hx}{\hat{x}}
\newcommand*{\bbN}{\mathbb{N}}
\newcommand*{\bbR}{\mathbb{R}}
\newcommand*{\bbC}{\mathbb{C}}
\newcommand*{\bX}{\mathbf{X}}
\newcommand*{\bU}{\mathbf{U}}
\newcommand*{\bY}{\mathbf{Y}}
\newcommand*{\bQ}{\mathbf{Q}}
\newcommand*{\cl}{\mathcal{l}}
\newcommand*{\Xb}{\bar{X}}
\newcommand*{\Yb}{\bar{Y}}
\newcommand*{\mximax}[1]{\Xi^\delta_{\max}(P_{#1})}
\newcommand*{\mximin}[1]{\Xi^\delta_{\min}(P_{#1})}
\newcommand*{\bx}{\mathbf{x}}
\newcommand*{\by}{\mathbf{y}}
\mathchardef\mhyphen="2D
\newcommand*{\uI}{\underline{I}}
\newcommand*{\uH}{\underline{H}}
\newcommand*{\oH}{\overline{H}}
\newcommand*{\renyi}{R\'{e}nyi }
\newcommand{\argmax}{\operatornamewithlimits{arg\ max}}
\newcommand*{\prank}{\mathrm{psdrank}}
\newcommand*{\enc}{\mathrm{Enc}}
\newcommand*{\qenc}{\mathrm{2NMC}}
\newcommand*{\qdec}{\mathrm{2NMD}}
\newcommand*{\lrenc}{\mathrm{lrShare}}
\newcommand*{\qshare}{\mathrm{qShare}}
\newcommand*{\qrec}{\mathrm{qRec}}
\newcommand*{\tssenc}{\mathrm{t \mhyphen SS}}
\newcommand*{\tssencq}{\mathrm{t \mhyphen QS}}
\newcommand*{\twossenc}{\mathrm{2 \mhyphen SS}}
\newcommand*{\tssdec}{\mathrm{t \mhyphen Rec}}
\newcommand*{\tssdecq}{\mathrm{t \mhyphen QR}}
\newcommand*{\twossdec}{\mathrm{2 \mhyphen Rec}}
\newcommand*{\nmenc}{\mathrm{nmShare}}
\newcommand*{\dec}{\mathrm{Dec}}
\newcommand*{\lrdec}{\mathrm{lrRec}}
\newcommand*{\nmdec}{\mathrm{nmRec}}
\newcommand*{\cenc}{\mathrm{cEnc}}
\newcommand*{\cdec}{\mathrm{cDec}}
\newcommand*{\nmcenc}{\mathrm{\nmext}}
\newcommand*{\nmcdec}{\mathrm{\nmext}}

\newcommand{\nmshare}{\mathrm{nmShare}}
\newcommand{\nmrec}{\mathrm{nmRec}}

\newcommand*{\cSC}{\mathcal{SC}}
\newcommand*{\denc}{\mathrm{d-enc}}
\makeatletter
\newcommand*{\rom}[1]{\expandafter\@slowromancap\romannumeral #1@}
\makeatother

\mathchardef\mhyphen="2D

\newcommand\llrightarrow[2][]{\RightExtendSymbol{-}{-}{\rightarrow}{#1}{#2}}
\newcommand\llrightleftarrow[2][]{\RightExtendSymbol{\leftarrow}{-}{\rightarrow}{#1}{#2}}
\newcommand\llleftarrow[2][]{\RightExtendSymbol{\leftarrow}{-}{-}{#1}{#2}}
\newcommand\llRightarrow[2][]{\RightExtendSymbol{=}{=}{\Rightarrow}{#1}{#2}}
\newcommand\llRightLeftarrow[2][]{\RightExtendSymbol{\Leftarrow}{=}{\Rightarrow}{#1}{#2}}
\newcommand\llLeftarrow[2][]{\RightExtendSymbol{\Leftarrow}{=}{=}{#1}{#2}}

\newfloat{Protocol}{tbhp}{lop}

\begin{document}
\title{On Split-State Quantum Tamper Detection and Non-Malleability}

\author{Thiago Bergamaschi\inst{1} 
\and Naresh Goud Boddu\inst{2}~\thanks{The entire work was carried out while NGB was a postdoctoral research fellow at NTT Research, Sunnyvale, USA.}}
\institute{UC Berkeley, 
\email{thiagob@berkeley.edu} \and NTT Research, 
\email{naresh.boddu@ntt-research.com}}


%
%\maketitle   % typeset the header of the contribution
%

\maketitle

\begin{abstract}
Tamper-detection codes (TDCs) are fundamental objects at the intersection of cryptography and coding theory. A TDC encodes messages in such a manner that tampering the codeword causes the decoder to either output the original message, or reject it. In this work, we study quantum analogs of one of the most well-studied adversarial tampering models: the so-called $t$-split-state tampering model, where the codeword is divided into $t$ shares, and each share is tampered with ``locally". \\

It is impossible to achieve tamper detection in the split-state model using classical codewords. Nevertheless, we demonstrate that the situation changes significantly if the message can be encoded into a multipartite quantum state, entangled across the $t$ shares. Concretely, we define a family of quantum TDCs defined on any $t\geq 3$ shares, which can detect arbitrary split-state tampering so long as the adversaries are unentangled, or even limited to a finite amount of pre-shared entanglement. Previously, this was only known in the limit of asymptotically large $t$.  \\

As our flagship application, we show how to augment threshold secret sharing schemes with similar tamper-detecting guarantees. We complement our results by establishing connections between quantum TDCs and quantum encryption schemes.

%We construct quantum TDCs in local (unentangled) operations model,  as well as a `bounded storage model' where the adversaries are limited to a finite amount of pre-shared entanglement.   \\

%Our main result is a construction of a quantum TDC in the bounded storage model, with just $3$ shares, of constant rate and inverse sub-exponential error. Previously, this was only known in the limit of asymptotically large $t$. As our flagship application, we show how to augment secret sharing schemes with similar tamper-detecting guarantees against local (unentangled) tampering. We complement our results by establishing connections between quantum TDCs and quantum encryption schemes.

\end{abstract}

\pagenumbering{arabic}
\newpage

\setcounter{tocdepth}{2}
{\small\tableofcontents}
\newpage

\section{Introduction}

The split-state tampering model, originally studied by Dziembowski, Pietrzak, and Wichs (ICS 2010) \cite{DPW10} and Liu and Lysyanskaya (Crypto 2012) \cite{LL12}, is by now a staple in information-theoretic cryptography. It corresponds to an adversarial coding scenario in which a message is encoded into a codeword divided into $t$ \textit{shares}, $c = (c_1, c_2, \cdots, c_t)$, and each part is sent to one of $t$ different, non-communicating adversaries. The goal of the adversaries is to independently corrupt the codeword, transforming $(c_1, c_2, \cdots, c_t)$ into $(f_1(c_1), f_2(c_2), \cdots , f_t(c_t))$, in order to change the underlying message. While apriori somewhat abstract, this model connects naturally to many other cryptographic settings, including secret sharing \cite{GK16}, commitment schemes \cite{GPR16}, multi-source extractors \cite{Cheraghchi2013NonmalleableCA}, etc. 

Arguably, what makes the split-state model interesting is that simple goals such as error-correction (where the message is recovered) or tamper-detection (where the receiver can additionally reject) are clearly impossible, at least with classical codewords (bits).\footnote{After all, the adversaries can always replace the ciphertext with a valid codeword of a pre-agreed message. The so-called substitution attack.} Nevertheless, \cite{DPW10,LL12} showed that a weaker, but related goal known as \textit{non-malleability} is possible, wherein the receiver is allowed to output the original message, or a completely uncorrelated message. This spurred a prolific line of work in the last decade on the design of split-state non-malleable codes with several desiderata, including smaller number of shares $t$, high coding rate, small security error, and even post-quantum security (see \cref{subsection:related}). However, only recently have coding schemes with quantum capabilities have been considered \cite{Bergamaschi2023PauliMD,ABJ22,Boddu2023SplitState}.\\

\noindent \textbf{Quantum Tamper Detection in the Split-State Model.} In this work, we consider a quantum analog of the split-state model, where the adversaries are non-communicating, nor are they allowed to share entanglement. However, the message is encoded into a quantum state entangled across the $t$ shares. We show, perhaps surprisingly, that \textit{tamper detection} is possible in this model, and arises naturally as a consequence of (classical or quantum) non-malleability \cite{Boddu2023SplitState}. Concretely, we design a family of codes defined on any $t\geq 3$ shares, which can detect arbitrary adversarial unentangled tampering. Previously, this was only known in an assymptotically large limit of $t$ \cite{Bergamaschi2023PauliMD} (Eurocrypt 2024).

At first glance, it may seem that this is only possible by giving the (entangled) decoder much more power over the (unentangled) adversaries.\footnote{To some extent, an asymmetry between the decoder and adversaries is inherent even to the classical split-state setting.} However, our conclusions are robust to a finite amount of pre-shared entanglement between the adversaries, up to a constant fraction of the blocklength $n$ in qubits of the code. This is near-optimal, in the sense that tamper-detection is impossible if the adversaries have $n$ qubits of entanglement (by the substitution attack!). Thus, we offer an alternative interpretation: The adversaries can only tamper with the message unnoticeably, if they share a macroscopic $\Omega(n)$ amount of entanglement. \\

\begin{figure}[t]
    \centering
        \includegraphics[width = 11cm]{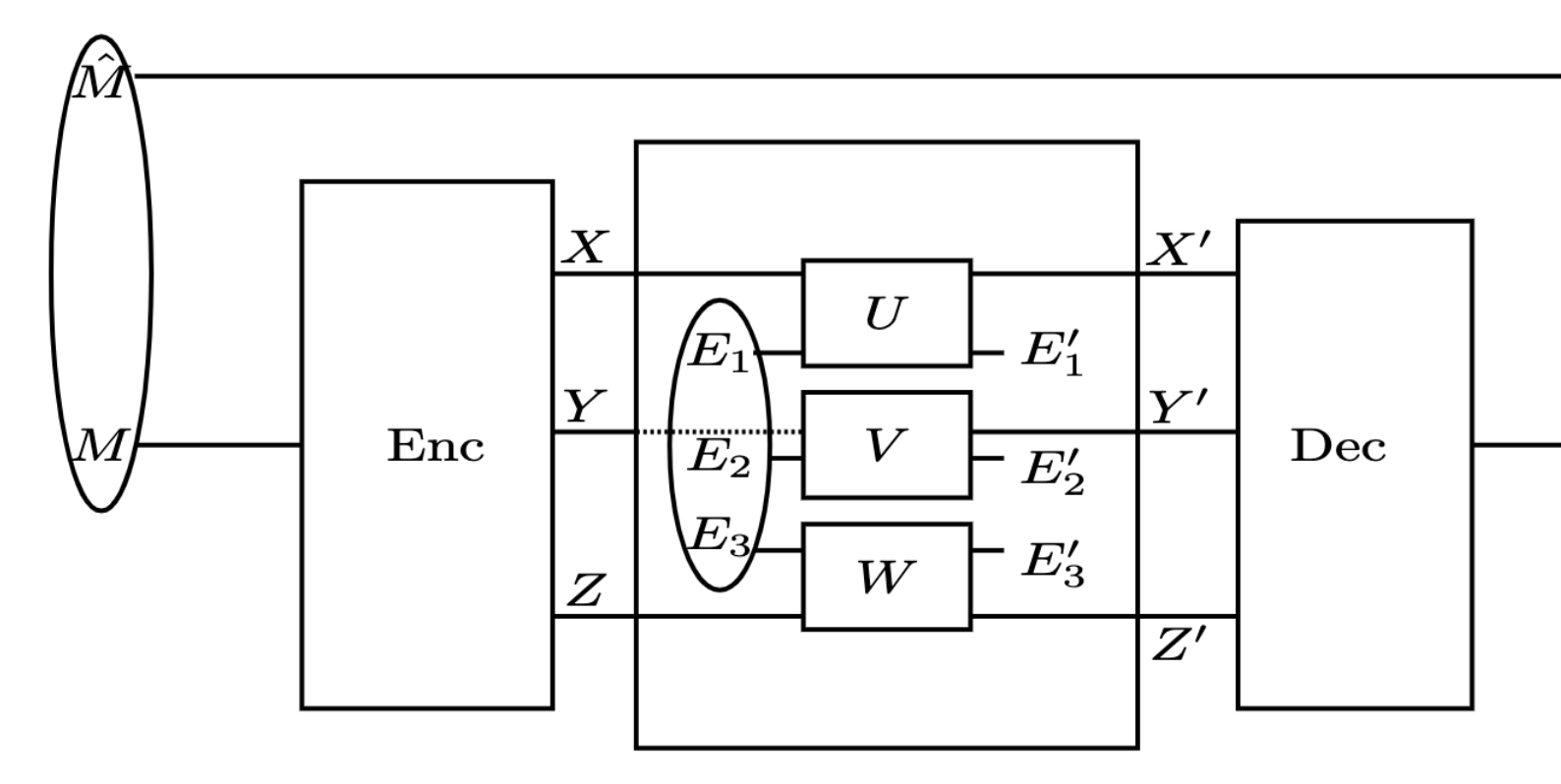}
        \caption{ The Quantum Split-State Tampering Model, $t=3$ \cite{Boddu2023SplitState}. A (possibly entangled) message $M$ is encoded, and subsequently tampered with using quantum channels $U, V, W$, jointly with an entangled state on registers $E_1, E_2, E_3$.}
        \label{fig:intro_diagram}
\end{figure}

\noindent \textbf{Applications.} The inspiration for our work lies in the study of non-locality in quantum systems, and their consequences to quantum cryptography. However, we find more concrete cryptographic motivation by complementing our results in two directions. First, we initiate the study of a new type of ``tamper-detecting" quantum secret sharing scheme, which always detects unentangled adversarial attacks. The most technical part of our work lies in constructing such a secret sharing scheme from our new 3-split TDCs, combined with techniques from a long literature of non-malleable and leakage-resilient secret-sharing schemes~\cite{GK16,GK18,ADNOP19,CKOS22}. Second, we draw connections between split-state quantum TDCs, quantum non-malleable codes, and quantum encryption schemes.

%motivation for our work lies in the study of non-locality in quantum systems and their applications to cryptography and communication. In fact, conceptually very closely related to our setting is the notion of an $\locc$ data-hiding scheme, wherein entangled states are used to hide information from unentangled adversaries. Moreover, our interest stems in that standard techniques in quantum error-correction are fundamentally unsuited towards the split-state setting.\footnote{This is since the logical operations on stabilizer codes are tensor products of Pauli operators, and consequently, unentangled adversaries which can tamper with all of the shares can change the message.} 

%We find more concrete cryptographic motivation, by complementing our results in two directions. First, we initiate the study of a new type of ``tamper-detecting" secret sharing scheme, which always detects unentangled adversarial attacks. Our flagship application (and most technical part of our work) lies in constructing such a secret sharing scheme from our new 3-split codes combined with techniques from a long literature of non-malleable and leakage-resilient secret-sharing schemes~\cite{GK16,GK18,ADNOP19,CKOS22}. Second, we prove connections between TDCs and quantum encryption schemes.\\

\subsection{Summary of contributions}
To begin, we first summarize the split-state models considered in this work. See \cref{fig:intro_diagram} for a diagram of the tampering setup. The ``amount" of pre-shared entanglement between the adversaries simply refers to the size of registers $E_1E_2E_3$.
\begin{itemize}
    \item[-] $\lo^t$: Local (unentangled) quantum operations on $t$ registers.   \\

    \item[-] $\lo_*^t$: Local quantum operations on $t$ registers with arbitrary pre-shared entanglement. \\
    
     \item[-] $\lo^t_{a}$: Local quantum operations on $t$ registers, with at most $\leq a$ qubits of pre-shared entanglement.  

\end{itemize}

While we discuss rigorous definitions shortly (\cref{subsection:overview}), the underlying intuition behind a Tamper-Detection code \cite{Jafargholi2015TamperDA} is that after the decoding operation $\Dec$, the receiver obtains a mixture (i.e a convex combination) of the original message (while preserving entanglement with $\hat{M}$) \textit{or} an abort flag $\bot$.

\subsubsection{Main Code Construction.}

Our main result is a construction showing that tamper detection is possible in the quantum split-state model $\lo^t_{a}$, even with a constant number of shares $t\geq 3$. Moreover, our construction has constant rate, and inverse sub-exponential error. 

\begin{theorem}[3-Split Tamper Detection Codes]\label{theorem:results-main}
For every $n\in \mathbb{N}$, $\gamma\in (0, \frac{1}{20})$, there exists a quantum tamper detection code secure against the 3-split-state tampering model $\lo_{\Theta(\gamma n)}^3$, of blocklength $n$, rate $\frac{1}{11} - \gamma$, and error $2^{-n^{\Omega(1)}}$.
\end{theorem}

Again, we emphasize that such a tamper-detection code is fundamentally impossible if the message were encoded into a classical codeword. Nonetheless, we achieve tamper detection (at high rate) by appealing to a modern line of work on quantum and post-quantum non-malleable codes and extractors \cite{ABJ22,Boddu2023SplitState,Bergamaschi2023PauliMD}, in combination with quantum authentication schemes. 

Unfortunately, our TDCs do not offer any security when composed in parallel. This is somewhat unsurprising: For instance, our TDCs against local adversaries $\lo^3$ involve entangled code-states. Given two TDCs, shared among the same set of parties, can the entanglement in one of them be used to break the security of the other?\footnote{Similar contraints arise in the parallel composition of (entangled) quantum data-hiding schemes against $\locc$ (unentangled) adversaries.} Nevertheless, the main application (and technical challenge) in our work shows how to augment secret sharing schemes with tamper-detection guarantees, by carefully reasoning about composition in combination with new ideas in \textit{leakage-resilient} secret sharing.

\subsubsection{Applications to Tamper-Resilient Secret Sharing.}
In a $t$-out-of-$p$ secret sharing scheme, a dealer encodes a secret into $p$ shares and distributes them among $p$ parties, satisfying the following privacy guarantee: the shares of any (unauthorized) subset of $< t$ parties don't reveal any information about the secret, while the shares of any (authorized) subset of $\geq t$ parties are enough to uniquely reconstruct it.

As we discuss shortly, our construction from \cref{theorem:results-main} is both a tamper-detection code and a 3-out-of-3 secret sharing scheme. In a natural generalization, we introduce the notion of a \textit{tamper-detecting secret sharing scheme} or $\tdss$s, strengthening the \textit{non-malleable} secret sharing schemes ($\nmss$s) introduced by Goyal and Kumar 
 (STOC 2018) \cite{GK16}. $\tdss$ simulatenously strengthens both quantum secret sharing and tamper-detection codes: In a $\tdss$s, even if all the $p$ parties tamper with each share (locally, but without entanglement), the tampered shares of \textit{an authorized set} are always enough to either recover the secret or reject, declaring adversaries have tampered with the shares.

 Here, we present a construction of a $\tdss$s with a \textit{ramp} access structure, a generalization of threshold secret sharing where the privacy and robust reconstruction thresholds differ:

\suppress{
\textit{Non-malleable} secret sharing schemes ($\nmss$s), introduced by Goyal and Kumar 
 (STOC 2018) \cite{GK16}, simultaneously strengthen both secret sharing schemes and non-malleable codes. In $\nmss$s, even if all the $p$ parties tamper with each share (independently), the tampered shares of \textit{an authorized set} are always enough to either recover the secret, or a completely unrelated secret. The design of $\nmss$s directly from split-state NMCs has seen a flurry of recent interest \cite{GK16,ADNOP19,BS19,CKOS22}, including extensions to quantum non-malleable secret sharing \cite{Boddu2023SplitState}.

 In this work, we introduce and construct a new type of (quantum) secret sharing scheme, which encodes classical messages into multipartite quantum states, offering threshold-type secret sharing guarantees, and \textit{tamper-detection} against unentangled adversaries. Our ``tamper-detecting" secret sharing schemes are in fact \textit{ramp} secret sharing schemes, a standard generalization of threshold secret sharing where the privacy and robust reconstruction thresholds differ:

 }

\begin{theorem}[Tamper-Detecting Quantum Secret Sharing]\label{theorem:results-lo-ss}
        For every $p, t$ s.t. $4\leq t\leq p-2$, there exists a $t$-out-of-$p$ quantum secret sharing scheme for $k$ bit secrets, of share size $k\cdot \poly(p)$ qubits, where even if all the shares are tampered with an $\lo^p$ channel, an authorized set of size $\geq (t+2)$ either reconstructs the secret or rejects with error $2^{-(kp)^{\Omega(1)}}$.
\end{theorem}

    Our secret sharing schemes are inspired by a well-known result by Aggarwal et al. \cite{ADNOP19} (Crypto 2019), who showed how to compile generic secret sharing schemes into a leakage-resilient and non-malleable secret sharing schemes. Informally, their approach to $\nmss$ can be viewed as ``concatenating" secret sharing schemes with 2-split NMCs. Here, we revisit their compiler by combining it with our new (3-split) tamper-detection code. As previously mentioned, proving the security of our compilers is the main challenge in this work and requires a number of new ideas, including relying on new quantum secure forms of leakage-resilient secret sharing schemes (see \cref{subsection:overview}). 

\subsubsection{Connections to Quantum Encryption Schemes.}

The notion of tamper detection studied in this work closely resembles the guarantees of quantum authentication schemes (QAS) \cite{Barnum2001AuthenticationOQ,Aharonov2008InteractivePF,BW16,garg2016new,Hayden2016TheUC}. In a QAS, a sender and receiver (who share a secret key), are able to communicate a private quantum state while detecting arbitrary man-in-the-middle attacks. A remarkable feature of \textit{quantum} authentication schemes, in contrast to their classical counterparts, is that they encrypt the message from the view of adversaries who do not know the key. Here, we similarly show that split-state TDCs must encrypt their shares:

    \begin{theorem}[Tamper-Detection Codes Encrypt their Shares]\label{theorem:results-encryption}
    The reduced density matrix of each share of a quantum tamper-detection code against $\lo^t$ with error $\epsilon$, is $4\sqrt{\epsilon}$ close to a state which doesn't depend on the message.
    \end{theorem}

    This follows from the simple but fundamental idea that adversaries that can distinguish between code-states of two known orthogonal messages, $\ket{x}, \ket{y}$, can also map between their superpositions $(\ket{x}\pm\ket{y})/\sqrt{2}$ \cite{Barnum2001AuthenticationOQ,Aaronson2020OnTH,Gunn2022CommitmentsTQ}, breaking tamper-detection. We refer the reader to \cref{section:connections-encryption} for more details on this.

\subsection{Related Work}
\label{subsection:related}

\noindent \textbf{Comparison to Prior Work}. Closely related to our results on tamper detection is work by \cite{Bergamaschi2023PauliMD}, who introduced the notion of a Pauli Manipulation Detection code, in a quantum analog to the well-studied classical algebraic manipulation detection (AMD) codes \cite{Cramer2008DetectionOA}. By combining these codes with stabilizer codes, he constructed tamper-detection codes in a ``qubit-wise" setting, a model which can be understood as the unentangled $t$-split-state model $\lo^t$ for assymptotically large $t$ (and each share is a single qubit). However, to achieve tamper detection (and high rate), \cite{Bergamaschi2023PauliMD} relied on strong quantum circuit lower bounds for stabilizer codes \cite{Anshu2020CircuitLB}. In contrast, as we will see shortly, we simply leverage Bell basis measurements. 

Even though TDCs are infeasible in entangled split-state model ($\lo^t_*$), NMCs are still feasible. \cite{ABJ22} studied \textit{quantum secure} versions of classical non-malleable codes in the $\lo^2_*$ model. \cite{Boddu2023SplitState} introduced quantum non-malleable codes, and constructed codes with inverse-polynomial rate in the $\lo^2_*$ model. \cite{BBJ23}, leveraging constructions of quantum secure non-malleable \textit{randomness encoders}, showed how to construct constant rate non-malleable codes in the $\lo^3_*$ model. \\

\suppress{

\begin{table}[h]
\centering
\begin{tabular}{||c c c c c c||} 
 \hline
 \textbf{Work by} & \textbf{Rate} & \textbf{Messages} & \textbf{Adversary} & \textbf{NM} & \textbf{TD}  \\ [1ex] 
 \hline\hline 
 \cite{ABJ22} & $\frac{1}{\poly(n)}$ & classical & $\lo^2_*$ & Yes  & No \\[0.5ex] 
 \hline
  \cite{BBJ23} & $1/3$ & {classical} & $\lo^3_*$ & Yes  & No \\[0.5ex] 
 \hline
  \cite{BBJ23} & $1/11$ & {quantum} & $\lo^3_*$ & Yes  & No \\[0.5ex] 
 \hline
 \cite{Boddu2023SplitState} & $\frac{1}{\poly(n)}$ & quantum & $\lo^2_*$  & Yes  & No \\ [1ex] 
 \hline
 \cite{Bergamaschi2023PauliMD} & $\approx 1$ & quantum & $\lo^n$  & Yes  & $\mathbf{Yes}$ \\ [1ex] 
 \hline
 \textbf{This Work} & $1/11$ & quantum & $\lo^3_{(\infty, \infty, \Omega(n))}$  & Yes  & $\mathbf{Yes}$ \\ [1ex] 
 \hline
% \textbf{This Work} & $\Omega(1)$ & quantum & $\lo^2_{}$  & Yes  & No \\ [1ex] 
 %\hline
\end{tabular}

\vspace{0.5cm}

\caption{Comparison between known explicit constructions of split-state NMCs and TDCs. Here, $n$ denotes the blocklength.}
\label{table:nmcs}
\end{table}

}

\noindent \textbf{Other Classical Non-Malleable Codes and Secret Sharing Schemes.} NMCs were originally introduced by Dziembowski, Pietrzak and Wichs \cite{DPW10}, in the context of algebraic and side-channel attacks to tamper-resilient hardware. In the past decade, the split-state and closely related tampering models for NMCs and secret sharing schemes have witnessed a flurry of work \cite{LL12,DKO13,CG14a,CG14b,Jafargholi2015TamperDA,ADL17,CGL15,li15} \cite{Li17,GK16,GK18,ADNOP19,BS19,FV19,Li19,AO20,BFOSV20,BFV21,GSZ21,AKOOS22,CKOS22,Li23}, culminating in recent explicit constructions of classical split-state NMCs of explicit constant rate $1/3$~\cite{AKOOS22} and constructions with (small) constant rate but also smaller error~\cite{Li23}.  \\

\suppress{
\noindent \textbf{Data Hiding Schemes and Leakage Resilience.} Peres and Wootters \cite{Peres1991OptimalDO} are credited with the first observations of data hiding features of bipartite quantum states. They showed that one could encode classical random variables into mixtures of bipartite product states, such that two parties restricted to $\locc$ operations would learn significantly less about the message than parties with generic entangling operations. This lead to a series of results  constructing and characterizing collections of states with different hiding attributes \cite{MassarP00,Terhal2000HidingBI,DiVincenzo2001QuantumDH,DiVincenzo2002HidingQD,Hayden2003RandomizingQS,Matthews2008DistinguishabilityOQ,Chitambar2012EverythingYA,Lami2017UltimateDH}, as well as studying connections to secret sharing \cite{Hayden2004MultipartyDH,Wang2017LocalDO,akan2023UnboundedLA}, with still with fascinating open questions \cite{winter_reflections2017}. 

Most related to our work is the recent result by \cite{akan2023UnboundedLA}, who designed cryptographic primitives (e.g. public key encryption,  digital signatures, secret sharing schemes, etc.) with quantum secret keys, for classical messages, which are leakage-resilient to unbounded $\locc$. In particular, their secret sharing schemes are robust only to non-adaptive (one-way) leakage, albeit have polynomially smaller share size compared to ours.\\

}

\noindent \textbf{Non-Malleable Quantum Encryption.} Other notions of non-malleability for quantum messages have previously been studied in the context of quantum encryption schemes. Ambainis, Bouda, and Winter \cite{ABW09}, and later Alagic and Majenz \cite{AM17}, introduced related notions of quantum non-malleable encryption schemes which (informally) can be understood as a keyed version of the NMCs we study here. \cite{AM17} showed that in the keyed setting, non-malleability actually implies encryption in a strong sense. We view the connections we raise to encryption schemes as analogs to these results in a coding-theoretic setting. 

\section{Our Techniques}
\label{subsection:overview}

We begin in \cref{subsection:overview-reductions} by presenting definitions of the various tampering guarantees we consider, and a reduction from non-malleability to quantum tamper-detection, which highlights the intuition behind our results. We proceed in \cref{subsection:overview-bounded-storage} by describing how to construct our 3-split tamper-detection codes of \cref{theorem:results-main} from classical, quantum-secure non-malleable codes. Finally, in \cref{subsection:overview-td-ss} we describe our constructions of tamper-detection secret sharing schemes (\cref{theorem:results-lo-ss}) based on our 3-split TDCs.

\subsection{Non-Malleability implies Quantum Tamper Detection}
\label{subsection:overview-reductions}

Consider the following tampering experiment. Given some message $M$, encode it into a coding scheme $c\leftarrow \Enc(M)$, tamper with it using a function $f$, and finally decode it to $M' = \Dec(f(c))$. Dziembowski, Pietrzak, and Wichs \cite{DPW10} formalized a coding-theoretic notion of non-malleability using a simulation paradigm: $M'$ can be simulated by a distribution $\mathcal{D}_f$ that depends only on $f$:
\begin{equation}
    M' \approx_\eps p_f\cdot M + (1-p_f)\cdot \mathcal{D}_f, \text{ for some }p_f\in [0, 1].
\end{equation}

\noindent In words, the distribution over $M'$ is a convex combination of the original message $M$, and a completely uncorrelated message.

Recently, \cite{Boddu2023SplitState} introduced the notion of a \textit{quantum non-malleable code}, where both the messages and the codewords can be quantum states, and showed how to construct them from certain \textit{quantum-secure} versions of classical non-malleable codes. The main difference to their classical counterparts lies in the notion of an ``uncorrelated message": which must either preserve, or completely break entanglement with any side-information (a property also known as \textit{decoupling}). 

\begin{definition}
    [Quantum Non-Malleable Codes \cite{Boddu2023SplitState}] A pair of CPTP channels $(\Enc, \Dec)$ is a \emph{quantum non-malleable code} against a tampering channel $\Lambda$ with error $\epsilon$ if the ``effective channel" satisfies
    \begin{equation}
        \forall \psi\in \cD(\mathcal{H}_M\otimes \mathcal{H}_{\hat{M}}): \big(\Dec\circ\Lambda \circ \Enc\otimes \mathbb{I}_{\hat{M}}\big)(\psi_{M\hat{M}}) \approx_\epsilon p\cdot \psi_{M\hat{M}}+(1-p)\cdot \sigma\otimes \psi_{\hat{M}},
    \end{equation}

    \noindent where $p\in [0, 1]$ and $\sigma \in \cD(\mathcal{H}_M)$ depend only on $\Lambda$, and are independent of $\psi$.
\end{definition}

 \noindent If instead of the arbitrary state $\sigma$, the decoder rejected by outputting a fixed flag $\sigma = \bot$, then we would refer to $(\Enc, \Dec)$ as a \emph{tamper-detection code} for $\Lambda$. \\

\noindent \textbf{Tamper Detection in the unentangled Split-State Model.} The starting point of our work lies in the observation that the decoupling nature of non-malleability actually implies strong tamper-detecting guarantees in the split-state model. Here, we sketch a reduction which converts any quantum non-malleable code $(\Enc^t, \Dec^t)$ for the unentangled $t$-split model, $\lo^t$, into a tamper-detection code for the unentangled $(t+1)$-split model $\lo^{t+1}$.

Given a message state $\psi$ and some integer parameter $\lambda$, the channel $\Enc_\lambda(\psi)$ simply prepares $\lambda$ EPR pairs, $\Phi^{\otimes \lambda}_{E\hat{E}}$, and encodes $\psi$ and half of the EPR pairs into $\Enc^t$ (see \cref{fig:tsplitstatenoentanglment}):
\begin{equation}
    \Enc_\lambda (\psi) = [\Enc^t_{ME}\otimes \mathbb{I}_{\hat{E}}](\psi\otimes \Phi^{\otimes \lambda}_{E\hat{E}}).
\end{equation}

\noindent The remaining EPR halves $\hat{E}$ are placed in the additional $(t+1)$st share.

\begin{figure}
    \centering
    \includegraphics[width = 9cm]{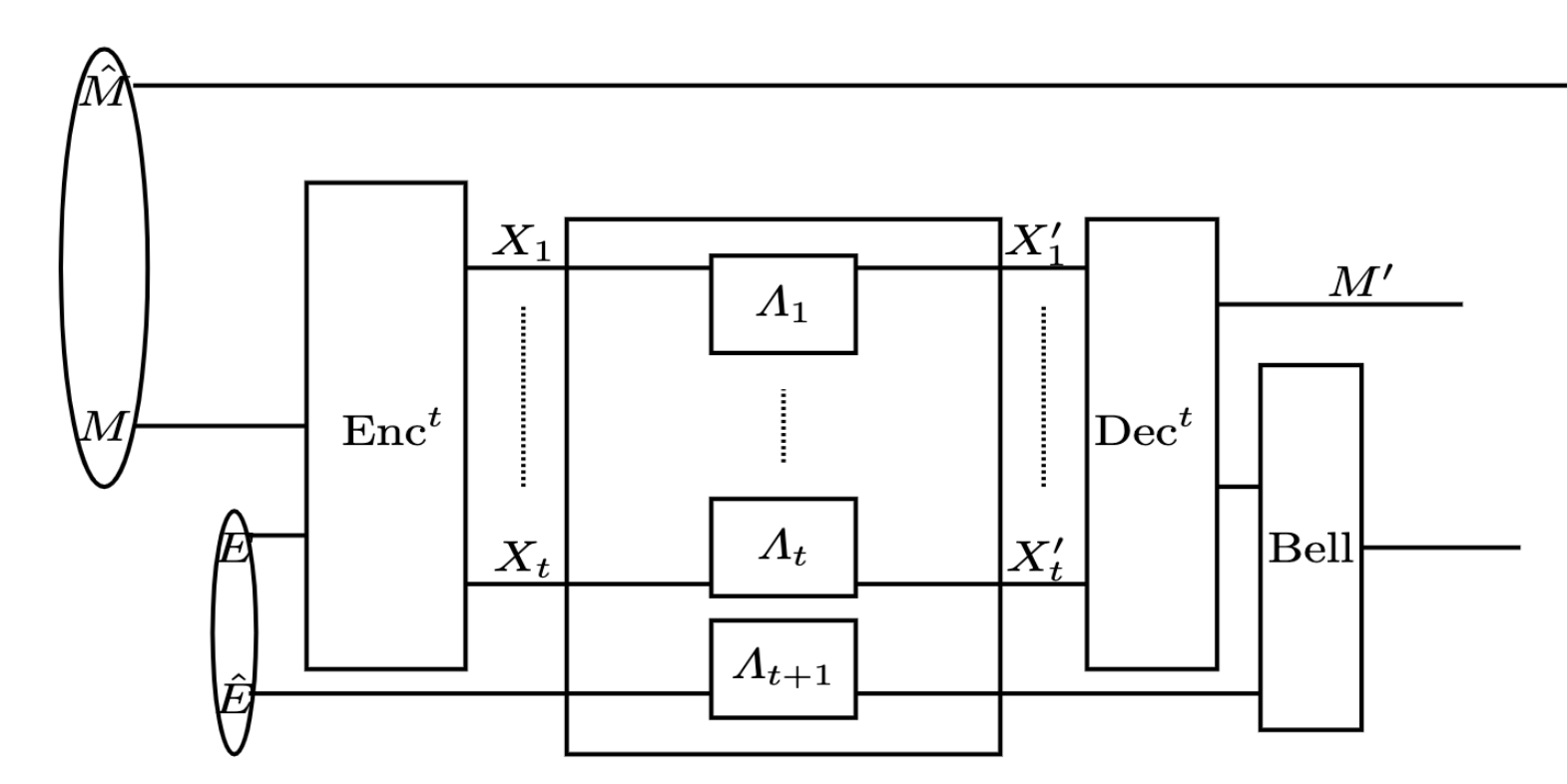}
    \caption{Quantum TDC against $\lo^{t+1}$.}\label{fig:tsplitstatenoentanglment}
\end{figure}

To decode (after tampering), we first decode the non-malleable code on the first $t$ shares using $\Dec^t$, and then measure the EPR registers in the Bell basis. Since the adversaries in $\lo^{t+1}$ are unentangled, the decoder $\Dec^t$ always either receives $\psi$ back, or a \textit{product state} with the $(t+1)$st share. However, the Bell basis measurement is guaranteed to reject product states, and therefore can be used to verify (and reject) if the adversaries have tampered with their shares. 

Unfortunately, our reduction only takes us so far. The reason is two-fold: First, to construct constant rate $3$-split TDCs, we require $2$-split quantum NMCs of constant rate, which remains an open problem \cite{Boddu2023SplitState}. More importantly, it is somewhat unclear whether this compiler can even be used to construction codes in the bounded storage model, or secret sharing schemes.

\subsection{Tamper-Detection Codes in the Bounded Storage Model}
\label{subsection:overview-bounded-storage}

Instead, to construct our 3-split TDCs in the `bounded storage' model $\lo_a^3$ (with high rate), we appeal to recent constructions of split-state quantum non-malleable codes by \cite{Boddu2023SplitState,BBJ23} against $\lo_*$. They observed that a certain \textit{augmented} property of classical non-malleable codes could be used to achieve non-malleability for quantum states. Here we briefly overview their approach and highlight our modifications.  \\

\vspace{-10pt}
\begin{figure}[h]
\noindent\fbox{%
    \parbox{\textwidth}{%
      \noindent \textbf{The 3-Split Quantum Non-Malleable Codes by} \cite{Boddu2023SplitState,BBJ23} combined:
\begin{enumerate}[(a)]
    \item $(\Enc_\NM, \Dec_\NM)$: A quantum secure \textit{augmented} 2-split non-malleable code.\\
    \item A family of unitary 2-designs $\{C_R\}_{R\in K}$\footnotemark. 
\end{enumerate}
To encode a message $\psi$, their encoding channel 

\begin{enumerate}
    \item Samples a uniformly random key $R$ and encodes it into the non-malleable code, $(X, Y) = \Enc_{\NM}(R)$\\
    \item Authenticates $\psi$: $C_R\psi C_R^\dagger$ (on some register $Q$), using the family of 2-designs keyed by $R$.
    \end{enumerate}
        2 of the 3 parties hold the classical registers $X, Y$, and the last holds $Q$: $(Q, X, Y)$.

}
}
    \label{fig:qnmcs}
\end{figure}
\footnotetext{Here we point out that any Quantum Authentication Scheme (QAS) would suffice \cite{Barnum2001AuthenticationOQ,Hayden2016TheUC,BW16,garg2016new}, as would a non-malleable \textit{randomness encoder} in place of a non-malleable code \cite{KOS18,BBJ23}. For background on these components, refer to \cref{section:prelim}.}
\vspace{-10pt}

The non-malleability of this scheme against local adversaries with shared entanglement, $\lo^3_*$, relies crucially on the \textit{augmented} security of the quantum-proof non-malleable code. Informally, suppose the two shares $X, Y$ are tampered to $X', Y'$ using shared entanglement, and $W$ denotes any quantum side-information held by one of the adversaries. By definition, the resulting joint distribution (state) over the original uniformly random key, the recovered key $R'=\Dec_\NM(X',Y')$, \textit{and the side-information}, is a convex combination of the original key and a fixed, independent key (See \cref{lem:qnmcodesfromnmext}):
\begin{equation}
     R, R', W \approx p\cdot U_{R=R'}\otimes W + (1-p) \cdot U_R\otimes (R', W),
\end{equation}
\noindent where $p\in [0, 1]$ is independent of $R$, and $U_{R=R'}$ denotes the recovered key is the same as the original and uniformly random. 

It turns out that the non-malleability of the secret key, suffices to achieve non-malleability for the quantum message. In the first case above, the key is recovered and is independent of the side information. By the properties of the two-design unitaries, the message is authenticated \cite{Aharonov2008InteractivePF}, and thereby we either recover it, or reject. In the second case, the key is lost to the decoder. Thereby, from the decoder's perspective, it receives a message scrambled by a random unitary $C_R$, and therefore `looks' encrypted and independent of the original message.

Inspired by our reductions, we differ from the outline above in only one detail: we encrypt the message state $\psi$ \textit{together} with $\lambda$ halves of EPR pairs. The remaining EPR halves are handed to one of the other two adversaries, together with one of the NMC splits. In detail,

\begin{figure}[h]
\noindent\fbox{%
    \parbox{\textwidth}{%
      \noindent \textbf{Our Construction of 3-Split Tamper-Detection Codes} is defined by an encoding channel $\Enc_\td^\lambda$, parametrized by an integer $\lambda,$ which:

      \begin{enumerate}
    \item Prepares $\lambda$ EPR pairs, on quantum registers $E, \hat{E}$.\\
    \item Samples a uniformly random key $R$ and encodes it into the 2-split non-malleable code $(X, Y)=\Enc_\NM(R)$\\
    \item Encrypts $E$ together with the message $\psi$, using a family of 2-design unitaries keyed by $R$. Let the resulting register be $Q$.
\end{enumerate}
\noindent The final 3 shares are $(Q, Y\hat{E},X)$.
}
}
    \label{fig:qtdcs}
\end{figure}

By combining the properties of the quantum secure non-malleable code with a Bell basis measurement (on states of low Schmidt rank), we prove that the construction above is a secure tamper-detection code against an even stronger `bounded storage' adversarial model, which we refer to as $\lo_{(a, *, *)}^3$. In $\lo_{(a, *, *)}^3$, the 3 adversaries hold a pre-shared entangled state 
 on registers $A_1, A_2, A_3$, where the 1st party (the one holding the register $Q$) is limited to $|A_1|\leq a$ qubits, but the other two may be unbounded. As we will see shortly, this model plays a particularly important role in our constructions of tamper-detecting secret sharing schemes.

\subsection{Tamper-Detecting Secret Sharing Schemes against \texorpdfstring{$\lo$}{Lo}}
\label{subsection:overview-td-ss}

We dedicate this section to an overview of our construction of ramp secret sharing schemes (for classical messages) that detect local tampering. At a very high level, our secret sharing schemes are based on the concatenation of classical secret sharing schemes, with our tamper-detection codes. The key challenge lies in reasoning about the security of our tamper-detection codes when composed in parallel, which we address by combining them with extra leakage-resilient properties of the underlying classical scheme.

The outline of our construction lies in a compiler by \cite{ADNOP19}. Their compiler takes secret sharing schemes for arbitrary access structures as input, and produces either leakage-resilient or non-malleable secret sharing schemes for the same access structure. We begin with a brief description:

\vspace{-5pt}
\begin{figure}[H]
\noindent\fbox{%
    \parbox{\textwidth}{%
    \vspace{5pt}

    \noindent \textbf{The Compiler by} \cite{ADNOP19}. To encode a message $m$, they begin by 

    \begin{enumerate}
        \item  Encoding it into the base secret sharing scheme $(M_1, \cdots, M_p)\leftarrow \share(m)$.\\
        \item Then, each share $M_i$ is encoded into some (randomized) \textit{bipartite} coding scheme $(\Enc, \Dec)$, to obtain two shares $L_i$ and $R_i$.\\
        \item For each $i\in [n]$, the new compiled shares amount to giving $L_i$ to the $i$-th party, and a copy of $R_i$ to every other party. At the end of this procedure, the $i$-th party has a compiled share $S_i$, where 
        \begin{equation}
    S_i = (R_1, \cdots, R_{i-1}, L_i, R_{i+1}, \cdots, R_p)
\end{equation}
    \end{enumerate}

}} \label{fig:compiler-adn19}
\end{figure}

\vspace{-10pt}

The natural decoding algorithm takes all pairs of shares $(L_i, R_i)$ within an authorized set of parties, recovers $M_i\leftarrow \Dec(L_i, R_i)$, and then uses the reconstruction algorithm for the secret sharing scheme to recover the message $m$. By instantiating this compiler using different choices of the underlying bipartite coding scheme (strong seeded extractors or non-malleable extractors), they construct secret sharing schemes satisfying different properties (leakage resilience or non-malleability).

\begin{figure}[h]
\noindent\fbox{%
    \parbox{\textwidth}{
    \vspace{5pt}
    \noindent \textbf{Our Construction} attempts a tripartite version of the compiler by \cite{ADNOP19}. We combine
\begin{enumerate}[(a)]
    \item $(\Enc^\lambda_\td, \Dec_\td^\lambda)$, our $3$-split-state TDC in the bounded storage model.\\
    
    \item $(\share, \rec)$, a classical $t$-out-of-$p$ secret sharing scheme, which is additionally leakage-resilient against a certain local quantum leakage model.
\end{enumerate}

\noindent On input a classical message $m$, we similarly begin by 

\begin{enumerate}
    \item Encoding it into the base secret sharing scheme $(M_1, \cdots, M_p)\leftarrow \share(m)$.\\
    \item Then, every triplet of parties $a<b<c \in [n]$ encodes a copy of their shares $M_a, M_b, M_c$ into a certain 3-out-of-3 ``gadget", $\Enc_\triangle(M_a, M_b,M_c)$, which is then re-distributed between the three parties.
\end{enumerate}
}}
\caption{A ``Tamper-Detecting" Secret Sharing Scheme}
   \label{fig:tdss-compiler1}
\end{figure}

We refer the reader to \cref{fig:tdss-compiler1} for a description of our compiler. The most technical part of our work is to find a choice of $\Enc_\triangle$ which guarantees a local version of tamper detection for the shares in the gadget, and can be combined with the leakage-resilience of the underlying secret sharing scheme to ensure tamper detection of the actual message. We motivate this challenge and our choice of $\Enc_\triangle$ below.\\

\noindent \textbf{Weak Tamper Detection and the Selective Bot Problem.} The most natural tripartite candidate for $\Enc_\triangle$ would directly be our tamper-detection code. That is, $M_a, M_b, M_c$ are jointly encoded into $\Enc_\td(M_a||M_b||M_c)$, resulting in 3 (entangled) quantum registers $A, B, C$, which are handed to parties $a, b, c$ respectively. 

One can readily verify that this construction already offers a certain ``weak" tamper detection guarantee. If any three parties $a, b, c$, were to tamper with their shares, then by the properties of $\Enc_\td$\footnote{Technically, we additionally require $\Enc_\td$ to satisfy a strong form of 3-out-of-3 secret sharing, where any two shares are maximally mixed.}, one can show that the \textit{marginal} distribution over the recovered shares $M_a', M_b', M_c'$ is near a convex combination of the original shares $M_a, M_b, M_c$, or rejection. Similarly, given tampered shares of any authorized subset of parties $T\subset [p]$, after decoding all the copies of $\Enc_\td$ in $T$, by a standard union bound we either reject or recover all the honest shares of $T$. By the correctness of the secret sharing scheme, this implies we either recover the message, or reject, $\Dec\circ \Lambda\circ \Enc(m) \in \{m, \bot\}$ with high probability. 

The reason this code only offers ``weak" tamper detection is that \textit{the probability the decoder rejects may depend on the message} (See \cref{def:weak-td} for a formal definition). That is, whether the gadgets in $T$ reject may correlate with all the shares of $T$, and therefore the message itself, violating tamper detection/non-malleability. This ``selective bot" problem is a fundamental challenge in the design of non-malleable secret sharing schemes \cite{GK16}, and we combine two new ideas to overcome it.\\

\noindent \textbf{A Multipartite Encoder with staggered Entanglement}. One might hope that instead of simply relying on the privacy of the underlying secret sharing scheme, one could leverage stronger properties to simplify the tampering process (to some other effective tampering on the classical shares). While this is indeed roughly the approach we undertake, the main obstacle lies in that to implement $\Enc_\td(M_a||M_b||M_c)$, one of the parties $i\in \{a, b, c\}$ will contain a register with an encryption of all three shares. Any simulation\footnote{By simulation, we mean that the joint distribution over the recovered shares is the same as that of the quantum process.} of $\Enc_\td$ acting directly on the shares would then seem to require communication, which is too strong to impose on the classical secret sharing scheme. 

Our first idea is to instead use independent tamper-detection codes $\Enc_\td^{\lambda_i}(M_i)$ for each $i\in \{a, b, c\}$, using different numbers $\lambda_i$ of ``trap" EPR pairs $(E_i\hat{E}_i)$. This new gadget $\Enc_\triangle$ can be implemented using only \textit{local knowledge} of $m_i$ and pre-shared entanglement between the parties, but no communication.

\vspace{-10pt}

\begin{figure}[h]
\noindent\fbox{%
    \parbox{\textwidth}{%
    \vspace{5pt}

   $\Enc_\triangle(M_a, M_b, M_c)$: Each party $i\in \{a, b, c\}$,
    \begin{enumerate}
        \item Encodes their share into $\Enc_\td^{\lambda_i}(M_i)$, resulting in registers $(Q_i,Y_i\hat{E}_i, X_i)$.\\
        \item Party $i$ keeps the register $Q_i$, hands $X_i$ to party $i-1$, and both $Y_i, \hat{E}_i$ to party $i+1$.
    \end{enumerate}

    The decoding channel $\Dec_\triangle$ simply runs $\Dec_\td^{\lambda_i}$ on the 3 TDCs, and rejects if anyone of them does.
    }
}
    \label{fig:gadget}
\end{figure}

\vspace{-10pt}

We find it illustrative to imagine parties $a, b, c$ in a triangle, where each party $i$ hands the EPR halves $\hat{E}_i$ to the party immediately to their left in the ordering $c>b>a$. As we will see shortly, the locality in this scheme plays an important role in addressing the selective bot problem, however, it also comes at a cost. Suppose we tamper with $\Enc_\triangle(M_a, M_b, M_c)$ using a channel in $\lo^3$. Since $\Enc_\triangle$ is comprised of 3 copies of $\Enc_\td$ ``in parallel", we apriori have no global guarantee on all 3 of them. 

Nevertheless, we prove that we can recover \textit{two out of the three shares} by picking an increasing sequence of EPR pairs, say, $\lambda_a, \lambda_b, \lambda_c = \lambda, 2\cdot\lambda, 3\cdot \lambda$ (where $a<b<c$). This follows by carefully inspecting the marginal distribution on any single recovered share $M_i'$. We observe that the effective tampering channel on the $i$th tamper-detection code can be simulated by a channel in the bounded storage model $\lo_{(\lambda_{i-1}, *, *)}^3$: where the other two copies of $\Enc_\td$ act as pre-shared entanglement which can aid the adversaries.\footnote{Recall from \cref{subsection:overview-bounded-storage} that $\lo_{(e, *, *)}^3$ corresponds to the set of 3-split-state adversaries aided by a pre-shared entangled state where only one of the parties is limited to $\leq e$ qubits, but the other two may have unbounded sized registers.} Since $\Enc_\td^{\lambda_i}$ is secure against $\lo_{(\lambda_{i-1}, *, *)}$ whenever $\lambda_i>\lambda_{i-1}$, we acquire weak tamper detection for both $M_b$ and $M_c$ (but not $M_a$).\footnote{This idea resembles \cite{GK16} approach to non-malleable secret sharing by combining secret sharing schemes with different privacy parameters to acquire one-sided independence. } 

We point out that this is the main technical reason for which we relax our model to ramp secret sharing (instead of threshold): our decoder will have essentially no guarantees on the smallest share of $T$. Nevertheless, by the same union bound, all the other shares of $T$ except for the smallest are ``weak" tamper-detected. If we assume $|T|\geq t+1$, we acquire weak tamper detection for the message as well. \\

\noindent \textbf{Correlations between Triangles and Leakage-Resilient Secret Sharing.} To conclude our construction, we show how to leverage leakage-resilient secret sharing to ensure the probability the decoder rejects doesn't depend on the message. Together with the weak tamper detection guarantee inherited by the outer tamper-detection codes, this implies the secret sharing scheme is actually (strongly) tamper detecting.

At a high level, we devise a \textit{bipartite} decoder which captures the following scenario. Suppose we partition an authorized subset $T = A\cup B$ into disjoint, unauthorized sets of size $\geq 3$. After independent tampering on the shares of $T$, suppose we use $\Dec_\triangle$ to decode all the triplets $a, b, c$ contained entirely in $A$ or $B$, but we completely ignore the gadgets which cross between $A$ and $B$. What is the probability any of the gadgets rejects? Can it correlate with the message $m$?

It turns out this scenario is cleanly captured by a model of quantum leakage-resilience between $A$ and $B$ on the underlying classical secret sharing scheme. Note that the amount of shared entanglement between $A$, $B$ comprises $\mu = O(p^3\cdot \lambda)$ qubits, due to all the EPR pairs in the triplets of $A\cup B$. Using the locality of the gadget $\Enc_\triangle$, we claim that the joint measurement performed to decide whether $A$ or $B$ reject in our compiler (a binary outcome POVM), can be simulated directly on the underlying classical secret sharing scheme by a $\mu$ qubit ``leakage channel" from $A$ to $B$.

In slightly more detail, this leakage channel comprises the following experiment on the classical secret sharing scheme. The parties in $A$ are allowed to jointly process their classical shares $M_A$ into a (small) $\mu \ll |M_A|$ qubit quantum state $\sigma_{M_A}$, which is then sent to the parties in $B$. In turn, the parties in $B$ are allowed to jointly use their own shares $M_B$, together with the leakage $\sigma_{M_A}$, to attempt to guess the message. If $B$ can't distinguish between any two messages $m^0, m^1$ with bias more than $\epsilon_{\lr}$, that is, 
\begin{equation}
    M_B^0, \sigma_{M_A^0} \approx_{\epsilon_{\lr}}  M_B^1, \sigma_{M_A^1} \text{ where }M_A^b, M_B^b = \big(\share_\lrss(m^b)\big)_{A\cup B},
\end{equation}

\noindent then we say the secret sharing scheme is leakage-resilient (to this bipartite quantum leakage model) with error  $\epsilon_{\lr}$ (see \cref{def:lrss}). In the context of our compiler, if we choose the underlying secret sharing scheme to be such an $\lrss$, we conclude that the probability $A, B$ reject is (roughly) $\epsilon_{\lr}$ close to independent of the message. Together with the weak tamper detection guarantee, we all but conclude the proof.

It only remains to construct such classical leakage-resilient secret sharing schemes. Unfortunately, this bipartite quantum leakage model is slightly non-standard, and to the best of our knowledge, there are no ``off-the-shelf" solutions in the literature. This is due in part to the joint nature of the leakage model, but also to its security against quantum adversaries. In \Cref{section:lrss}, we show how to construct leakage-resilient secret sharing schemes in this model following simple modifications to a classical $\lrss$ compiler by \cite{CKOS22}.

\subsection{Discussion}

We dedicate this section to a discussion on the directions and open questions raised in our work.  \\

%\noindent \textbf{Is tamper detection possible against $\locc$?} Our techniques fall short of constructing tamper-detection codes against adversaries who are allowed classical communication, even though they are limited to an unentangled form of tampering. This is by and large due to the lack of composability of (entangled) $\locc$ data hiding protocols, together with the fact that we require entanglement to achieve tamper detection. \\

\noindent \textbf{Optimal tamper detection in the bounded storage model.} Our TDCs against $\lo^3_a$ achieve tamper detection even against adversaries whose amount of pre-shared entanglement is at most a constant fraction of the blocklength $n$. Can one design codes robust against $(1-o(1))\cdot n$ qubits of pre-shared entanglement, approaching the non-malleability threshold?\\

\noindent \textbf{Is tamper detection against $\lo$ possible in the 2-split model?} Based on our constructions for $\lo$, there are a number of ``natural" 2-split candidates which combine classical 2-split non-malleable codes with entanglement across the cut. Nevertheless, they all seem to fall short of tamper detection, due in part to adversaries which correlate their attack on the entangled state with their attack \textit{on both of the shares} of the non-malleable code. \\
%Related bottlenecks arise in the $\locc$ case.\\

\noindent \textbf{Stronger secret sharing schemes.} Can one extend our results to quantum secret sharing schemes, where the messages are quantum states? Extending \cite{GK16} or \cite{BS19} approach to tamper-detecting secret sharing schemes for quantum messages seems to require completely new ideas, including at the very least a definition of leakage resilient quantum secret sharing. \\

\noindent \textbf{Paper Organization} We organize the rest of this work as follows.

\begin{itemize}
    \item [--] In \Cref{section:prelim}, we present the necessary background in quantum information, and define the quantum authentication schemes we require in our work.
    \item [--] In \cref{section:coding-definitions}, we present the definitions of the split-state tampering models considered in this work, and in \Cref{section:bounded-storage} we present the explicit construction of tamper-detection codes in the bounded storage model of \cref{theorem:results-main}. 
    \item [--] In \cref{section:prelim-ss}, we present definitions of variants of secret sharing schemes we require. Subsequently, in \Cref{section:tamper-detecting-secret-sharing}, we present the construction of ramp secret sharing schemes which detect unentangled tampering of \cref{theorem:results-lo-ss}. 
    \item [--] Finally, in \Cref{section:connections-encryption}, we discuss connections between tamper-resilient quantum codes and quantum encryption schemes.
\end{itemize}

\noindent Additionally,

\begin{itemize}
    \item [--] In \Cref{section:cliffords}, we present relevant facts and background on Pauli and Clifford operators as well as unitary 2-designs.
    
    \item [--] In \Cref{section:lrss}, we present our construction of secret sharing schemes resilient to quantum leakage, based on a compiler by \cite{CKOS22}, which we use in \Cref{section:tamper-detecting-secret-sharing}.
\end{itemize}

\section*{Acknowledgements}

A preliminary version of this work appeared on arxiv at the reference \cite{bergamaschi2023splitstate}. We thank João Ribeiro, Venkat Guruswami, Rahul Jain, Marshall Ball, Umesh Vazirani and Nathan Ju for illuminating discussions on classical and quantum non-malleability and tamper detection.

%We thank João Ribeiro for insights on his compilers for secret sharing schemes, Venkat Guruswami and Rahul Jain for discussions on the capacity of split-state classical and quantum non-malleable codes, and Marshall Ball for conversations on augmented quantum non-malleable codes during the early stages of this work. Finally, Umesh Vazirani and Nathan Ju for comments on the manuscript and presentation.

TB acknowledges support by the National Science Foundation Graduate Research Fellowship under Grant No. DGE 2146752.

\section{Preliminaries}
\label{section:prelim}

In \cref{subsection:qi}, we cover some important basic prerequisites from quantum information theory. We refer the reader familiar with quantum states and channels directly to \Cref{subsection:prelim-2-designs}. We define Clifford authentication schemes in \cref{subsection:prelim-2-designs}. 

\subsection{Quantum and Classical Information Theory}\label{subsection:qi}

The notation and major part of this subsection is Verbatim taken from~\cite{Boddu2023SplitState}.

\subsubsection{Basic General Notation}

 We denote sets by uppercase calligraphic letters such as $\X$ and use uppercase roman letters such as $X$ and $Y$ for both random variables and quantum registers. We denote the uniform distribution over $\{0,1\}^d$ by $U_d$.
For a {\em random variable} $X \in \X$, we use $X$ to denote both the random variable and its distribution, whenever it is clear from context. 

\subsubsection{Quantum States and Registers}

Consider a finite-dimensional Hilbert space $\cH$ endowed with an inner-product $\langle \cdot, \cdot \rangle$ (we only consider finite-dimensional Hilbert spaces). A quantum state (or a density matrix or a state) is a positive semi-definite operator on $\cH$ with trace $1$. 
It is called {\em pure} if and only if its rank is $1$. Let $\ket{\psi}$ be a unit vector on $\cH$, that is $\langle \psi,\psi \rangle=1$.  
With some abuse of notation, we let $\psi$ represent both the state and the density matrix $\ketbra{\psi}$. 
 
A {\em quantum register} $A$ is associated with some Hilbert space $\cH_A$. Define $\vert A \vert := \log\left(\dim(\cH_A)\right)$ the size (in qubits) of $A$. The identity operator on $\cH_A$ is denoted $\id_A$. $U_A$ denotes the maximally mixed state in $\cH_A$. For a sequence of registers $A_1,\dots,A_n$ and a set $T\subseteq[n]$, we define the projection according to $T$ as $A_T=(A_i)_{i\in T}$.
Let $\mathcal{L}(\cH_A)$ represent the set of all linear operators on the Hilbert space $\cH_A$, and $\mathcal{D}(\cH_A)$ the set of all quantum states on the Hilbert space $\cH_A$. The state $\rho$ with subscript $A$ indicates that $\rho_A \in \mathcal{D}(\cH_A)$. 

Composition of two registers $A$ and $B$, denoted $AB$, is associated with the Hilbert space $\cH_A \otimes \cH_B$.  For two quantum states $\rho\in \mathcal{D}(\cH_A)$ and $\sigma\in \mathcal{D}(\cH_B)$, $\rho\otimes\sigma \in \mathcal{D}(\cH_{AB})$ represents the tensor product of $\rho$ and $\sigma$. Let $\rho_{AB} \in \mathcal{D}(\cH_{AB})$. Define
$$ \rho_{B} \defeq \tr_{A}{\rho_{AB}} \defeq \sum_i (\bra{i} \otimes \id_{B})
\rho_{AB} (\ket{i} \otimes \id_{B}) , $$
where $\{\ket{i}\}_i$ is an orthonormal basis for the Hilbert space $\cH_A$.
The state $\rho_B\in \mathcal{D}(\cH_B)$ is the marginal or partial trace of $\rho_{AB}$ on $B$. Given $\rho_A\in\mathcal{D}(\cH_A)$, a {\em purification} of $\rho_A$ is a pure state $\rho_{AB}\in \mathcal{D}(\cH_{AB})$ such that $\tr_{B}{\rho_{AB}}=\rho_A$. 

\subsubsection{Quantum Measurements, Channels and Instruments}

A {\em Hermitian} operator $H:\cH_A \rightarrow \cH_A$ is such that $H=H^{\dagger}$. A projector $\Pi \in  \mathcal{L}(\cH_A)$ is a Hermitian operator such that $\Pi^2=\Pi$. A {\em unitary} operator $V_A:\cH_A \rightarrow \cH_A$ is such that $V_A^{\dagger}V_A = V_A V_A^{\dagger} = \id_A$. The set of all unitary operators on $\cH_A$ is  denoted by $\mathcal{U}(\cH_A)$. A positive operator-valued measure ({\em POVM}) $\{M_i\}_i$ is a collection of Hermitian operators where $0 \le M_i \le \id$ and $\sum_i M_i = \id$. Each $M_i$ is associated to a measurement outcome $i$, where outcome $i$ occurs with probability $\tr[M_i\rho]$ when measuring state $\rho$. We use the shorthand $\bar{M} \defeq \id - M$, and $M$ to represent $M \otimes \id$, whenever $\id$ is clear from context. 

A quantum {map} $\cE: \mathcal{L}(\cH_A)\rightarrow \mathcal{L}(\cH_B)$ is a completely positive and trace preserving (CPTP) linear map. A CPTP map $\cE$ is described by the Kraus operators $\{ K_i : \cH_A\rightarrow \cH_B \}_i$ such that $\cE(\rho) = \sum_i K_i \rho K^\dagger_i$ but $\sum_i K^\dagger_i K_i =\id_A$. A (trace non-increasing) CP map $\cE$ is similarly described by the Kraus operators $\{ K_i\}_i$ but $\sum_i K^\dagger_i K_i  < \id_A.$  Finally, we use the notation $\mathcal{N}\circ \mathcal{M}$ to denote the composition of maps $\mathcal{N}, \mathcal{M}$.

\subsubsection{Norms, Distances, and Divergences}

This section collects definitions of some important quantum information-theoretic quantities and related useful properties. 

\begin{definition}[Schatten $p$-norm]
    For $p\geq 1$ and a matrix $A$, the \emph{Schatten $p$-norm} of $A$, denoted by $\|A\|_p$, is defined as $\| A \|_p  \defeq (\tr(A^\dagger A)^{\frac{p}{2}})^{\frac{1}{p}}.$
\end{definition}

\begin{definition}[Trace distance]
    The \emph{trace distance} between two states $\rho$ and $\sigma$ is given by $\|\rho-\sigma\|_1$. We write $\rho\approx_\eps \sigma$ if $\|\rho-\sigma\|_1\leq \eps$.
\end{definition}

\begin{fact}[Data-processing inequality, Proposition 3.8 \cite{Tomamichel2015QuantumIP}]
\label{fact:data}
Let $\rho, \sigma$ be states and $\cE$ be a (trace non-increasing) CP map. Then, $\Vert  \cE(\rho)  - \cE(\sigma)\Vert_1  \le \Vert \rho  - \sigma \Vert_1$, which is saturated whenever $\cE$ is a CPTP map corresponding to an isometry.
\end{fact}

\subsubsection{Quantum-Classical, Separable, and Low Schmidt Rank States}

This section collects definitions of certain registers and operations on them.

\begin{definition}\label{def:classicalinpurestate}
Let $\X$ be a set. 
A {\em classical-quantum} (c-q) state $\rho_{XE}$ is of the form \[ \rho_{XE} =  \sum_{x \in \X}  p(x)\ket{x}\bra{x} \otimes \rho^x_E , \] where ${\rho^x_E}$ are density matrices. Whenever it is clear from context, we identify the random variable $X$ with the register $X$ via $\mathbb{P}[X=x]=p(x)$. 
%Let $\rho_{XEA}$ be a pure state. We call $X$ a \emph{classical register} in $\rho_{XEA}$ if $\rho_{XE}$ (or $\rho_{XA}$) is a c-q state. 
\end{definition}

\begin{definition}[Conditioning] \label{def:conditioning}
Let  
\[ \rho_{XE} =  \sum_{x \in \{0,1\}^n}  p(x)\ket{x}\bra{x} \otimes \rho^x_E , \]
be a c-q state. For an event $\mathcal{S} \subseteq \{0,1\}^n$, define  $$\mathbb{P}[\mathcal{S}]_\rho \defeq  \sum_{x \in \mathcal{S}} p(x) \quad \textrm{and} \quad (\rho|X\in \mathcal{S})\defeq \frac{1}{\mathbb{P}[\mathcal{S}]_\rho} \sum_{x \in \mathcal{S}} p(x)\ket{x}\bra{x} \otimes \rho^x_E.$$
We sometimes shorthand $(\rho|X\in \mathcal{S})$ as $(\rho|\mathcal{S})$ when the register $X$ is clear from context. 
\end{definition}

Entangled bi- or multi-partite quantum states are those which are not separable. 

\begin{definition}\label{def:separable}
    A bipartite mixed state $\rho$ is said to be separable if it can be written as a convex combination over product states, 
    \begin{equation}
        \rho = \sum_i p_i\cdot \sigma_i\otimes \tau_i
    \end{equation}
    \noindent Where $\{\sigma_i,  \tau_i\}_i$ are collections of density matrices, and $p_i\in \mathbb{R^+}$ s.t. $\sum_ip_i=1$.
\end{definition}

Entanglement can't be created across multi-partite quantum states using only local operations and classical communication (a.k.a $\locc$), and thereby separable states remain separable under $\locc$. The notion of a Schmidt rank extends this behavior to states with only a finite amount of entanglement across its partitions. For a bipartite pure state $\ket{\psi}$ in $\cD(\mathcal{H}_A\otimes \mathcal{H}_B)$, we write its Schmidt
decomposition as: 

\begin{equation}
    \ket{\psi} = \sum_{i=1}^R \sqrt{\lambda_i} \ket{a_i}\otimes \ket{b_i}
\end{equation}

\noindent Where $\{\ket{a_i}\}_{i\in [R]}$ and $\{\ket{b_i}\}_{i\in [R]}$ define collections of orthonormal vectors. The integer $R$ corresponds to the Schmidt rank of $\psi$, and is a measure of how entangled $\psi$ is across the bipartition. Moreover, $R\leq |A|, |B|$. \cite{Terhal1999SchmidtNF} introduced a notion of Schmidt number for density matrices, generalizing the definition for pure states. 

\begin{definition}[The Schmidt Number, \cite{Terhal1999SchmidtNF}]\label{def:schmidt-number} A bipartite density matrix $\rho$ has Schmidt number $R$ if:
    \begin{enumerate}
        \item For any decomposition of $\rho = \sum_i p_i \ketbra{\psi_i} $, at least one of the vectors with $p_i>0$, $\ket{\psi_i}$ has Schmidt rank at least $R$.
        \item There exists a decomposition of $\rho$ with all vectors $\{\ket{\psi_i}\}$ of Schmidt rank at most $R$.
    \end{enumerate}
\end{definition}

In other words, the Schmidt number of a mixed state $\rho$ is the minumum (over all possible decompositions of $\rho$) of the maximum Schmidt rank in the decomposition. \cite{Terhal1999SchmidtNF} proved that much like the the Schmidt rank for pure states, the mixed state Schmidt number cannot increase under unentangled operations. 

\begin{proposition}[\cite{Terhal1999SchmidtNF}]\label{prop:schimidt}
    The Schmidt number of a density matrix cannot increase under local quantum operations and classical communication $(\locc)$.
\end{proposition}

\suppress{
\begin{proposition}[\cite{Terhal1999SchmidtNF}]\label{prop:schimidt2} \mycomment{provide a proof or say its obvious}
Let $\rho_{ABC}$ be a state such that 
    the Schmidt number across the bipartition $(AB,C)$ is atmost $1$. Then the Schmidt number across the bipartition $(A,BC)$ is atmost $2^{\vert B \vert}$.
\end{proposition}
}

\subsection{Clifford Authentication Schemes}
\label{subsection:prelim-2-designs}

Roughly speaking, a quantum authentication scheme (QAS) enables two parties who share a secret key, to reliably communicate an encoded quantum state over a noisy channel with the following guarantee: If the state is untampered, the receiver accepts, whereas if the state was altered by the channel, the receiver rejects \cite{Barnum2001AuthenticationOQ,Aharonov2008InteractivePF}.

\cite{Aharonov2008InteractivePF} showed how to construct QASs from unitary 2-designs, and in particular \textit{the Clifford group}.\footnote{The Clifford group is the collection of unitaries that map Pauli matrices to Pauli matrices (up to a phase). We refer the reader to \Cref{section:cliffords} for a more comprehensive definitions of Paulis, Cliffords and their properties.} To encode a message state $\psi$, \cite{Aharonov2008InteractivePF} apply a uniformly random Clifford unitary $C$ to $\psi$, which later is decoded by applying $C^\dagger$ (potentially after adversarial tampering). The description of $C$ is used as a shared secret key between sender and receiver, and is unknown to the adversary. \cite{Aharonov2008InteractivePF} showed this scheme either recovers the message, or completely scrambles it, by appealing to Pauli randomization properties of the Clifford group. Here, we briefly summarize two of these properties, and describe a key-efficient family of unitary 2-designs by \cite{CLLW16}. 

The most basic property is that conjugating any state by a random Pauli (or Clifford), encrypts it:

\begin{fact}[Pauli's are $1$-designs]\label{fact:notequal}  Let $\rho_{AB}\in \cD(\cH_A\otimes \cH_B)$, and $\cP(\cH_A)$ be the Pauli group on $\cH_A$. Then,
$$\frac{1}{\vert \cP(\cH_A) \vert} \sum_{Q \in \cP(\cH_A) } (Q \otimes \id)   \rho_{A B}  ( Q^\dagger \otimes \id  )  = U_{A } \otimes  \rho_B.$$
\end{fact}

The second property we require is known as the ``Clifford Twirl" property. Informally, applying a uniformly random Clifford operator (by conjugation) maps any Pauli $Q$ to a uniformly random non-identity Pauli operator. \cite{Aharonov2008InteractivePF} used this property to argue the security of their authentication scheme after tampering.  We actually require a (standard) strengthening of the Clifford twirl property in the presence of side-information, which we depict visually in \cref{fig:splitstate21} and prove in \Cref{section:cliffords} for completeness:

\begin{lemma}[Clifford Twirl with Side Information]\label{lem:twirl-wsi}
 Consider \cref{fig:splitstate21}. Let $R$ denote a uniformly random key of $\log |\cC(\cH_A)|$ bits, where $\cC(\cH_A)$ denotes the Clifford group on $\cH_A$. Let $ \Lambda:\cD(\cH_A\otimes \cH_{E})\to \cD(\cH_{A}\otimes\cH_{E})$ be an arbitrary CPTP map on $A, E$. Then, there exists CP maps $\Phi_1, \Phi_2 :\cL(\cH_E)\to \cL(\cH_{E}) $ acting only on register $E$, depending only on $\Lambda$, such that
 \[  \rho_{\hat{A}AE} = \frac{1}{\vert \cC(\cH_A)\vert } \sum_{C \in \cC(\cH_A)}  C^\dagger\circ  \Lambda\circ  C  (\psi_{ \hat{A}AE})  \approx_{\frac{2}{2^{2 \vert A \vert}-1}}   \Phi_1 ( \psi_{\hat{A}AE}) + (\Phi_2( \psi_{\hat{A}E}) \otimes U_A),\]
$\forall \psi_{A \hat{A}E}\in \cD(\cH_A\otimes \cH_{\hat{A}}\otimes \cH_E)$. Moreover, $\Phi_1+\Phi_2$ is CPTP. 
\end{lemma}

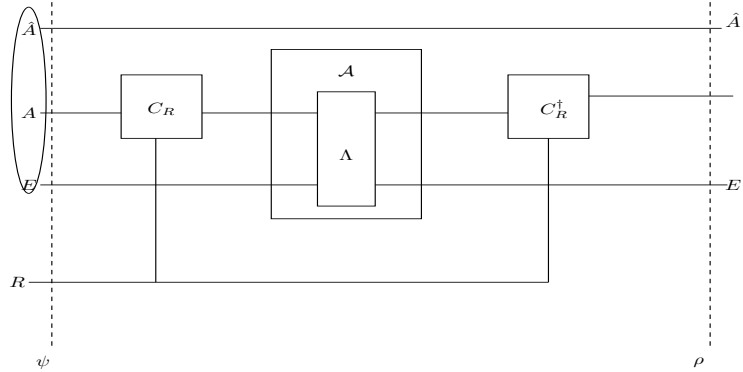
\begin{figure}[h]
\centering
\resizebox{10cm}{5cm}{
\begin{tikzpicture}
%\draw (1,1.2) ellipse (0.3cm and 2.5cm);
  %        \node at (1,3.3) {$R$};
  %          \node at (5.2,3.7) {$R$};
   %      \node at (14.5,3.7) {$R$};
   %      \draw (1.2,3.5) -- (3,3.5);
   %       \draw (5,3.5) -- (15,3.5);
%\draw (10,3.5) -- (15,3.5);
%\node at (1,-0.8) {$\hat{R}$};
%\node at (14.5,-0.8) {$\hat{R}$};
%\draw (1.2,-1) -- (15,-1);

%\draw (1,5) ellipse (0.2cm and 1cm);
\node at (1,6.5) {$\hat{A}$};
\node at (13.2,6.7) {$\hat{A}$};
\draw (1.2,6.5) -- (13,6.5);
\draw (1,4.8) ellipse (0.3cm and 2.2cm);
\draw (2.6,3.9) rectangle (4,5.4);
\draw (9.3,3.9) rectangle (10.7,5.4);
\node at (1,4.5) {$A$};

%\draw (0.8,0.5) ellipse (0.3cm and 1.5cm);
%\node at (0.8,-0.7) {$\hat{R}$};
\node at (3.3,4.6) {$C_R$};
\node at (10.1,4.6) {$C^\dagger_{R}$};
%\node at (1,6.5) {$\hat{M}$};
%\node at (14.8,6.5) {$\hat{M}$};

%\draw (1.2,6.5) -- (14.6,6.5);

%\node at (14.5,4.3) {$S'$};
\draw (1.2,4.5) -- (2.6,4.5);
%\draw (2.4,4.2) -- (2.6,4.2);
%\node at (2.1,4) {$\ket{0}^d$};
\draw (4,4.5) -- (6,4.5);
\draw (7,4.5) -- (9.3,4.5);
\draw (10.7,4.9) -- (13.2,4.9);
%\draw (10.7,4.5) -- (11,4.5);

%\draw (10.7,3) rectangle (11.5,4.4);

%\node at (11.7,4.5) {$\{  \ketbra{0^d}{0^d}, $};
%\node at (11.6,4) {$ \id -  \ketbra{0^d}{0^d}  \}$};

%\draw (10,4.5) -- (15,4.5);
%\node at (5.2,4.7) {$S$};

\node at (0.8,0.5) {$R$};
%\node at (1.5,1.5) {$S$};
\draw (1,0.5) -- (10,0.5);
%\node at (2.0,0) {$\ketbra{0^{n}}$};
%\draw (2.8,0) -- (3,0);
%\draw (3,-0.1) rectangle (4.5,3.3);
%\draw (1.5,-0.5) rectangle (4.8,5.9);
%\node at (3.3,6.1) {$\enc$};
%\node at (3.8,1.5) {$\nmcenc$};

%\draw (9.4,-0.5) rectangle (14.2,5.9);
%\node at (11.5,6.1) {$\dec$};

\draw  (3.2,0.5) -- (3.2,3.9);

\draw [dashed] (1.4,-1) -- (1.4,7.2);
%\draw [dashed] (3.5,-1.4) -- (3.5,7.2);
%\draw [dashed] (9.3,-1.4) -- (9.3,7.2);
%\draw [dashed] (10,-1.4) -- (10,7.2);
\draw [dashed] (12.8,-1) -- (12.8,7.2);
%\draw [dashed] (14.3,-0.8) -- (14.3,5.5);

%\draw (0.9,-1) -- (13.2,-1);

%\node at (6.2,1.6) {$\ket{\psi}_{W_1W_2}$};
\node at (1.25,-1.4) {$\psi$};
%\node at (3.35,-1.4) {$\rho$};
%\node at (9.15,-1.4) {$\tau$};
%\node at (10.2,-1.4) {$\sigma_3$};
\node at (12.6,-1.4) {$\rho$};
%\node at (14.13,-0.8) {$\rho$};

%\node at (4.6,3) {$E$};
\node at (1,2.8) {$E$};

\draw (1.2,2.8) -- (6,2.8);
%\draw (7,2.8) -- (8.5,2.8);
%\draw (9.5,2.8) -- (10.5,2.8);
\draw (7,2.8) -- (13.1,2.8);

%\node at (12.7,2.6) {$R'$};
%\draw (12.5,2.4) -- (13.5,2.4);
\draw (10,0.5) -- (10,3.9);

%\node at (4.6,0.4) {$Y$};
%\node at (8,0.0) {$Y'$};
%\draw (4.5,0.2) -- (6,0.2);
%\draw (7,0.2) -- (8.5,0.2);

%\draw (9.5,0.2) -- (10.5,0.2);
%\draw (6,2.5) -- (5.7,2.5);
%\draw (7,2.5) -- (7.3,2.5);
%\node at (5.6,2.3) {$E'$};
%\node at (7.4,2.3) {$E'$};

\draw (6,2.3) rectangle (7,5);
\node at (6.5,3.5) {$ \Lambda$};
%\draw (6,0) rectangle (7,1);
%\node at (6.5,0.5) {$V$};

\node at (6.5,5.5) {$\mathcal{A}$};
\draw (5.2,2) rectangle (7.8,6);

%\draw (5.5,1.5) ellipse (0.2cm and 1cm);
%\node at (5.5,2) {$W_1$};
%\draw (5.7,2.2) -- (6,2.2);
%\node at (7.4,2.1) {$W'_1$};
%\draw (7,2.2) -- (7.2,2.2);
%\node at (5.5,1) {$W_2$};
%\draw (5.7,0.7) -- (6,0.7);
%\node at (7.4,0.9) {$W'_2$};
%\draw (7.0,0.7) -- (7.2,0.7);

%\draw (9,2.8) circle (0.5);
%\node at (9,2.8) {$\mathcal{M}$};
%\draw (9,0.2) circle (0.5);
%\node at (9,0.2) {$\mathcal{M}$};

%\draw (8,-0.5) rectangle (10,1.2);
%\draw (8,1.5) rectangle (10,3.2);
%\node at (9,1.5) {$2\mhyphen\nmext$};

%\draw (10.5,-0.2) rectangle (12.5,3.0);
%\node at (11.5,1.5) {${\nmcdec}$};
%\node at (12.5,1.6) {$ \mathcal{M}= \{ \ketbra{0^{2n}},$};
%\node at (12.5,1) {$ \id- \ketbra{0^{2n}} \}$};
\node at (13.2,2.8) {$E$};
\end{tikzpicture}}

\caption{Clifford Twirling with Side Information.}\label{fig:splitstate21}
\end{figure}

In this work, following \cite{Boddu2023SplitState}, we leverage a particular key-efficient (near-linear) construction of 2-designs by Cleve, Leung, Liu and Wang \cite{CLLW16}:

\begin{fact}[Subgroup of the Clifford group~\cite{CLLW16}]\label{lem:subclifford}
There exists a subgroup $\mathcal{SC}(\cH_A)$ of the Clifford group $\mathcal{C}(\cH_A)$ such that given any non-identity Pauli operators $P ,Q \in \cP(\cH_A)$ we have that
\begin{equation*}
    \vert \{ C \in \cSC(\cH_A) \vert C^\dagger P C =Q \} \vert = \frac{\vert \cSC(\cH_A)  \vert}{\vert\cP(\cH_A) \vert -1} \quad \textrm{and} \quad \vert \cSC(\cH_A)  \vert = 2^{5 \vert A \vert }-2^{3 \vert A \vert}.
\end{equation*} 
Moreover, there exists a procedure $\samp$ which given as input a uniformly random string $R\leftarrow \bits^{5|A|}$ outputs in time $\poly(|A|)$ a Clifford operator $C_R\in\mathcal{SC}(\cH_A)$ drawn from a distribution $2 \cdot 2^{-2|A|}$ close to the uniform distribution over $\cSC(\cH_A)$.
\end{fact}

We remark that both properties above are true as stated even for unitaries chosen at random from $\mathcal{SC}(\cH_A)$ using $\samp$.

\newpage

\section{Coding-Theoretic Definitions in the Split-State Tampering Model}\label{section:coding-definitions}

In this section, we introduce several coding-theoretic notions of security in various adversarial settings. We begin in \cref{subsection:formal-tampering} by introducing the split-state tampering models considered in this work, and proceed in \cref{subsection:prelim-non-malleability} by discussing classical non-malleability, quantum non-malleability and quantum tamper detection.

\subsection{Formal Definitions of Split-State Tampering Models}
\label{subsection:formal-tampering}

In the absence of a secret key, it is impossible to offer tamper-resilience against arbitrary tampering channels. In this subsection, we formalize the various models of split-state tampering that we consider in this work. Let $t$ parties share a quantum state defined on $\cD(\mathcal{H}_1\otimes \cdots \otimes \mathcal{H}_t)$.

 \begin{definition}[$\lo^t$: Unentangled $t$-Split-State Model]\label{def:lot}
   These are $t$-split quantum adversaries \emph{without shared entanglement}, comprised of tensor product channels $\Lambda = \Lambda_1\otimes \cdots\otimes  \Lambda_t$ where each $\Lambda_i:\cL(\mathcal{H}_i)\rightarrow \cL(\mathcal{H}_i)$ is a CPTP map. 
 \end{definition}

If each channel $\Lambda_i$ were a deterministic function $f_i$ and the ciphertext were classical, then we we would recover the classical model \cite{LL12},\footnote{Moreover, codes which are secure against this model also capture security in the presence of shared randomness.} see \cref{fig:splitstate121}. \cite{ABJ22,Boddu2023SplitState,BBJ23} generalized the definition above to adversaries with an unbounded amount of pre-shared entanglement, a model we refer to as $\lo^t_{*}$.

\begin{definition}[$\lo^t_{*}$: Entangled $t$-Split-State Model]\label{def:lotsharedentanglementunbounded}
     These are $t$-split quantum adversaries with \emph{unbounded shared entanglement}, comprised of an arbitrary multi-partite ancilla state $\psi\in \cD(\mathcal{H}_{E_1}\otimes\cdots \otimes \mathcal{H}_{E_t})$, and a tensor product channel $(\Lambda_1\otimes \cdots\otimes \Lambda_t)$ where each $\Lambda_i:\cL(\mathcal{H}_i\otimes \mathcal{H}_{E_i})\rightarrow \cL(\mathcal{H}_i\otimes \mathcal{H}_{E'_i})$ is a CPTP map.
 \end{definition}

Tamper detection is impossible in $\lo^t_{*}$, due to the standard substitution attack wherein the adversaries store in $\phi$ a pre-shared copy of the encoding a fixed message. This impossibility motivates the following ``bounded storage model", where we restrict the amount of pre-shared entanglement between the parties. 

 \begin{definition}[$\lo^t_{(e_1,e_2, \ldots, e_t)}$: Bounded Storage Model]\label{def:lotsharedentanglement}
     These are $t$-split quantum adversaries with \emph{bounded shared entanglement}, where the ancilla state $\psi\in \cD(\mathcal{H}_{E_1}\otimes\cdots \otimes \mathcal{H}_{E_t})$, satisfies $|E_i|\leq e_i, \forall i\in [t]$.
 \end{definition}

 A visual representation of this bounded storage model with $t=3$ can be found in \cref{fig:splitstate211d}. Note $\lo^t_{(e_1,e_2, \ldots, e_t)} $ simultaneously generalizes $\lo^t$ (when each $e_i=0$) and $\lo^t_*$  (when each $e_i\rightarrow \infty$).

\subsection{Non-Malleable Codes and Tamper-Detection Codes}
\label{subsection:prelim-non-malleability}

We dedicate this section to define classical and quantum non-malleable codes and tamper-detection codes.

\begin{figure}[ht]
		\centering
  \resizebox{10cm}{5cm}{
	\begin{tikzpicture}
		%\draw (1,1.2) ellipse (0.3cm and 2.5cm);
		%        \node at (1,3.3) {$R$};
		%          \node at (5.2,3.7) {$R$};
		%      \node at (14.5,3.7) {$R$};
		%      \draw (1.2,3.5) -- (3,3.5);
		%       \draw (5,3.5) -- (15,3.5);
		%\draw (10,3.5) -- (15,3.5);
		%\node at (1,-0.8) {$\hat{R}$};
		%\node at (14.5,-0.8) {$\hat{R}$};
		%\draw (1.2,-1) -- (15,-1);
		
		%\draw (1,5) ellipse (0.2cm and 1cm);
		%\node at (1,5.5) {$\hat{S}$};
		%\node at (14.5,5.5) {$\hat{S}$};
		%\draw (1.2,5.7) -- (15,5.7);
		\node at (1.8,1.5) {$M$};
		\node at (11,1.5) {$M'$};
		%\draw (1.2,4.5) -- (5,4.5);
		%\draw (5,4.5) -- (15,4.5);
		%\draw (10,4.5) -- (15,4.5);
		%\node at (5.2,4.7) {$S$};
	%	\node at (14.5,4.7) {$S$};
		
		%\node at (1.5,1.5) {$S$};
		\draw (2,1.5) -- (3,1.5);
		%\node at (2.0,0) {$\ketbra{0^{n}}$};
		%\draw (2.8,0) -- (3,0);
		\draw (3,-0.5) rectangle (4.5,3.5);
		\node at (3.8,1.5) {$\enc$};
		
	%	\draw  (2,1.5) -- (2,4.5);
	%	\draw [dashed] (4.78,-0.8) -- (4.78,5.5);
	%	\draw [dashed] (8.3,-0.8) -- (8.3,5.5);
	%	\draw [dashed] (10,-0.8) -- (10,5.5);
		%\draw [dashed] (14.3,-0.8) -- (14.3,5.5);
		
	%	\node at (6.2,1.6) {$\ket{\psi}_{NM}$};
	%	\node at (4.64,-0.8) {$\sigma$};
	%	\node at (8.1,-0.8) {$\hat{\rho}$};
	%	\node at (10.2,-0.8) {${\rho}$};
		%\node at (14.13,-0.8) {$\rho$};
		
		\node at (4.7,3) {$X$};
		\node at (8.1,3) {$X'$};
		\draw (4.5,2.8) -- (6,2.8);
		\draw (7,2.8) -- (8.5,2.8);
		\draw (10,1.5) -- (10.7,1.5);
	%	\node at (14.5,2.6) {$S'$};
		%\draw (14,2.4) -- (15,2.4);

            \node at (4.7,1.8) {$Y$};
		\node at (8.1,1.4) {$Y'$};
          \draw (4.5,1.6) -- (6,1.6);
		\draw (7,1.6) -- (8.5,1.6);
         \draw (6,1) rectangle (7,2);
         \node at (6.5,1.5) {$f_2$};
  
		\node at (4.7,0.4) {$Z$};
		\node at (8.1,0) {$Z'$};
		\draw (4.5,0.2) -- (6,0.2);
		\draw (7,0.2) -- (8.5,0.2);
	%	\draw (9.5,0.2) -- (11,0.2);
		
		\draw (6,2.2) rectangle (7,3.2);
		\node at (6.5,2.7) {$f_1$};
		\draw (6,-0.2) rectangle (7,0.8);
		\node at (6.5,0.3) {$f_3$};
		
		%\node at (6.5,-0.5) {$\mathcal{A}=(f_1,f_2,f_3)$};
		\draw (5.2,-0.8) rectangle (7.8,3.8);

		%\draw (5.5,1.5) ellipse (0.2cm and 1cm);
	%	\node at (5.5,2) {$W_1$};
	%	\draw (5.7,2.2) -- (6,2.2);
		%\node at (7.4,2.1) {$W'_1$};
	%	\draw (7,2.2) -- (7.2,2.2);
	%	\node at (5.5,1) {$M$};
	%	\draw (5.7,0.7) -- (6,0.7);
		%\node at (7.4,0.9) {$M'$};
		%\draw (7.0,0.7) -- (7.2,0.7);
		
		%\draw (9,2.8) circle (0.5);
	%	\node at (9,2.8) {$\mathcal{M}$};
	%	\draw (9,0.2) circle (0.5);
	%	\node at (9,0.2) {$\mathcal{M}$};
		%\draw (8,-0.5) rectangle (10,1.2);
		%\draw (8,1.5) rectangle (10,3.2);
		%\node at (9,1.5) {$2\mhyphen\nmext$};
		
		\draw (8.5,-0.5) rectangle (10,3.2);
		\node at (9.2,1.5) {$\dec$};
		%\node at (12.5,1.6) {$ \mathcal{M}= \{ \ketbra{0^{2n}},$};
		%\node at (12.5,1) {$ \id- \ketbra{0^{2n}} \}$};
		
	\end{tikzpicture} }
	\caption{$t$-split-state tampering model for $t=3$.}\label{fig:splitstate121}
\end{figure}
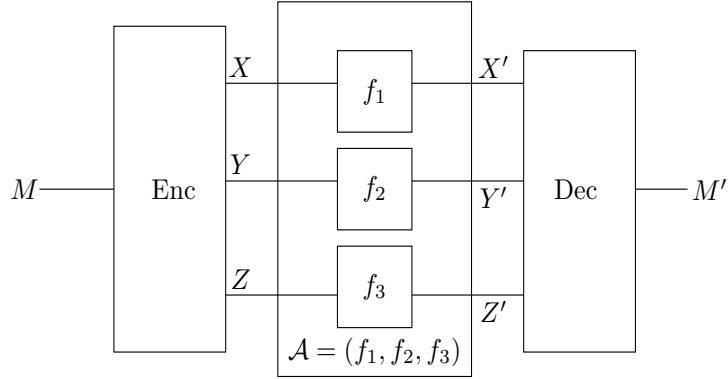

\subsubsection{Classical Non-Malleable Codes}
\label{subsection:prelim-classical-nm}

Consider the following tampering experiment. Given some message $m$, encode it into a coding scheme $c\leftarrow \Enc(m)$, tamper with it using a function $f:$ $c'=f(c)$, and finally decode it to $\Dec(f(c)) =  m'$. If $\Enc$ were an error correcting code robust against $f$, then our goal would be to recover $m'=m$. If $\Enc$ were a \textit{tamper} detecting code robust against $f$, we would like either $m'=m$ or we reject, $m'=\bot$. Unfortunately, error correction or detection isn't possible for every $f$.\footnote{For instance, constant functions.} 

Nevertheless, there are cases where the notion of non-malleability can be much more versatile. Dziembowski, Pietrzak, and Wichs formalized a coding-theoretic notion of non-malleability using a simulation paradigm: $m'$ can be simulated by a distribution that depends only on  $f$ (and not the message):

\begin{definition}  [Non-Malleable Codes, adapted from \cite{DPW10}] \label{definition:non-malleability}
   A pair of (randomized) algorithms $(\Enc :\{0, 1\}^k \rightarrow\{0, 1\}^n, \Dec : \{0, 1\}^n \rightarrow\{0, 1\}^k)$, is an $(\epsilon, \delta)$ \emph{non-malleable code} w.r.t a family of functions $\mathcal{F} \subset \{f : \{0, 1\}^n \rightarrow\{0, 1\}^n\}$, if the following properties hold:
    \begin{enumerate}
        \item \textbf{Correctness:} For every message $m\in \{0, 1\}^k$, $\mathbb{P}[\Dec\circ \Enc(m) \neq m] \leq \delta$.\\
        \item \textbf{Non-Malleability:} For every $f\in \mathcal{F}$ there is a constant $p_f\in [0, 1]$ and a random variable $D_f$ on $\{0, 1\}^k$ which is independent of the message, s.t. 
        \begin{equation}
        \forall m\in \{0, 1\}^k: \Dec\circ f\circ \Enc(m) \approx_\epsilon p_f\cdot m + (1-p_f)\cdot \mathcal{D}_f
    \end{equation}
    \end{enumerate} 
\end{definition}

In other words, the outcome of the effective channel on the message is a convex combination of the original message or a fixed distribution.\footnote{Non-malleability is originally defined in terms of a simulator which is allowed to output a special symbol \emph{same}$^*$ which is later replaced by the encoded $m$. The reader may be more accustomed to this notion, nevertheless, it is well-known to be equivalent to the above. } As previously introduced, the classical $t$-split-state model corresponds to the the collection of individual tampering functions $\{(f_1, f_2,\cdots ,f_t)\}$ on the shares (see \cref{fig:splitstate121}).

    \suppress{When designing NMCs for classical messages encoded into quantum states, we change the definitions above only slightly. If the Hilbert space $\mathcal{H}$ denotes the codespace, we modify the image and domain of $\Enc, \Dec$ to be $\cD(\cH)$, and allow the adversary to tamper with the code using channels $\Lambda:\cD(\cH)\rightarrow \cD(\cH)$.}

\subsubsection{Quantum Secure Non-Malleable Codes and Extractors}
\label{subsubsection:prelim-nmext}

Unfortunately, constructions of high-rate (quantum secure) non-malleable codes remain an open problem. Instead, in order to optimize the rate of our results, we will appeal to fact that we don't actually require the non-malleability of a worst case message in our codes, just that of a random secret key. For this purpose, \cite{BBJ23,Boddu2023SplitState} leverage the refined notion of a \textit{quantum secure, 2-source non-malleable extractor}, and we follow suit.

The connection between 2-source non-malleable extractors and codes was first made explicit by \cite{Cheraghchi2013NonmalleableCA}. Informally, a $2$-source extractor $\nmext$ takes as input two some-what random sources $X$ and $Y$, and outputs a near-uniformly random string $R=\nmext(X,Y)$. What makes its non-malleable is that $R$ is still near uniformly random even conditioned on any $R'=\nmext(f(X),g(Y))$, which is the outcome of independent tampering functions $f, g$ on the sources $X, Y$. In our work we make use of a construction of \emph{quantum secure} $2$-source non-malleable extractors by \cite{BBJ23}:

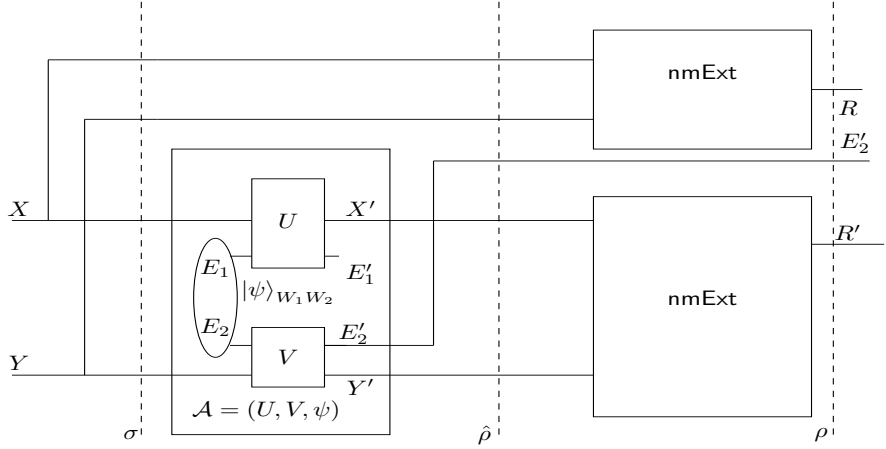
\begin{figure}[t]
\centering

\resizebox{12cm}{6cm}{

\begin{tikzpicture}
%\draw (1,1.2) ellipse (0.3cm and 2.5cm);
  %        \node at (1,3.3) {$R$};
  %          \node at (5.2,3.7) {$R$};
   %      \node at (14.5,3.7) {$R$};
   %      \draw (1.2,3.5) -- (3,3.5);
   %       \draw (5,3.5) -- (15,3.5);
%\draw (10,3.5) -- (15,3.5);
%\node at (1,-0.8) {$\hat{R}$};
%\node at (14.5,-0.8) {$\hat{R}$};
%\draw (1.2,-1) -- (15,-1);

%\draw (1,5) ellipse (0.2cm and 1cm);
%\node at (1,5.5) {$\hat{S}$};
%\node at (14.5,5.5) {$\hat{S}$};
%\draw (1.2,5.7) -- (15,5.7);
%\node at (1,4.5) {$R$};
%\node at (14.5,4.3) {$R'$};
\draw (4,4.5) -- (5,4.5);
\draw (5,4.5) -- (11,4.5);

\draw (3.5,5.5) -- (5,5.5);
\draw (5,5.5) -- (11,5.5);
%\draw (10,4.5) -- (15,4.5);
%\node at (5.2,4.7) {$S$};
\node at (14.5,4.7) {$R$};
\node at (14.6,4.1) {$E'_2$};

%\node at (1.5,1.5) {$S$};
%\draw (2,1.5) -- (3,1.5);
%\node at (2.0,0) {$\ketbra{0^{n}}$};
%\draw (2.8,0) -- (3,0);
%\draw (3,-0.5) rectangle (4.5,3.5);
\node at (12.5,5.3) {$\nmext$};

%\draw  (2,1.5) -- (2,4.5);
\draw [dashed] (4.78,-0.8) -- (4.78,6.5);
\draw [dashed] (9.7,-0.8) -- (9.7,6.5);
%\draw [dashed] (10,-0.8) -- (10,5.5);
\draw [dashed] (14.3,-0.8) -- (14.3,6.5);

\node at (6.8,1.6) {$\ket{\psi}_{W_1W_2}$};
\node at (4.64,-0.8) {$\sigma$};
\node at (9.5,-0.8) {$\hat{\rho}$};
%\node at (10.2,-0.8) {${\rho}$};
\node at (14.13,-0.8) {$\rho$};

\node at (3.1,3) {$X$};
\node at (7.8,3) {$X'$};
\draw (3,2.8) -- (6.3,2.8);
%\draw (7,2.8) -- (8.5,2.8);
%\draw (9.5,2.8) -- (11,2.8);
\draw (7.3,2.8) -- (11,2.8);
\node at (14.5,2.6) {$R'$};
\draw (14,2.4) -- (15,2.4);

\node at (3.1,0.4) {$Y$};
\node at (7.8,0.0) {$Y'$};
\draw (3,0.2) -- (6.3,0.2);
%\draw (7,0.2) -- (8.5,0.2);
%\draw (9.5,0.2) -- (11,0.2);
\draw (7.3,0.2) -- (11,0.2);
\draw (3.5,2.8) -- (3.5,5.5);
\draw (4,0.2) -- (4,4.5);

%\draw (4,0.2) -- (4,4.5);
%\draw (3.5,2.8) -- (3.5,5.5);
%\draw (11,4) rectangle (14,6);
%\draw (4.5,5.5) -- (5,5.5);
%\draw (5,5.5) -- (11,5.5);
%\draw (14,5) -- (14.7,5);
\draw (14,5) -- (14.7,5);

\draw (7.3,0.7) -- (8.8,0.7);
\draw (8.8,0.7) -- (8.8,3.8);
\draw (8.8,3.8) -- (14.8,3.8);

\draw (6.3,2) rectangle (7.3,3.5);
\node at (6.8,2.8) {$U$};
\draw (6.3,0) rectangle (7.3,1);
\node at (6.8,0.5) {$V$};

\node at (6.5,-0.4) {$\mathcal{A}=(U,V,\psi)$};
\draw (5.2,-0.8) rectangle (8.2,4);

\draw (5.8,1.5) ellipse (0.3cm and 1cm);
\node at (5.8,2) {$E_1$};
\draw (6,2.2) -- (6.3,2.2);
\node at (7.8,1.9) {$E'_1$};
\draw (7.3,2.2) -- (7.5,2.2);
\node at (5.8,1) {$E_2$};
\draw (6,0.7) -- (6.3,0.7);
\node at (7.7,0.9) {$E'_2$};
\draw (7.3,0.7) -- (7.6,0.7);
%\draw (9,2.8) circle (0.5);
%\node at (9,2.8) {$\mathcal{M}$};
%\draw (9,0.2) circle (0.5);
%\node at (9,0.2) {$\mathcal{M}$};

%\draw (8,-0.5) rectangle (10,1.2);
%\draw (8,1.5) rectangle (10,3.2);
%\node at (9,1.5) {$2\mhyphen\nmext$};

%\draw (11,-0.5) rectangle (14,3.2);
%\node at (12.5,1.5) {$\cdec$};
%\node at (12.5,1.6) {$ \mathcal{M}= \{ \ketbra{0^{2n}},$};
%\node at (12.5,1) {$ \id- \ketbra{0^{2n}} \}$};
\draw (11,4) rectangle (14,6);
\draw (11,-0.5) rectangle (14,3.2);
%\node at (11.5,3.6) {$\cdec$};

%\draw (8.4,-0.6) rectangle (14.1,3.5);
\node at (12.5,1.5) {${\nmext}$};

\end{tikzpicture} }
\caption{Tampering experiment for quantum secure $2$-source non-malleable extractors.}\label{fig:splitstate5}
\end{figure}

\begin{fact}[Quantum Secure $2$-Source Non-malleable Extractor~\cite{BBJ23}]\label{lem:qnmcodesfromnmext}
 Co-nsider the split-state tampering experiment in \cref{fig:splitstate5} with a split-state tampering adversary $\cA=(U,V,\ket\psi_{E_1 E_2})$.
Based on this figure, define $p_\sm=\mathbb{P}[(X,Y)=(X',Y')]_{\hat\rho}$ and the conditioned quantum states
\begin{gather}
    \rhosame = (\nmext\otimes\nmext)(\hat\rho | (X,Y)=(X',Y')) \\\text{ and }\rhotamp = (\nmext\otimes\nmext)(\hat\rho | (X,Y)\neq(X',Y')).
\end{gather}
Then, for any $n\geq n_0$ and constant $\delta>0$, there exists an explicit function $\nmext : \{ 0,1\}^n \times \{0,1 \}^{\delta n} \rightarrow \{0,1 \}^{r}$ with $r=(1/2-\delta)\cdot n$ such that for independent sources $X\leftarrow\bits^n$ and $Y \leftarrow \bits^{\delta n}$ and any such split-state tampering adversary $\cA=(U,V,\ket\psi_{E_1 E_2})$, $\nmext$ satisfies
\begin{enumerate}
    \item Strong Extraction: $ \Vert  \nmext(X,Y)X - U_{r} \otimes U_n  \Vert_1 \leq \eps$ and $\Vert  \nmext(X,Y)Y - U_{r} \otimes U_{\delta n}  \Vert_1 \leq \eps,$\\
    
    % \item $\rho_{RR'W_2'} = p_{\mathcal{A}} \rho^1_{RR'W_2'} + (1-p_{\mathcal{A}}) \rho^2_{RR'W_2'}$,
    
    \item Augmented Non-Malleability: $p_\sm\Vert \rhosame_{RE_2'}-  U_r \otimes \rhosame_{E_2'} \Vert_1 +  (1-p_\sm)  \Vert \rhotamp_{RR'E_2'}-  U_r \otimes \rhotamp_{R'E_2'}  \Vert_1  \leq \eps$,
\end{enumerate}
with $\eps =2^{-n^{\Omega_\delta(1)}}$. Furthermore, $\nmext(x,y)$ can be computed in time $\poly(n)$. 
\end{fact}

Item 1 in \cref{lem:qnmcodesfromnmext} should be the familiar 2-source (strong) extraction guarantee. Item 2 however, captures the \textit{augmented} non-malleability of $\nmext$. If the tampering attack doesn't change $X$ and $Y$, then $R$ should be close to uniformly random even together with any updated side-entanglement $E'_2$ held by one of the adversaries. Conversely, if the tampering attack changed $X$ and $Y$, then $R$ should be independent of the output $R'=\nmext(X',Y')$, jointly with $E'_2$.\footnote{The role of \textit{augmented} non-malleability is more subtle in the quantum secure model. Classically, the tampering functions are deterministic, and thus the left share $X$ (before tampering) uniquely determines the ``view" of the left adversary. In contrast, in the quantum secure model, the entangled state that the split-state adversaries share may collapse and correlate with \textit{both} post-tampered shares $X', Y'$.
}

 \subsubsection{Quantum Non-Malleable Codes} To introduce our new Tamper detection codes, we first need to recollect the definition of quantum NMCs from~\cite{Boddu2023SplitState}. For concreteness, we focus on the model $\lo^3_{(e_1,e_2,e_3)}$. Our coding scheme is given by encoding and decoding Completely Positive Trace-Preserving (CPTP) maps
\begin{gather}
    \enc : \cL( \cH_M) \to \cL(\cH_{X}\otimes \cH_Y\otimes\cH_{Z}),  \dec  : \cL(\cH_{X}\otimes\cH_{Y} \otimes\cH_{Z}) \to  \cL(\cH_{M})
\end{gather}

Here, $\cH_{X}, \cH_{Y}, \cH_{Z}$ denote the Hilbert spaces for the three shares. The most basic property we require of this coding scheme $(\enc,\dec)$ is correctness:
\begin{equation*}
\dec(\enc(\sigma_{M\hat{M}}))=\sigma_{M\hat{M}}, \text{ where } \hat{M} \text{ represents entangled side-information.}
\end{equation*}

Now, suppose we tamper with the code using a tampering adversary $\cA\in \lo^3_{(e_1,e_2,e_3)}$. Recall from \cref{def:lotsharedentanglement} that such adversaries are specified by three tampering maps, $U, V, W$ along with a pre-shared state $\ket\psi_{E_1 E_2E_3}$. After the tampering channel, the decoding procedure $\dec$ is applied to the corrupted codeword $\tau_{X' Y'Z'}$ and stored in register $M'$, see \cref{fig:splitstate211d}. Let 
\begin{equation*}
    \eta = \dec( ( U \otimes V \otimes W)  \left(\enc( \sigma_{M \hat{M}}) \otimes \ketbra{\psi} \right)  )
\end{equation*}
be the final state after applying the tampering adversary $\cA$ followed by the decoding procedure. The non-malleability of the coding scheme $(\enc,\dec)$ is defined as follows.

\begin{figure}[t]
		\centering
   \resizebox{12cm}{5cm}{
		\begin{tikzpicture}
			%\draw (1,1.2) ellipse (0.3cm and 2.5cm);
			%        \node at (1,3.3) {$R$};
			%          \node at (5.2,3.7) {$R$};
			%      \node at (14.5,3.7) {$R$};
			%      \draw (1.2,3.5) -- (3,3.5);
			%       \draw (5,3.5) -- (15,3.5);
			%\draw (10,3.5) -- (15,3.5);
			%\node at (1,-0.8) {$\hat{R}$};
			%\node at (14.5,-0.8) {$\hat{R}$};
			%\draw (1.2,-1) -- (15,-1);
			
			%\draw (1,5) ellipse (0.2cm and 1cm);
			%\node at (1,5.5) {$\hat{S}$};
			%\node at (14.5,5.5) {$\hat{S}$};
			%\draw (1.2,5.7) -- (15,5.7);

     \draw (2,4.3) -- (11,4.3);
     \node at (1.8,4.3) {$\hat{M}$};
      \node at (11.2,4.3) {$\hat{M}$};
			\node at (1.8,1.5) {$M$};
			\node at (11,1.5) {$M'$};

\draw (1.8,2.8) ellipse (0.3cm and 1.8cm);
   
			%\draw (1.2,4.5) -- (5,4.5);
			%\draw (5,4.5) -- (15,4.5);
			%\draw (10,4.5) -- (15,4.5);
			%\node at (5.2,4.7) {$S$};
			%	\node at (14.5,4.7) {$S$};
			
			%\node at (1.5,1.5) {$S$};
			\draw (2,1.5) -- (3,1.5);
			%\node at (2.0,0) {$\ketbra{0^{n}}$};
			%\draw (2.8,0) -- (3,0);
			\draw (3,-0.5) rectangle (4.5,3.5);
			\node at (3.8,1.5) {$\enc$};
			
			%	\draw  (2,1.5) -- (2,4.5);
				\draw [dashed] (2.78,-0.8) -- (2.78,5.5);
			%	\draw [dashed] (8.3,-0.8) -- (8.3,5.5);
			%	\draw [dashed] (10,-0.8) -- (10,5.5);
			\draw [dashed] (10.3,-0.8) -- (10.3,5.5);
			
			%	\node at (6.2,1.6) {$\ket{\psi}_{NM}$};
				\node at (2.64,-0.8) {$\sigma$};
			%	\node at (8.1,-0.8) {$\hat{\rho}$};
			%	\node at (10.2,-0.8) {${\rho}$};
			\node at (10.13,-0.8) {$\eta$};
			
			\node at (4.7,3) {$X$};
			\node at (8.25,3) {$X'$};
  
			\draw (4.5,2.8) -- (6,2.8);
			\draw (7,2.8) -- (8.6,2.8);
			\draw (10,1.5) -- (10.7,1.5);
			%	\node at (14.5,2.6) {$S'$};
			%\draw (14,2.4) -- (15,2.4);
			
			\node at (4.7,0.4) {$Z$};
			\node at (8.25,0) {$Z'$};
			\draw (4.5,0.2) -- (6,0.2);
			\draw (7,0.2) -- (8.6,0.2);
			%	\draw (9.5,0.2) -- (11,0.2);

                \node at (4.7,1.8) {$Y$};
		      \node at (8.25,1.8) {$Y'$};
         \draw  (4.5,1.6) -- (5,1.6);
            \draw [dashed] (5,1.6) -- (6,1.6);
		  \draw (7,1.6) -- (8.6,1.6);
                \draw (6,1.1) rectangle (7,1.9);
                \node at (6.5,1.5) {$V$};
            \node at (7.5,1.3) {$E'_2$};
             \node at (5.5,1.4) {$E_2$};
             \draw (5.8,1.4) -- (6,1.4);
             \draw (7,1.4) -- (7.2,1.4);
   
			\draw (6,2.1) rectangle (7,2.9);
			\node at (6.5,2.5) {$U$};
			\draw (6,0.1) rectangle (7,0.9);
			\node at (6.5,0.5) {$W$};
			
			\node at (6.5,-0.4) {$\mathcal{A}=(U,V,W,\ket{\psi}_{})$};
			\draw (5,-0.8) rectangle (8,3.8);

			\draw (5.5,1.5) ellipse (0.3cm and 1.1cm);
			\node at (5.5,2.2) {$E_1$};
				\draw (5.7,2.2) -- (6,2.2);
			\node at (7.5,2.2) {$E'_1$};
				\draw (7,2.2) -- (7.2,2.2);
				\node at (5.5,0.8) {$E_3$};
				\draw (5.7,0.7) -- (6,0.7);
			\node at (7.5,0.6) {$E'_3$};
			\draw (7.0,0.7) -- (7.2,0.7);
			
			%\draw (9,2.8) circle (0.5);
			%	\node at (9,2.8) {$\mathcal{M}$};
			%	\draw (9,0.2) circle (0.5);
			%	\node at (9,0.2) {$\mathcal{M}$};
			%\draw (8,-0.5) rectangle (10,1.2);
			%\draw (8,1.5) rectangle (10,3.2);
			%\node at (9,1.5) {$2\mhyphen\nmext$};
			
			\draw (8.6,-0.5) rectangle (10,3.2);
			\node at (9.2,1.5) {$\dec$};
			%\node at (12.5,1.6) {$ \mathcal{M}= \{ \ketbra{0^{2n}},$};
			%\node at (12.5,1) {$ \id- \ketbra{0^{2n}} \}$};
			
		\end{tikzpicture} }
		\caption{$\lo^t_{(e_1, \ldots, e_t)}:$ $t$-split tampering model with shared entanglement for $t=3$. The shared entanglement $\psi$ is stored in registers $E_1, E_2,E_3$ such that $\vert E_i \vert \leq e_i$.}\label{fig:splitstate211d}
	\end{figure}
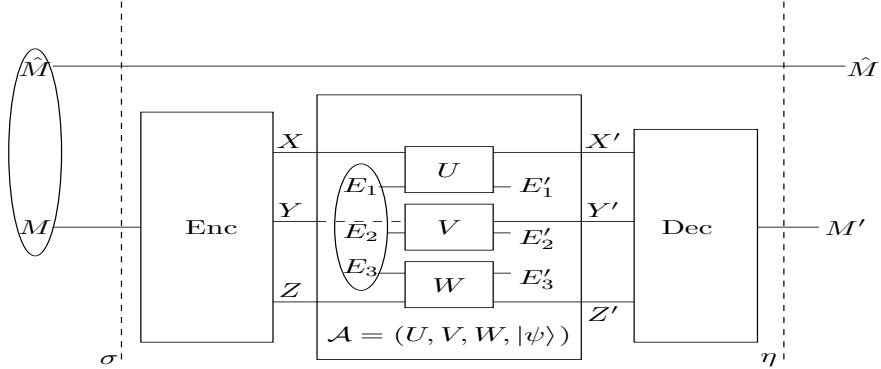

\begin{definition}[Quantum Non-Malleable Code against $\lo^3_{(e_1,e_2,e_3)}$~\cite{Boddu2023SplitState}]\label{def:qnmcodesfinaldef}  
See \cref{fig:splitstate211d} for the tampering experiment. The coding scheme
$(\enc, \dec)$ is an \emph{$\eps$-secure quantum non-malleable code} against $\lo^3_{(e_1,e_2,e_3)}$ if for every adversary $\cA=(U,V,W,\ket\psi_{E_1 E_2E_3})\in \lo^3_{(e_1,e_2,e_3)}$ and message  $\sigma_{M\hat{M}}$, 
\begin{equation}\label{eq:splitnmc}
    \eta_{M'\hat{M}} \approx_\eps p_{\cA} \sigma_{M \hat{M}}+(1-p_{\cA})\gamma^{\cA}_{M'}\otimes \sigma_{\hat{M}},
\end{equation}
where $p_{\cA}\in[0,1]$ and $\gamma^{\cA}_{M'}$ depend only on the tampering adversary $\cA$.
\end{definition}

\subsubsection{Quantum Tamper-Detection Codes}

Finally, we formally define the notion of quantum tamper-detection codes. Following the previous subsection, for concreteness we present an explicit definition for $\lo^t_{(e_1,e_2, \ldots, e_t)}$ tampering adversaries with $t=3$, which can be easily generalized to the different tampering models of \cref{subsection:formal-tampering}.

\begin{definition}[Quantum Tamper-Detection Code against $\lo^3_{(e_1,e_2, e_3)}$]\label{def:qtdcodesfinaldef}  
See \cref{fig:splitstate211d} for the tampering experiment. We say that the coding scheme
$(\enc, \dec)$ is an \emph{$\eps$-secure quantum tamper-detection code} against $\lo^3_{(e_1,e_2, e_3)}$ if for every $\lo^3_{(e_1,e_2, e_3)}$ adversary $\cA=(U,V,W,\ket\psi_{E_1 E_2E_3})$ and for every quantum message $\sigma_M$ (with canonical purification $\sigma_{M\hat{M}}$)
it holds that
\begin{equation}\label{eq:splittdc}
    \eta_{M'\hat{M}} \approx_\eps p_{\cA} \sigma_{M \hat{M}}+(1-p_{\cA}) \bot_{M'}\otimes \sigma_{\hat{M}},
\end{equation}
where $p_{\cA}\in[0,1]$ depend only on the tampering adversary $\cA$ and $ \bot_{M'}$ denotes a special abort symbol stored in register $M'$ to denote tamper detection.
\end{definition}

\newpage
\section{Tamper Detection in the Bounded Storage Model}
\label{section:bounded-storage}

In this section, we present and analyze our quantum tamper detection code in the $3$-split-state model $\lo^3_{(e_1, e_2, e_3)}$, where the adversaries are allowed a finite pre-shared quantum state to assist in tampering with their shares of the code (see \Cref{def:lotsharedentanglement}). In particular, we show security even when $e_1, e_2$ are both unbounded, that is, only the entanglement of one of the parties is restricted.

\begin{theorem}[\cref{theorem:results-main}, restatement]
    For every integer $k$, $a\leq \lambda$, there exists an efficient tamper-detection code against $\lo^3_{(a, *, *)}$ for $k$ qubit messages with error $2^{-(k+\lambda)^{\Omega(1)}} + 6\cdot 2^{a-\lambda}$ and rate $(11 + 12\frac{\lambda}{k})^{-1}$.
\end{theorem}

Our construction combines quantum secure non-malleable extractors and families of unitary 2 designs. We refer the reader to \Cref{subsubsection:prelim-nmext} and \Cref{subsection:prelim-2-designs} respectively for formal definitions of the ingredients. We dedicate \cref{subsection:tdc-components} for a description of our construction, and \cref{subsection:tdc-analysis} for its analysis.

\subsection{Code Construction}
\label{subsection:tdc-components}

\textbf{Ingredients} Let $\delta, \delta'>0$ be constants, $k = (1/2 -\delta -5\delta')n/5$, and $\lambda = \delta' \cdot n$. We combine

\begin{enumerate}
    \item $\nmext:\{0, 1\}^n\times \{0, 1\}^{\delta \cdot n}\rightarrow \{0, 1\}^{(1/2-\delta)n}$, a quantum secure two source non-malleable extractor with error $\epsilon_\nmext =2^{-n^{\Omega_\delta(1)}} $ from \cref{lem:qnmcodesfromnmext}.\\

    \item  The family of $2$-design unitaries $C_R$ from \cref{lem:subclifford}.

\end{enumerate}

\noindent\textbf{Our candidate TDC} Our candidate construction of quantum TDC against $\lo^3_{(e_1 = \infty,e_2 =\infty,e_3)}$ is given in \cref{fig:3splitstatetdc} along with the tampering process. We describe it explicitly in the following \cref{alg:algorithm3qtdc}. The decoding scheme, denoted as $\dec$, for the quantum TDC operates as in \Cref{alg:algorithm3qtdc_dec}.

\begin{algorithm}[H]
    \setstretch{1.35}
    \caption{$\Enc$: Quantum TDC against $\lo^3_{(e_1=\infty,e_2=\infty,e_3)}$ (see \cref{fig:3splitstatetdc}).}
    \label{alg:algorithm3qtdc}
    \KwInput{ A $k$ qubit quantum message $\sigma_{M}$ (with purification $\hat{M}$). }
    %\KwOutput{The product state $\rho^m_L\otimes \rho^r_R$ for some $m\in \{0, 1\}$ correlated with $b$}

    \begin{algorithmic}[1]

    \State Sample classical registers $X, Y$ uniformly and independently of size $n, \delta n$ respectively. Evaluate $R= \nmext(X,Y)$. 
    
    \State Prepare $\lambda =\delta' n$ EPR pairs, $\Phi^{\otimes \lambda}$, on a bipartite register $E, \hat{E}$.

    \State Consider $X$ as the first share and $(Y,E)$ as the second share.
    
    \State Let $C_R$ be the Clifford unitary picked using sampling process $\samp$ in \cref{lem:subclifford}. Apply $C_R$ on registers $(\hat{E},M)$ to generate $Z$. This is possible since $\vert R \vert = (1/2-\delta)n=5(k+\delta'n)$ from our choice of parameters. Consider $Z$ as the third share. 

    \item Output shares $X, (Y, {E}), Z$.
    \end{algorithmic}
\end{algorithm}

\begin{algorithm}[H]
    \setstretch{1.35}
    \caption{$\Dec$: }
    \label{alg:algorithm3qtdc_dec}
    \KwInput{ Any tampered code-state $\Lambda\circ \Enc(\sigma_{M})$ on 3 shares $(X, YE, Z)$.}
    %\KwOutput{The product state $\rho^m_L\otimes \rho^r_R$ for some $m\in \{0, 1\}$ correlated with $b$}

    \begin{algorithmic}[1]

    \State The decoder first computes $R'= \nmext(X',Y')$.
    
    \State Subsequently $C^\dagger_{R'}$ is applied on register $Z'$ to obtain $\hat{E}' , M'$. 

    \State Perform an $\epr$ test on registers $E'\hat{E}'$. This entails a binary outcome projective measurement $({\Pi = \Phi_{E\hat{E}}, \Bar{\Pi} = \id_{E \hat{E}} - \Phi_{E\hat{E}}})$.

    \item An acceptance decision is made if the outcome is $\Pi = \Phi_{E\hat{E}}$. In this case, the ``replace" operation $\swap$ outputs the register $M'$.

    \item If the outcome is not $\Pi$, the operation $\swap$ outputs a special symbol $\bot$ in register $M'$.

    \end{algorithmic}

\end{algorithm}

\subsection{Analysis}
\label{subsection:tdc-analysis}

We represent the ``tampering experiment" in \Cref{fig:3splitstatetdc}, where a message is encoded into $\Enc$, tampered with a channel in the bounded storage model, and then decoded.

\begin{figure}[h]
\centering
\resizebox{12cm}{6cm}{
\begin{tikzpicture}
%\draw (1,1.2) ellipse (0.3cm and 2.5cm);
  %        \node at (1,3.3) {$R$};
  %          \node at (5.2,3.7) {$R$};
   %      \node at (14.5,3.7) {$R$};
   %      \draw (1.2,3.5) -- (3,3.5);
   %       \draw (5,3.5) -- (15,3.5);
%\draw (10,3.5) -- (15,3.5);
%\node at (1,-0.8) {$\hat{R}$};
%\node at (14.5,-0.8) {$\hat{R}$};
%\draw (1.2,-1) -- (15,-1);

%\draw (1,5) ellipse (0.2cm and 1cm);
%\node at (1,5.5) {$\hat{S}$};
%\node at (14.5,5.5) {$\hat{S}$};
%\draw (1.2,5.7) -- (15,5.7);
\draw (0,7.4) ellipse (0.2cm and 1cm);
\draw (1,5.5) ellipse (0.3cm and 1.2cm);
\draw (3,5.9) rectangle (4.4,7.4);
\draw (11.6,5.4) rectangle (12.5,7.4);
\node at (0,6.9) {$M$};
\draw (0.2,6.9) -- (3,6.9);
\node at (1,6.4) {$\hat{E}$};
\node at (12.7,6.25) {$\hat{E}'$};
\node at (1,4.7) {$E$};
\draw (1.2,4.5) -- (4.5,4.5);
\node at (16.2,7.1) {$M'$};

%\draw (0.8,0.5) ellipse (0.3cm and 1.5cm);
%\node at (0.8,-0.7) {$\hat{R}$};
\node at (3.6,6.6) {$C_R$};
\node at (12.1,6.1) {$C^\dagger_{R'}$};
\node at (0,8) {$\hat{M}$};
\node at (16.2,8) {$\hat{M}$};

\draw (0.2,8) -- (16,8);
\draw (1.2,6.5) -- (3,6.5);
\draw (4.4,6.5) -- (6.3,6.5);
\draw (7.3,6.5) -- (11.6,6.5);
\draw (12.5,6.5) -- (13,6.5);
\draw (12.5,7.1) -- (14.2,7.1);

%\node at (14.5,4.3) {$S'$};
%\draw (0.2,4.5) -- (3,5.5);
\draw (4.4,4.5) -- (6.3,4.5);
\draw (7.3,4.5) -- (13,4.5);
%\draw (12.5,4.5) -- (15.5,4.5);
%\draw (12.5,5) -- (13,5);
\node at (12.7,4.7) {$E'$};
\draw (13,4.2) rectangle (14,6.8);
\node at (14.8,6.8) {$\mathsf{Rep}$};
\draw (14.2,6) rectangle (15.4,7.4);
\draw (15.4,6.6) -- (16,6.6);

\node at (13.5,5.4) {$\mathsf{EPR}$};
\draw (14,5.6) -- (14.6,5.6);
\draw (14.6,6) -- (14.6,5.6);
\node at (14.8,5.4) {$O$};
%\draw (10,4.5) -- (15,4.5);
%\node at (5.2,4.7) {$S$};
%\node at (15.3,4.7) {$M'$};
\node at (0.8,2) {$Y$};
\node at (0.8,1) {$X$};
%\node at (1.5,1.5) {$S$};
\node at (3.4,1.2) {$R$};
\draw (1,1) -- (1.8,1);
\draw (1,2) -- (1.8,2);
\draw (3.2,1.5) -- (4,1.5);
\draw (4,1.5) -- (4,5.9);
\draw (1.2,0.2) -- (4.6,0.2);
\draw (1.2,2.8) -- (4.6,2.8);
\draw (1.2,0.2) -- (1.2,1);
\draw (1.2,2) -- (1.2,2.8);
%\node at (2.0,0) {$\ketbra{0^{n}}$};
%\draw (2.8,0) -- (3,0);
\draw (1.8,0.5) rectangle (3.2,2.5);
\draw (1.6,-0.5) rectangle (4.8,7.6);
\node at (3.3,7.8) {$\enc$};
\node at (2.5,1.5) {$\nmcenc$};

\draw (9.4,-0.5) rectangle (15.6,7.6);
\node at (11.5,7.8) {$\dec$};

\node at (4.6,6.3) {$Z$};
%\draw  (2.2,1.5) -- (2.2,3.9);

\draw [dashed] (1.4,-1.4) -- (1.4,8.5);
\draw [dashed] (4.98,-1.4) -- (4.98,8.5);
\draw [dashed] (9.1,-1.4) -- (9.1,8.5);
%\draw [dashed] (10,-1.4) -- (10,7.2);
\draw [dashed] (15.8,-1.4) -- (15.8,8.5);
%\draw [dashed] (14.3,-0.8) -- (14.3,5.5);

%\draw (0.9,-1) -- (15,-1);

\node at (6.9,1.6) {$\ket{\psi}_{W_1W_2W_3}$};
\node at (1.25,-1.4) {$\sigma$};
\node at (4.8,-1.4) {$\sigma_1$};
\node at (8.9,-1.4) {$\sigma_2$};
%\node at (10.2,-1.4) {$\sigma_3$};
\node at (15.6,-1.4) {$\eta$};
%\node at (14.13,-0.8) {$\rho$};

\node at (4.6,4.7) {$E$};
\node at (8.5,4.7) {$E'$};
\node at (4.6,3) {$Y$};
\node at (8.5,3) {$Y'$};
\node at (8.5,6.2) {$Z'$};
\draw (4.5,2.8) -- (6.3,2.8);
%\draw (7,2.8) -- (8.5,2.8);
%\draw (9.5,2.8) -- (10.5,2.8);
\draw (7.3,2.8) -- (10,2.8);

\node at (11.7,2.6) {$R'$};
\draw (11.5,2.4) -- (12,2.4);
\draw (12,2.4) -- (12,5.4);

\node at (4.6,0.4) {$X$};
\node at (8.5,0.0) {$X'$};
\draw (4.5,0.2) -- (6.3,0.2);
%\draw (7,0.2) -- (8.5,0.2);
%\draw (9.5,0.2) -- (10.5,0.2);
\draw (7.3,0.2) -- (10,0.2);

\draw (6.3,5.5) rectangle (7.3,7);
\node at (6.8,6.2) {$T$};
\draw (6.3,2) rectangle (7.3,5);
\node at (6.8,3.5) {$V$};
\draw (6.3,0) rectangle (7.3,1);
\node at (6.8,0.5) {$U$};

\node at (6.5,-0.4) {$\mathcal{A}=(U,V,T,\psi)$};
\draw (5.2,-0.8) rectangle (8.2,7.3);

\draw (5.8,3.1) ellipse (0.3cm and 2.8cm);
\node at (5.8,5.3) {$W_3$};
\draw (5.9,5.7) -- (6.3,5.7);
\draw (7.3,5.7) -- (7.45,5.7);
\node at (7.7,5.7) {$W'_3$};
\node at (5.8,2) {$W_2$};
\draw (6,2.2) -- (6.3,2.2);
\node at (7.7,2) {$W'_2$};
\draw (7.3,2.2) -- (7.5,2.2);
\node at (5.8,1) {$W_1$};
\draw (6,0.7) -- (6.3,0.7);
\node at (7.7,0.9) {$W'_1$};
\draw (7.3,0.7) -- (7.45,0.7);

%\draw (9,2.8) circle (0.5);
%\node at (9,2.8) {$\mathcal{M}$};
%\draw (9,0.2) circle (0.5);
%\node at (9,0.2) {$\mathcal{M}$};

%\draw (8,-0.5) rectangle (10,1.2);
%\draw (8,1.5) rectangle (10,3.2);
%\node at (9,1.5) {$2\mhyphen\nmext$};

\draw (10,-0.2) rectangle (11.5,3.5);
\node at (10.8,1.5) {${\nmcdec}$};
%\node at (12.5,1.6) {$ \mathcal{M}= \{ \ketbra{0^{2n}},$};
%\node at (12.5,1) {$ \id- \ketbra{0^{2n}} \}$};

\end{tikzpicture}}

    \caption{
       A quantum TDC against $\lo^3_{(e_1=\infty, e_2=\infty, e_3)}$ along with  tampering process.
        %Split-state NMC for quantum messages.
    }\label{fig:3splitstatetdc}
\end{figure}
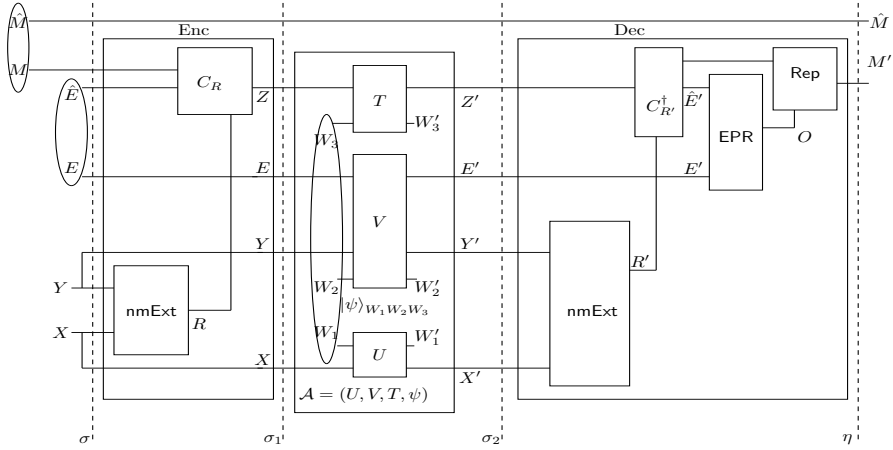

In this section, we prove two key results about our code construction. Our main result is that the code described above is a secure tamper-detection code in the bounded storage model:

\begin{theorem}[Tamper Detection]\label{theorem:tdc-loa}
Consider \cref{fig:3splitstatetdc}.  $(\Enc, \Dec)$ as shown in \cref{fig:3splitstatetdc} is a quantum TDC against $\lo^3_{(\infty, \infty, e_3)}$ with error $\eps' \leq 2 \eps_{\nmext}  +6\cdot 2^{e_3- \lambda}$.
\end{theorem}

The proof of which we defer to \Cref{subsubsection:proof-of-tdc-loa}. By instantiating \cref{theorem:tdc-loa} with the ingredients above, we obtain the following explicit construction. 
\begin{corollary}
    For any constant $0 < \gamma < 1/12$, there exists an efficient quantum TDC of blocklength $n$ and rate $\frac{1-12\cdot \gamma}{11}$ secure against $\lo^3_{(\infty, \infty, \gamma \cdot n)}$ with error $2^{-n^{\Omega(1)}}$.
\end{corollary}

We further prove that our construction inherits quite strong secret sharing properties, which we extensively leverage in our future applications to ramp secret sharing. 

\begin{theorem}[3-out-of-3 Secret Sharing]\label{corr:2nmssq}
$(\enc,\dec)$ from \cref{theorem:tdc-loa} is also a $3$-out-of-$3$ secret sharing scheme with error $\epsilon'$. In fact, any two shares of the quantum TDC are $\epsilon'$-close to the maximally mixed state. 
\end{theorem}
\begin{proof}

Note that $\enc(.)$ in \cref{fig:3splitstatetdc} first samples an independent $(X,Y)$ and then generates $R= \nmext(X,Y)$. It also independently prepares $\lambda = \delta' n $ EPR pairs $\Phi_{E \hat{E}}$. It follows from the strong-extraction property of $\nmext$ (see \cref{lem:qnmcodesfromnmext}) that
\[ RX \approx_{\eps_\nmext} U_R \otimes U_X \quad ; \quad RY \approx_{\eps_\nmext} U_R \otimes U_Y .\]

Recall that the three shares are $(X, YE, Z = C_R(\hat{E} M) C^\dagger_R)$. Thus for every message ${\sigma_{M\hat{M}}}$, using the fact that Cliffords are 1-Designs (\cref{fact:notequal}), and the approximate sampler $\samp$ from \cref{lem:subclifford}, we have
\begin{gather}\label{eqfd}
     (\enc( \sigma))_{\hat{M}ZX}  \approx_{\eps_\nmext+2^{1-\lambda}} \sigma_{\hat{M}} \otimes  U_{Z} \otimes U_X \quad ;\\ \quad  (\enc( \sigma))_{\hat{M}ZYE}  \approx_{\eps_\nmext+2^{1-\lambda}} \sigma_{\hat{M}} \otimes  U_{Z} \otimes U_Y \otimes U_E. 
\end{gather}
Since $|Z|>\lambda=\delta' n $. Moreover, since $(X,Y)$ are sampled independently, we also have
\begin{equation}\label{eqfdsfvf}
    (\enc( \sigma))_{\hat{M}XYE}  =\sigma_{\hat{M}} \otimes  U_{X} \otimes U_Y \otimes U_E . 
\end{equation}
\qed \end{proof}

\subsubsection{Proof of \Cref{theorem:tdc-loa}}
\label{subsubsection:proof-of-tdc-loa}

To show that $(\enc,\dec)$ is an $\eps'$-secure quantum TDC, it suffices to show that for every $\mathcal{A}=(U,V,T,\psi_{W_1W_2W_3})$ it holds that (in \cref{fig:3splitstatetdc})
\begin{equation}\label{eq:finalgoal}
    \eta_{\hat{M}M'} \approx_{\eps'} p_{\mathcal{A}} \sigma_{\hat{M}M}  + (1-p_\mathcal{A}) (\sigma_{\hat{M}} \otimes  {\bot}_{M'}),
\end{equation}where $p_{\mathcal{A}}$ depends only on the tampering adversary $\cA$. In \cref{fig:3splitstatetdc1} below, we represent the same tampering experiment as \cref{fig:3splitstatetdc}, except for the the delayed action of the tampering map $T$. We show that \cref{eq:finalgoal} holds in \cref{fig:3splitstatetdc1}, which completes the proof.

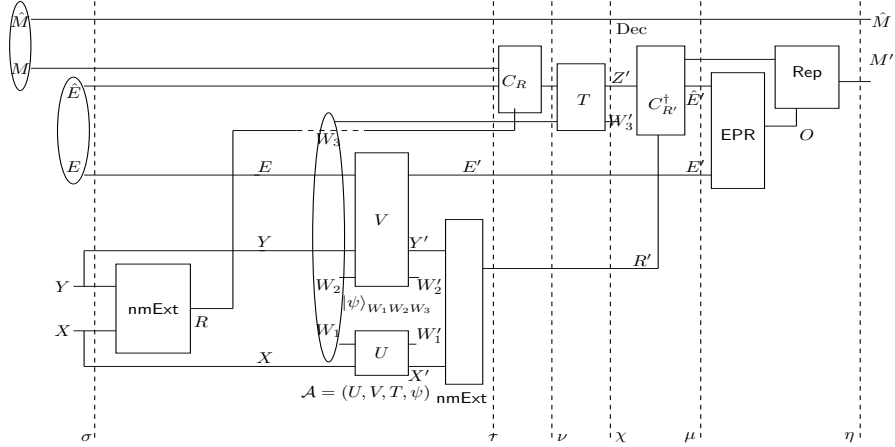
\begin{figure}[H]
\centering
\resizebox{12cm}{6cm}{
\begin{tikzpicture}
%\draw (1,1.2) ellipse (0.3cm and 2.5cm);
  %        \node at (1,3.3) {$R$};
  %          \node at (5.2,3.7) {$R$};
   %      \node at (14.5,3.7) {$R$};
   %      \draw (1.2,3.5) -- (3,3.5);
   %       \draw (5,3.5) -- (15,3.5);
%\draw (10,3.5) -- (15,3.5);
%\node at (1,-0.8) {$\hat{R}$};
%\node at (14.5,-0.8) {$\hat{R}$};
%\draw (1.2,-1) -- (15,-1);

%\draw (1,5) ellipse (0.2cm and 1cm);
%\node at (1,5.5) {$\hat{S}$};
%\node at (14.5,5.5) {$\hat{S}$};
%\draw (1.2,5.7) -- (15,5.7);
\draw (0,7.4) ellipse (0.2cm and 1cm);
\draw (1,5.5) ellipse (0.3cm and 1.2cm);
\draw (9,5.9) rectangle (9.8,7.4);
\draw (11.6,5.4) rectangle (12.5,7.4);
\node at (0,6.9) {$M$};
\draw (0.2,6.9) -- (9,6.9);
\node at (1,6.4) {$\hat{E}$};
\node at (12.7,6.25) {$\hat{E}'$};
\node at (1,4.7) {$E$};
\draw (1.2,4.5) -- (4.5,4.5);
\node at (16.2,7.1) {$M'$};

%\draw (0.8,0.5) ellipse (0.3cm and 1.5cm);
%\node at (0.8,-0.7) {$\hat{R}$};
\node at (9.3,6.6) {$C_R$};
\node at (12.1,6.1) {$C^\dagger_{R'}$};
\node at (0,8) {$\hat{M}$};
\node at (16.2,8) {$\hat{M}$};

\node at (10.6,6.2) {$T$};

\draw (0.2,8) -- (16,8);
\draw (1.2,6.5) -- (9,6.5);
%\draw (4.4,6.5) -- (6.3,6.5);
\draw (11,6.5) -- (11.6,6.5);
\draw (9.8,6.5) -- (10.1,6.5);
\draw (12.5,6.5) -- (13,6.5);
\draw (12.5,7.1) -- (14.2,7.1);

%\node at (14.5,4.3) {$S'$};
%\draw (0.2,4.5) -- (3,5.5);
\draw (4.4,4.5) -- (6.3,4.5);
\draw (7.3,4.5) -- (13,4.5);
%\draw (12.5,4.5) -- (15.5,4.5);
%\draw (12.5,5) -- (13,5);
\node at (12.7,4.7) {$E'$};
\draw (13,4.2) rectangle (14,6.8);
\node at (14.8,6.8) {$\mathsf{Rep}$};
\draw (14.2,6) rectangle (15.4,7.4);
\draw (15.4,6.6) -- (16,6.6);

\node at (13.5,5.4) {$\mathsf{EPR}$};
\draw (14,5.6) -- (14.6,5.6);
\draw (14.6,6) -- (14.6,5.6);
\node at (14.8,5.4) {$O$};
%\draw (10,4.5) -- (15,4.5);
%\node at (5.2,4.7) {$S$};
%\node at (15.3,4.7) {$M'$};
\node at (0.8,2) {$Y$};
\node at (0.8,1) {$X$};
%\node at (1.5,1.5) {$S$};
\node at (3.4,1.2) {$R$};
\draw (1,1) -- (1.8,1);
\draw (1,2) -- (1.8,2);
\draw (3.2,1.5) -- (4,1.5);
\draw (4,1.5) -- (4,5.5);
\draw (4,5.5) -- (5.2,5.5);
\draw [dashed] (5.2,5.5) -- (6.5,5.5);
\draw (6.5,5.5) -- (9.3,5.5);
\draw (9.3,6) -- (9.3,5.5);

\draw (1.2,0.2) -- (4.6,0.2);
\draw (1.2,2.8) -- (4.6,2.8);
\draw (1.2,0.2) -- (1.2,1);
\draw (1.2,2) -- (1.2,2.8);
%\node at (2.0,0) {$\ketbra{0^{n}}$};
%\draw (2.8,0) -- (3,0);
\draw (1.8,0.5) rectangle (3.2,2.5);
%\draw (1.6,-0.5) rectangle (4.8,7.6);
%\node at (3.3,7.8) {$\enc$};
\node at (2.5,1.5) {$\nmcenc$};

%\draw (9.4,-0.5) rectangle (15.6,7.6);
\node at (11.5,7.8) {$\dec$};

%\node at (4.6,6.3) {$Z$};
%\draw  (2.2,1.5) -- (2.2,3.9);

\draw [dashed] (1.4,-1.4) -- (1.4,8.5);
%\draw [dashed] (4.98,-1.4) -- (4.98,8.5);
\draw [dashed] (8.9,-1.4) -- (8.9,8.5);
\draw [dashed] (10,-1.4) -- (10,8.5);
\draw [dashed] (11.1,-1.4) -- (11.1,8.5);
\draw [dashed] (12.8,-1.4) -- (12.8,8.5);
\draw [dashed] (15.8,-1.4) -- (15.8,8.5);
%\draw [dashed] (14.3,-0.8) -- (14.3,5.5);

%\draw (0.9,-1) -- (15,-1);

\node at (6.9,1.6) {$\ket{\psi}_{W_1W_2W_3}$};
\node at (1.25,-1.4) {$\sigma$};
%\node at (4.8,-1.4) {$\sigma_1$};
\node at (8.9,-1.4) {$\tau$};
\node at (10.2,-1.4) {$\nu$};
\node at (11.3,-1.4) {$\chi$};
\node at (12.6,-1.4) {$\mu$};
\node at (15.6,-1.4) {$\eta$};
%\node at (14.13,-0.8) {$\rho$};

\node at (4.6,4.7) {$E$};
\node at (8.5,4.7) {$E'$};
\node at (4.6,3) {$Y$};
\node at (7.5,3) {$Y'$};
\node at (11.3,6.7) {$Z'$};
\draw (4.5,2.8) -- (6.3,2.8);
%\draw (7,2.8) -- (8.5,2.8);
%\draw (9.5,2.8) -- (10.5,2.8);
\draw (7.3,2.8) -- (8,2.8);

\node at (11.7,2.6) {$R'$};
\draw (8.7,2.4) -- (12,2.4);
\draw (12,2.4) -- (12,5.4);

\node at (4.6,0.4) {$X$};
\node at (7.5,0.0) {$X'$};
\draw (4.5,0.2) -- (6.3,0.2);
%\draw (7,0.2) -- (8.5,0.2);
%\draw (9.5,0.2) -- (10.5,0.2);
\draw (7.3,0.2) -- (8,0.2);

\draw (10.1,5.5) rectangle (11,7);
%\node at (6.8,6.2) {$T$};
\draw (6.3,2) rectangle (7.3,5);
\node at (6.8,3.5) {$V$};
\draw (6.3,0) rectangle (7.3,1);
\node at (6.8,0.5) {$U$};

\node at (6.5,-0.4) {$\mathcal{A}=(U,V,T,\psi)$};
%\draw (5.2,-0.8) rectangle (8.2,7.3);

\draw (5.8,3.1) ellipse (0.3cm and 2.8cm);
\node at (5.8,5.3) {$W_3$};
\draw (5.9,5.7) -- (10.1,5.7);
\draw (11,5.7) -- (11.25,5.7);
\node at (11.3,5.7) {$W'_3$};
\node at (5.8,2) {$W_2$};
\draw (6,2.2) -- (6.3,2.2);
\node at (7.7,2) {$W'_2$};
\draw (7.3,2.2) -- (7.5,2.2);
\node at (5.8,1) {$W_1$};
\draw (6,0.7) -- (6.3,0.7);
\node at (7.7,0.9) {$W'_1$};
\draw (7.3,0.7) -- (7.45,0.7);

%\draw (9,2.8) circle (0.5);
%\node at (9,2.8) {$\mathcal{M}$};
%\draw (9,0.2) circle (0.5);
%\node at (9,0.2) {$\mathcal{M}$};

%\draw (8,-0.5) rectangle (10,1.2);
%\draw (8,1.5) rectangle (10,3.2);
%\node at (9,1.5) {$2\mhyphen\nmext$};

\draw (8,-0.2) rectangle (8.7,3.5);
\node at (8.3,-0.5) {${\nmcdec}$};
%\node at (12.5,1.6) {$ \mathcal{M}= \{ \ketbra{0^{2n}},$};
%\node at (12.5,1) {$ \id- \ketbra{0^{2n}} \}$};

\end{tikzpicture}}

    \caption{
       Analysis of quantum TDC against $\lo^3_{(e_1=\infty, e_2=\infty, e_3)}$.
        %Split-state NMC for quantum messages.
    }\label{fig:3splitstatetdc1}
\end{figure}

\begin{proof}
    Consider the state ${\tau}$ in~\cref{fig:3splitstatetdc1}. Note ${\tau}_{\hat{M}M} \equiv \sigma_{\hat{M}M}$ is a pure state (thus independent of other registers in ${\tau}$) and 
\[{\tau} = (\nmext_{X'Y'} \otimes \nmext_{XY})  \left( (U \otimes V)  (\sigma \otimes \ketbra{\psi}_{W_1W_2W_3}) \right). \]

Our analysis will proceed by cases, depending on whether the $X$ and $Y$ registers are modified by the tampering experiment in \cref{fig:3splitstatetdc1} (i.e., $XY\neq X' Y'$) or not.
To this end, we consider two different conditionings of $\tau$ based on these two cases. Let $\tau^1$ be the state if the tampering adversary ensures $(XY=X'Y')$ and $\tau^0$ be the state conditioned on $(XY \ne X'Y')$. 

Using \cref{lem:qnmcodesfromnmext}, state $\tau$ can be written as convex combination of two states $\tau^1$ and $\tau^0$ such that:
\begin{equation}\label{eq:2classical198}
     (\tau)_{{R}R'W_2'W_3E'\hat{E}M\hat{M}}= p_{\sm} (\tau^1)_{{R}R'W_2'W_3E' \hat{E}M\hat{M}} + (1-p_{\sm}) (\tau^0)_{{R}R'W_2'W_3E' \hat{E}M\hat{M}},
    \end{equation}
where $p_\sm$ depends only on the adversary $\mathcal{A}'=(U,V,\psi_{W_1W_2W_3})$. In case $(XY=X'Y')$, the key is recovered $\Pr({R} =R')_{\tau^1}=1$. \cref{lem:qnmcodesfromnmext} guarantees the non-malleability of the secret key $R$:
\begin{multline}\label{eq:2classical1}
        p_{\sm} \Vert (\tau^1)_{{R}W_2'W_3E' \hat{E}M\hat{M}} -  U_{\vert R \vert} \otimes (\tau^1)_{W_2'W_3E' \hat{E}M\hat{M}} \Vert_1 + \\ (1-p_{\sm})\Vert (\tau^0)_{{R}R'W_2'W_3E' \hat{E}M\hat{M}} -  U_{\vert R \vert} \otimes (\tau^0)_{R'W_2'W_3E' \hat{E}M\hat{M}} ) \Vert_1  \leq \eps_{\nmext}.
    \end{multline}

    Suppose $\Upsilon$ denotes the CPTP map from registers ${R}R'W_3E' \hat{E}M\hat{M}$ to $M'\hat{M}$ (i.e. $\Upsilon$ maps state $\tau$ to $\eta$) in~\cref{fig:3splitstatetdc1}. We present two lemmas (proved in the next subsections) which allow us to conclude the proof. The first one stipulates that if the key is not recovered, the Bell basis measurement rejects with high probability:
 
     \begin{lemma}[Key \textit{not} recovered]\label{lemma:dfskjboca}
   $ \Upsilon(U_{\vert R \vert} \otimes (\tau^0)_{R'W_2'W_3E' \hat{E}M\hat{M}}) 
 \approx_{2 \cdot  2^{ \vert W_3 \vert - \vert E \vert  }} \sigma_{\hat {M}}  \otimes \bot_{M'}$ . 
 \end{lemma}
    
    In the second, if the key is recovered, then the 2-design functions as an authentication code. In this manner, we either recover the original message, or reject:

    \begin{lemma}[Key recovered]\label{lemma:nekjbs2}
   $ \Upsilon(U_{\vert R \vert} \otimes (\tau^1)_{W_2'W_3E' \hat{E}M\hat{M}}) 
 \approx_{ \frac{2}{(4^{\vert E \vert + \vert M \vert}-1)} + \frac{2}{2^{\vert E \vert}} }  p\sigma_{\hat {M}M}  + (1-p) \sigma_{\hat {M}}  \otimes \bot_{M'}$ . Furthermore, $p$ depends only on adversary $\mathcal{A} = (U,V,T, \ket{\psi}_{W_1W_2W_3})$.
 \end{lemma}

    We are now in a position to conclude the proof. First, leveraging \cref{eq:2classical1}
\begin{gather}
  \eta_{M'\hat{M}} =\Upsilon( (\tau)_{{R}R'W_3E' \hat{E}M\hat{M}})\nonumber = p_{\sm}  \Upsilon((\tau^1)_{{R}R'W_3E' \hat{E}M\hat{M}}) 
    + (1-p_{\sm})  \Upsilon((\tau^0)_{{R}R'W_3E' \hat{E}M\hat{M}}) \\
    \approx_{\eps_{\nmext}} p_{\sm}  \Upsilon(U_{\vert R \vert} \otimes (\tau^1)_{W_2'W_3E' \hat{E}M\hat{M}} ) + (1-p_{\sm}) \Upsilon(U_{\vert R \vert} \otimes (\tau^0)_{R'W_2'W_3E' \hat{E}M\hat{M}} ) 
\end{gather}
Next, by applying \cref{lemma:nekjbs2}, proceeded by \cref{lemma:dfskjboca}:
\begin{gather}
     \approx_{ \frac{2}{(4^{\vert E \vert + \vert M \vert}-1)} + \frac{2}{2^{\vert E \vert}} } p_{\sm}\cdot (p\cdot \sigma_{M\hat{M}}+  (1-p) (\bot_{{M'}} \otimes \sigma_{\hat{M}}))+ (1-p_{\sm}) \Upsilon(U_{\vert R \vert} \otimes (\tau^0)_{R'W_2'W_3E' \hat{E}M\hat{M}} )\\
     \approx_{2 \cdot \frac{2^{\vert W_3 \vert}}{ 2^{\vert E \vert}} } p_{\sm}\cdot p\cdot \sigma_{M\hat{M}}+  p_{\sm}  (1-p) (\bot_{{M'}} \otimes \sigma_{\hat{M}})  + (1-p_{\sm}) \bot_{M'} \otimes \sigma_{\hat{M}}
\end{gather}

The total error is thus $\leq \eps_\nmext + 6\cdot 2^{|W_3|-|E|}$. The observation that $p_\sm$ and $p$ depend only on the adversary $\mathcal{A} = (U,V,T ,\ket{\psi}_{W_1W_2W_3})$ completes the proof.

\qed\end{proof}

To prove the remaining lemmas above, we require a short statement on the Bell basis measurements of states of bounded Schmidt rank. 

\begin{lemma}\label{lemma:tss-bellbasis} Let 
$ \Phi_{E \hat{E}} \equiv \Phi^{\otimes \lambda} $ denote 
 $\lambda$ EPR pairs on a bipartite register $E, \hat{E}$. Let $\tau_{E\hat{E}}$ be a pure state of Schmidt rank $R$. Then, measuring $\tau$ with the binary measurement $\{ \Phi_{E\hat{E}} , \id_{E\hat{E}}  -\Phi_{E\hat{E}}  \}$ outputs the $\Phi_{E\hat{E}}$ with negligible probability, i.e., 
    \begin{equation*}
       \Tr[\Phi^{\otimes \lambda}\tau]\leq R\cdot 2^{- \lambda}.
    \end{equation*}

\noindent Similarly, if $\tau_{E\hat{E}}$ were a mixed state with Schmidt number $R$, then $\Tr[\Phi^{\otimes \lambda}\tau]\leq R\cdot 2^{- \lambda}.$
\end{lemma}

\begin{proof}

    [of \Cref{lemma:tss-bellbasis}] Let us first consider a pure state $\ket{\phi}$ of rank $R$, and let $\ket{\phi} = \sum_{i=1}^R \alpha_i \ket{u_i}_E\otimes \ket{v_i}_{\hat{E}}$, $ \Phi_{E \hat{E}} \equiv \ket{\Phi}^{\otimes \lambda} = 2^{-\lambda/2} \sum_j  \ket{j}_E\otimes \ket{j}_{\hat{E}}$ define Schmidt decompositions. Then, by the triangle and the Cauchy-Schwartz inequalities:
    \begin{gather}
        \Tr[\Phi^{\otimes \lambda}\phi]^{1/2} = |\bra{\Phi}^{\otimes \lambda}\ket{\phi}| \leq \sum_i^R |\alpha_i|\cdot |\bra{\Phi}^{\otimes \lambda}\ket{u_i}\otimes \ket{v_i}|\leq \\ \leq \max_{\ket{u}\otimes \ket{v}} |\bra{\Phi}^{\otimes \lambda}\ket{u}\otimes \ket{v}| \cdot R^{1/2}\cdot \bigg(\sum_i |\alpha_i|^2\bigg)^{1/2} \leq \\ \leq R^{1/2}\cdot \max_{\ket{u}\otimes \ket{v}} |\bra{\Phi}^{\otimes \lambda}\ket{u}\otimes \ket{v}|
    \end{gather}
    \noindent In turn, 
    \begin{gather}
        |\bra{\Phi}^{\otimes \lambda}\ket{u}\otimes \ket{v}| =  2^{-\lambda/2}\cdot \bigg|\sum_j \bra{j}\ket{u}\cdot  \bra{j}\ket{v}\bigg| \leq \\ \leq  2^{-\lambda/2}\bigg(\sum_j \cdot |\bra{j}\ket{u}|^2\bigg)^{1/2} \bigg(\sum_j  |\bra{j}\ket{v}|^2\bigg)^{1/2} \leq 2^{-\lambda/2},
    \end{gather}
    \noindent which gives us the desired bound. If $\tau$ is a mixed state of Schmidt number $R$, then it can be written as a convex combination of pure states of Schmidt rank $\leq R$, which concludes the lemma. \qed \end{proof}

 \begin{proof}
 
 [of \cref{lemma:dfskjboca}]
Let $\nu^0, \chi^0, \mu^0$ be the intermediate states and $\eta^0$ be the final state when we run the CPTP map $\Upsilon$ on $U_{\vert R \vert} \otimes \tau^0_{R'W_2'W_3E' \hat{E}M\hat{M}}$ (see \Cref{fig:3splitstatetdc1}). Since, $\tau^0_{R'W_2'W_3E' \hat{E}M\hat{M}} = \tau^0_{R'W_2'W_3E' \hat{E}}  \otimes \sigma_{M\hat{M}}$, using~\cref{fact:notequal} (Cliffords are 1-Designs) it follows that in the state $\nu^0$ the two registers $E$ and $\hat{E}$ are decoupled:

\begin{equation}
    \nu^0_{R'W_2'W_3E' \hat{M} \hat{E}M} = \tau^0_{R'W_2'W_3E'  } \otimes \sigma_{\hat{M}} \otimes U_{\hat{E}M}
\end{equation}

We fix $R'=r'$ and argue that we output $\bot_{M'}$ with high probability for every such fixing. Let 
  $\tau^{0,r'} \defeq \tau^{0} \vert (R'=r')$ and similarly define $\nu^{0,r'} , \chi^{0,r'} $, $\mu^{0,r'} $ and $\eta^{0,r'}$. Note $\nu^{0,r'}_{W_2'W_3E' \hat{M} \hat{E}M} = \tau^{0,r'}_{W_2'W_3E'  } \otimes \sigma_{\hat{M}} \otimes U_{\hat{E}M} $. This implies that the Schmidt number of the state $\nu^{0,r'}$ across the bipartition $(W_2'E', W_3\hat{E}M \hat{M})$ is $\leq 2^{\vert W_3 \vert}$. 

  The states $\chi^{0, r'}, \mu^{0, r'}$ can be prepared from $\nu^{0,r'}$ using just local operations on each side of the cut $(W_2'E' , W_3\hat{E}M \hat{M})$. From \cref{prop:schimidt}, their Schmidt numbers are at most $2^{\vert W_3 \vert}$. Moreover, again by \cref{prop:schimidt}, the reduced density matrix $\mu^{0,r'}_{E'\hat{E}}$ also has Schmidt number at most $2^{\vert W_3 \vert}$. \Cref{lemma:tss-bellbasis} then ensures that the $\epr$ test on state $\mu^{0,r'}$ fails with probability at least $1-2^{\vert W_3 \vert-\vert E  \vert}$. We conclude $\eta^{0,r'}_{ \hat{M} M'} \approx_{ 2 \cdot 2^{\vert W_3 \vert-\vert E  \vert} }  \sigma_{\hat{M}} \otimes \bot_{M'}$. 

Since the above argument works for every fixing of $R'=r'$ and $\eta^0 = \mathbb{E}_{r'} \eta^{0,r'}$, we have the desired.

 \qed\end{proof}

 \begin{proof} 

 [of \cref{lemma:nekjbs2}] Let $\nu^1, \chi^1, \mu^1$ be the intermediate states and $\eta^1$ be the final state when we run the CPTP map $\Upsilon$ on $U_{\vert R \vert} \otimes \tau^1_{W_2'W_3E' \hat{E}M\hat{M}}$ (see \Cref{fig:3splitstatetdc1}).

 Note we have $(\tau^1)_{W_2'W_3E' \hat{E}M\hat{M}} = (\tau^1)_{W_2'W_3E' \hat{E}} \otimes \sigma_{M \hat{M}}$. Consider the state $\mu^1$, i.e. the state obtained by the action of $C_R$ on $\tau^1$ followed by CPTP map $T$, followed by $C^\dagger_R$. Using, \cref{lem:twirl-wsi}, we have 
\[\mu^1_{\hat{M}M'\hat{E}' E' W'_3} \approx_{2/(4^{\vert M \vert +\vert \hat{E} \vert} -1)}  T_1 ( (\tau^1)_{\hat{M}M\hat{E}E'W_3}  ) + T_2 (( \tau^1)_{\hat{M} E'W_3 } \otimes U_{M \hat{E}}  ), \]
where $T_1(.) : \cL( \cH_{W_3}) \to  \cL( \cH_{W'_3})$ ,  $T_2(.) : \cL( \cH_{W_3}) \to  \cL( \cH_{W'_3})$ are CP maps such that $T_1 +T_2$ is trace preserving and they depend only on adversary CPTP map $T(.)$. Note both \[  T_1 ( (\tau^1)_{\hat{M}M\hat{E}E'W_3}  ) \quad ; \quad T_2 (( \tau^1)_{\hat{M} E'W_3 } \otimes U_{M \hat{E}}  ) \] are sub-normalized density operators. Let $p_1 \defeq \Tr( T_1 ( (\tau^1)_{\hat{M}M\hat{E}E'W_3}  ) ) $. Let  \[ \mu^{1,0}  \defeq  \frac{1}{p_1} T_1 ( (\tau^1)_{\hat{M}M\hat{E}E'W_3}  )  \quad ; \quad  \mu^{1,1}  \defeq  \frac{1}{1-p_1} T_2 (( \tau^1)_{\hat{M} E'W_3 } \otimes U_{M \hat{E}}  ). \]
Note $\mu^1  \approx_{2/(4^{\vert M \vert + \vert \hat{E} \vert} -1)}  p_1 \mu^{1,0} + (1-p_1)\mu^{1,1}$. Let the final states be $\eta^{1,0}$, $\eta^{1,1}$ when we run the $\epr$ test followed by $\mathsf{Rep}$ on $\mu^{1,0}$, $\mu^{1,1}$ respectively. Since,  $\mu^1  \approx_{2/(4^{\vert M \vert + \vert \hat{E} \vert}-1)}  p_1 \mu^{1,0} + (1-p_1)\mu^{1,1}$, we conclude,
\begin{equation}\label{lemma:eq1djfvsissdd}
   \eta^1  \approx_{2/(4^{\vert M \vert + \vert \hat{E} \vert}-1)}  p_1 \eta^{1,0} + (1-p_1)\eta^{1,1}.
\end{equation}
In the first case, since  $(\mu^{1,0})_{W'_3E' \hat{E}'M'\hat{M}} = (\mu^{1,0})_{W'_3E' \hat{E}'} \otimes \sigma_{M \hat{M}}$, we can conclude that
\begin{equation}\label{lemma:eq1djfvsi}
    \eta^{1,0} = p_2 \sigma_{ \hat{M}M} + (1-p_2) \sigma_{\hat{M}} \otimes \bot_{M'},
\end{equation}
and furthermore $p_2$ depends only on CP map $T_1(.)$ and state $\tau^1$, which further depends only on tampering adversary $\mathcal{A}$.

In the second case, since  $(\mu^{1,1})_{W'_3E' \hat{E}M'\hat{M}} = (\mu^{1,1})_{W'_3E' } \otimes U_{\hat{E}'M} \otimes  \sigma_{ \hat{M}}$, we can conclude 
\begin{equation}\label{lemma:eq1djfvsiw}
    \eta^{1,1}_{\hat{M}M'} \approx_{ 2 \cdot 2^{-\vert E \vert}} \sigma_{\hat{M}} \otimes \bot_{M'},
\end{equation} 
using \cref{lemma:tss-bellbasis} as $\epr$ test rejects with probability atleast $1- 2^{-\vert E \vert}$.

Combining \cref{lemma:eq1djfvsissdd}, \cref{lemma:eq1djfvsi} and  \cref{lemma:eq1djfvsiw}, we have the following:
 \begin{align*}
   &\Upsilon(U_{\vert R \vert} \otimes (\tau^1)_{W_2'W_3E' \hat{E}M\hat{M}})  \\
     &=  \eta^{1}_{\hat{M}M'} \\
     &  \approx_{2/(4^{\vert M \vert + \vert \hat{E} \vert}-1)}  p_1 \eta^{1,0} + (1-p_1)\eta^{1,1}  \\
     & = p_1p_2 \sigma_{\hat{M}M} + p_1 (1-p_2) \sigma_{\hat{M}} \otimes \bot_{M'} +(1- p_1) \eta^{1,1}  \\
     & \approx_{2 \cdot 2^{-\vert E \vert}}  p_1 \cdot p_2 \sigma_{\hat{M}M} + p_1 (1-p_2) \sigma_{\hat{M}} \otimes \bot_{M'} +(1- p_1) \sigma_{\hat{M}} \otimes \bot_{M'} \\
     &= p_1 \cdot p_2 \sigma_{\hat{M}M} + (1-p_1 \cdot p_2) \sigma_{\hat{M}} \otimes \bot_{M'}.
 \end{align*}

Considering $p =p_1 \cdot p_2  $, we have the desired. Through $p_1, p_2$, $p$ depends only on tampering adversary $\mathcal{A}$. 
 \qed\end{proof}

\section{Secret Sharing Schemes and definitions of their many variants}\label{section:prelim-ss}

In this Section we present basic definitions of secret sharing schemes, in addition to known constructions in the literature that we require to instantiate our compilers. We begin in \Cref{subsubsection:prelim-threshold-ss} by introducing threshold secret sharing, proceeded by leakage-resilient secret sharing in \Cref{subsubsection:prelim-lrss}, and non-malleable secret sharing in \Cref{subsubsection:prelim-nmss}. Finally, in \Cref{subsubsection:prelim-tamper-detecting-ss} we formally introduce the tamper detecting secret sharing schemes we construct in this work.

\subsection{Threshold Secret Sharing Schemes}
\label{subsubsection:prelim-threshold-ss}

Informally, in a $t$-out-of-$p$ secret sharing scheme, any $t$ honest shares suffice to reconstruct the secret, but no $t-1$ shares offer any information about it.

\begin{definition}
    [$(p, t, \epspriv, \epsilon_c)$-Secret Sharing Scheme]\label{definition:secret-sharing} Let $\mathcal{M}$ be a finite set of secrets, where $|\mathcal{M}| \geq 2$. Let $[p] = \{1, 2, \cdots , p\}$ be a set of identities (indices) of $p$ parties. A sharing channel $\Enc$ with domain of secrets $\mathcal{M}$ is a $(p, t, \epspriv, \epsilon_c)$ \emph{threshold secret sharing scheme} if the following two properties hold:
    \begin{enumerate}
        \item \textbf{Correctness}: the secret can be reconstructed by any set of parties $T\subset [p], |T|\geq t$. That is, for every such $T$, there exists a reconstruction channel $\Dec_T$ such that
        \begin{equation}
         \forall m\in \mathcal{M}: \text{ } \mathbb{P}[\Dec_T\circ \Enc(m)_T\neq m]\leq \epsilon_c
        \end{equation}
        \item \textbf{Statistical Privacy}: Any collusion of $|T|\leq t-1$ parties has “almost” no information about the underlying secret. That is, for every distinguisher $D$ with binary output:
        \begin{equation}
         \forall m_0, m_1\in \mathcal{M}: \text{ } |\mathbb{P}[D(\Enc(m_0)_T)=1] - \mathbb{P}[D(\Enc(m_1)_T)=1]|\leq \epspriv
        \end{equation}
    \end{enumerate}
\end{definition}

We point out that we use the same syntactic definition of secret sharing for schemes which hide classical messages into classical or quantum ciphers. However, when the cipher is classical, we refer to the encoding/decoding channels as $\share$ and $\rec$ (reconstruction) functions, as opposed to quantum channels $(\Enc, \Dec)$. The block-length of the secret sharing scheme is the total qubit $+$ bit-length of all the shares, and its rate as the ratio of message length (in bits) to block-length.

To instantiate our constructions, we use Shamir’s  standard threshold secret sharing scheme:

\begin{fact}
    [\cite{Sha79}] For any number of parties $p$ and threshold $t$ such that $t\leq p$, there exists a $t$-out-of-$p$ secret sharing scheme $(\share, \rec)$ for classical messages of length $b$ with share size at most $\max(p, b)$, where both the sharing and reconstruction procedures run in time poly$(p, b)$.
\end{fact}

\subsection{Leakage-Resilient Secret Sharing Schemes}
\label{subsubsection:prelim-lrss}

Integral in our constructions are classical secret sharing schemes which offer privacy guarantees even if unauthorized subsets are allowed to leak information about their shares to each other in an attempt to distinguish the message. Moreover, we depart from standard classical leakage models in that the leaked information itself could be a quantum state, even if the shares are classical. The definition of leakage-resilient secret sharing formalizes this idea as follows:

\begin{definition}
[Leakage-Resilient Secret Sharing]\label{def:lrss} Let $(\share, \rec)$ be a secret sharing scheme with randomized sharing function $\Share:\mathcal{M}\rightarrow \{\{0, 1\}^{l'}\}^p$, and let $\mathcal{F}$ be a family of leakage channels. Then $\share$ is said to be $(\mathcal{F}, \epsilon_{lr})$ leakage-resilient if, for every channel $\Lambda\in \mathcal{F}$, 
    \begin{equation}
       \forall m_0, m_1\in \mathcal{M}: \text{ } \Lambda(\Share(m_0)) \approx_{\epsilon_{lr}}  \Lambda(\Share(m_1))
    \end{equation}
\end{definition}

As an example, the standard local leakage model in the context of classical secret sharing schemes allows bounded leakage queries $\{f_i: \{0, 1\}^{l'} \rightarrow \{0, 1\}^\mu\}_{i\in K}$, on each share corresponding to an arbitrary set of indices $K\subset [n]$, and further allows full share queries corresponding to an unauthorised subset $T\subset [n]$:

\begin{equation}
    \mathsf{Leak}(\share(m)) = \share(m)_T, \{f_i(\share(m)_i)\}_{i\in K}.
\end{equation}

In our constructions, we unfortunately require a slight modification to this setting, where the leakage parties $K$ are an unauthorized subset, but are allowed to \textit{jointly} leak a small quantum state (dependent on their shares) to $T$. We formally introduce this model as follows:

\begin{definition}
[Quantum $k$-Local Leakage Model]\label{def:k-local-leakage}
    For any integer sizes $p, t, k$ and leakage length (in qubits) $\mu$, we define the $(p, t, k, \mu)$-local leakage model to be the collection of channels specified by
    \begin{equation}
    \mathcal{F}_{k, \mu}^{p, t} =\bigg\{ (T, K, \Lambda): T, K\subset [p],|T|< t, |K|\leq k, \text{ and }\Lambda:\{0, 1\}^{l'\cdot |K|}\rightarrow \cL(\mathcal{H})\bigg\},
\end{equation}
\noindent where $\log \dim(\mathcal{H}) = \mu$. A leakage query $(T, K, \Lambda)\in \mathcal{F}_{k, \mu}^{p, t}$ on a secret $m$ is the density matrix:
\begin{equation}
    (\mathbb{I}_{T}\otimes \Lambda_K)(\Share(m)_{T \cup K})
\end{equation}
\end{definition}

 \noindent In other words, the parties in $K$ perform a quantum channel on their (classical) shares and send the $\mu$ qubit output state to $T$. To the extent of our knowledge, in the literature we do not know of $\lrss$ constructions in this model (even with classical leakage). However, in \Cref{section:lrss} we show that simple modifications to a construction of $\lrss$ against local leakage by \cite{CKOS22} provides such guarantees. 

\begin{theorem}[\Cref{theorem:lrss-local-leakage}, restatement]\label{theorem:lrss-local-leakage-prelim}
    For every $\mu, k, t, p, l\in \mathbb{N}$ such that $k<t, t+k\leq p,$ and $p\leq l$, there exists an $(p, t, 0, 0)$ threshold secret sharing scheme on messages of $l$ bits and shares of size $l+\mu +o(l, \mu)$ bits, which is perfectly correct and private and $p\cdot 2^{-\Tilde{\Omega}(\sqrt[3]{\frac{l+\mu}{p}})}$ leakage-resilient against $\mathcal{F}_{k, \mu}^{p, t}$.
\end{theorem}

\subsection{Non-Malleable Secret Sharing Schemes}
\label{subsubsection:prelim-nmss}

\cite{GK16} introduced the notion of a non-malleable secret sharing scheme (NMSS). An NMSS is a secret sharing scheme which is robust to certain types of tampering on the shares. 

\begin{definition}
    [$(p, t, \epspriv, \epsilon_c, \epsilon_\NM)$ Non-Malleable Secret Sharing Scheme] Let $(\Enc, \Dec)$ be a $(p, t, \epspriv, \epsilon_c)$ secret sharing scheme, and let $\mathcal{F}$ be a family of tampering channels. Then, we refer to $(\Enc, \Dec)$ as a $(p, t, \epspriv, \epsilon_c, \epsilon_\NM)$ \emph{non-malleable secret sharing scheme} against $\mathcal{F}$ if, for all $\Lambda\in \mathcal{F}$ and authorized subset $T\subset [p]$ of size $|T|\geq t$, there exists $p_\Lambda\in [0, 1]$ and a distribution $\mathcal{D}_\Lambda$ such that 
        \begin{equation}
        \forall  m\in \mathcal{M}, \quad \Dec_T\circ \Lambda\circ \Enc(m) \approx_{\epsilon_\NM} p_\Lambda\cdot m + (1-p_\Lambda)\cdot \mathcal{D}_\Lambda
        \end{equation}
\end{definition}

In other words, the recovered secret after tampering is a convex combination of the original secret, or a fixed distribution over secrets. While we don't actually use NMSS's directly in this paper, they provide an intuitive starting point towards the new definitions of TDSS's below.

\subsection{Tamper Detecting Secret Sharing Schemes}
\label{subsubsection:prelim-tamper-detecting-ss}

Below we formalize two tamper detection properties for the secret sharing schemes that we consider in this section. The first of the two is known as ``weak" tamper detection, and informally can be understood as the ability to \textit{either} recover the message or reject after tampering, with high probability. However, the probability that either of these two events occurs may depend on $m$:

\begin{definition}[Weak Tamper-Detecting Secret Sharing Schemes]\label{def:weak-td}
    Let a family of CPTP maps $(\Enc, \Dec)$ be a $(p, t, \epspriv, \epsilon_c)$ secret sharing scheme over $k$ bit messages and $\mathcal{F}$ a family of tampering channels. We refer to $(\Enc, \Dec)$ as $(\mathcal{F}, r, \epsilon_\td)$ \emph{weak tamper-detecting} if, for every subset $R\subset [p]$ of size $|R|\geq r$ and tampering channel $\Lambda\in \mathcal{F}$,
    \begin{equation}
   \forall m\in \{0, 1\}^k, \quad  \mathbb{P}[\Dec_R\circ\Lambda\circ \Enc(m)\notin \{m, \bot\}]\leq \epsilon_\td \text{ }
    \end{equation}
\end{definition}

Note that in the above, we implicitly consider \textit{ramp} secret sharing schemes, where the threshold for privacy, reconstruction under honest shares, and reconstruction under tampered shares, are different. In a (generic) tamper-detecting secret sharing scheme, we stipulate that the distribution over the recovered message is near convex combination of the original $m$ and $\bot$, which doesn't depend on $m$:

\begin{definition}[Tamper-Detecting Ramp Secret Sharing Schemes]
    Let a family of CPTP maps $(\Enc, \Dec)$ be a $(p, t, \epspriv, \epsilon_c)$ secret sharing scheme over $k$ bit messages and $\mathcal{F}$ a family of tampering channels. We refer to $(\Enc, \Dec)$ as $(\mathcal{F}, r, \epsilon_\td)$ \emph{tamper-detecting} if, for every subset $R\subset [p]$ of size $|R|\geq r$ and tampering channel $\Lambda\in \mathcal{F}$, there exists a constant $p_\Lambda\in [0, 1]$ such that
    \begin{equation}
       \forall m\in \{0, 1\}^k, \quad  \Dec_R\circ\Lambda\circ \Enc(m) \approx_{\epsilon_\td} p_\Lambda \cdot m + (1-p_\Lambda)\cdot \bot\text{ }
    \end{equation}
\end{definition}

\section{Tamper-Detecting Secret Sharing Schemes}
\label{section:tamper-detecting-secret-sharing}

In this section, we prove our main result on the construction of secret sharing schemes which detect local tampering, restated below. We refer the reader to \Cref{subsubsection:prelim-tamper-detecting-ss} for the definition of tamper-detecting secret sharing schemes.

\begin{theorem}[\cref{theorem:results-lo-ss}, restatement]
    For every $p, t$ s.t. $4 \leq t \leq p-2$, there exists a $(p, t, 0, 0)$ secret sharing scheme for $k$ bit messages which is $(\lo^p, t+2, 2^{-\max(p, k)^{\Omega(1)}})$-tamper detecting. 
\end{theorem}

In other words, no $t-1$ shares reveal any information about the message, any $t$ honest shares uniquely determine the message, and one can detect $\lo$ tampering on any $(t + 2)$ shares. 

We organize this section as follows. We begin in \Cref{subsection:tdss-components} by introducing the relevant code components, overviewing the construction, as well as presenting the construction and analysis of the ``triangle gadgets" discussed in~\cref{subsection:overview-td-ss}. Next, in \Cref{subsection:tdss-construction} we present our secret sharing scheme, and in the subsequent \Cref{subsection:tdss-analysis} we present its analysis.

\subsection{Code Components and the Triangle Gadget}
\label{subsection:tdss-components}

In \Cref{subsubsection:tdss-ingredients}, we describe the ingredients in our secret sharing scheme, in \Cref{subsubsection:triangle} we introduce the ``triangle gadget", and in \Cref{subsubsection:analysis-triangle} we prove that it inherits two basic properties: (weak) tamper detection, and 3-out-of-3 secret sharing. 

\subsubsection{Ingredients and Overview}
\label{subsubsection:tdss-ingredients}

We combine: 

\begin{enumerate}

    \item $(\share_{\lrss}, \rec_\lrss)$: A classical $(p, t, 0,0)$-secret sharing scheme which is $\epsilon_{lr}$-leakage resilient to $\mu$ qubits of the $3$-local leakage model $\mathcal{F}_{3, \mu}^{p, t}$, such as that of \Cref{theorem:lrss-local-leakage-prelim}. 
    
    For basic definitions of leakage-resilient secret sharing, see \Cref{subsubsection:prelim-lrss}.\\
    
    \item $(\Enc_\td^\lambda, \Dec_\td^\lambda)$, $\lambda\in \mathbb{N}$: The 3-split-state quantum tamper-detection code in the bounded storage model of \Cref{theorem:tdc-loa} with $\lambda$ EPR pairs, comprised of:

    \begin{enumerate}
    \item $\nmext:\{0, 1\}^q\times \{0, 1\}^{\delta \cdot q}\rightarrow \{0, 1\}^{r}$, a quantum secure two source non-malleable extractor, with error $\epsilon_\NM =2^{-q^{\Omega_\delta(1)}}$ and output size $r=(1/2-\delta)q$ from \cref{lem:qnmcodesfromnmext}.

    \item The family of $2$-design unitaries $C_R$ from \cref{lem:subclifford} on $\frac{1}{5}\cdot r$ qubits.
\end{enumerate}

    From \Cref{theorem:tdc-loa}, $(\Enc_\td^\lambda, \Dec_\td^\lambda)$ has message length $\frac{r}{5}-\lambda$ and is $\epsilon_\td\leq 2(\epsilon_\NM+2^{a-\lambda})$ secure against $\lo_{(a, *, *)}$. Definitions of these code components can be found in \Cref{section:bounded-storage} and \cref{section:prelim}.
\end{enumerate}

Our approach to secret sharing attempts to extend the compiler by \cite{ADNOP19} to a tripartite setting. First, the classical message $m$ is shared into the secret sharing scheme, resulting in classical shares $(M_1, \cdots, M_p)\leftarrow \share_\lrss(m)$. Then every triplet of parties $a<b<c$ jointly encodes their shares $M_a,M_b, M_c$ into a code $\Enc_\triangle$ supported on a tripartite register, the shares of which are redistributed among $a, b, c$. 

In this subsection, we present our gadget code $\Enc_\triangle$ and prove two of its relevant properties: a strong form of 3-out-of-3 secret sharing, as well as a weak form of tamper detection. In the next subsections, we show how to leverage these gadget properties together with that of the underlying leakage-resilient secret sharing scheme to achieve tamper detection, globally. 

\subsubsection{The Triangle Gadget}.
\label{subsubsection:triangle}

\begin{algorithm}[H]
    \setstretch{1.35}
    \caption{$\Enc_{\triangle}$: The ``triangle gadget"}
    \label{alg:triangle-gadget}
    \KwInput{Three parties $p_0<p_1<p_2\in [p]$, three messages $M_0, M_1, M_2$, and an integer parameter $\lambda$.}
    \KwOutput{A quantum state $\Enc_{\triangle}^\lambda(M_0, M_1, M_2)$ on a tripartite register $A, B, C$.}

    \begin{algorithmic}[1]

    \State Each party $i\in \{0, 1, 2\}$ encrypts their message $M_i$ into the tamper-detection code of \Cref{alg:algorithm3qtdc}, into a tripartite register $(Q_i, Y_i\hat{E}_i, X_i)$. Explicitly, party $i$:
    \begin{enumerate}
        \item Samples uniformly random classical registers $X_i, Y_i$ and evaluates the key $R_i = \nmext(X_i, Y_i)$.
        \item Prepares $\lambda_i = (i+1)\cdot \lambda$ EPR pairs, on a bipartite register $E_i\hat{E}_i$.
        \item Samples the Clifford unitary $C_{R_i}$ using the sampling process $\samp$ in \cref{lem:subclifford}. Applies $C_{R_i}$ on registers $M_i, E_i$ to generate $Q_i$.
        \item Hands register $X_i$ to party $i-1$ (mod $3$), registers $Y_i, \hat{E}_i$ to party $i+1$ (mod $3$), and keeps the authenticated register $Q_i$. 
        \end{enumerate}

    \State Output registers $A, B, C$ for parties $p_0, p_1, p_2$ respectively, where
    \begin{equation}
        A = (Q_0, Y_2\hat{E}_2, X_1), B = (Q_1, Y_0\hat{E}_0, X_2), C = (Q_2, Y_1\hat{E}_1, X_0)
    \end{equation}
    \end{algorithmic}
\end{algorithm}

$\Dec_\triangle$: To decode, we simply run the decoder $\Dec_\td$ for the tamper-detection code on the registers $(Q_i, Y_i\hat{E}_i, X_i)$ corresponding to $\Enc_\td(M_i)$. If any of the three are $\bot$, output $\bot$. If otherwise, output the resulting messages $({M}'_0, {M}'_1, {M}'_2)$. Note that if no tampering is present, then from the perfect correctness of the tamper-detection code $(\Enc_\triangle, \Dec_\triangle)$ is also perfectly correct. 

\subsubsection{Analysis of the Triangle Gadget}
\label{subsubsection:analysis-triangle}

Before moving on to our global construction, we prove two important properties of the gadget $(\Enc_\triangle, \Dec_\triangle)$. The first is that not only is $\Enc_{\triangle}$ $3$-out-of-$3$ secret sharing, but in fact any two shares $\Enc_{\triangle}$ are near maximally mixed. Note that 3-out-of-3 secret sharing only implies that the two-party reduced density matrices are independent of the message, but apriori they could still be entangled. 

\begin{lemma}
[Pairwise Independence]\label{lemma:triangle-pairwise-indep} For every share $W\in \{A, B, C\}$ and fixed strings $m_0, m_1, m_2$, the two party marginal of $\Enc_{\triangle}(m_0, m_1, m_2)$ without $W$ is $\delta_{\pair}$-close to maximally mixed, where $\delta_{\pair}\leq 6\cdot (\epsilon_\NM + 2^{-\lambda})$.
\end{lemma}

\begin{proof}
    Recall that $\Enc_\td^e$ is pairwise independent with error $\leq 2(\epsilon_\NM + 2^{-e})$ from \Cref{corr:2nmssq}. The result follows from a triangle inequality. 
\qed\end{proof}

The second property concerns the resilience of $\Enc_{\triangle}$ to tampering in $\lo^3$. We prove that two of the three inputs, namely $m_1, m_2$, are individually tamper-detected: 

\begin{lemma}[Share-wise Tamper Detection]\label{lemma:share-weak-td}
    Fix $\Lambda\in \lo^3$, three strings $m_0, m_1, m_2$, and an integer $\lambda$. Let $({M}'_0, {M}'_1, {M}'_2)\leftarrow \Dec_\triangle\circ\Lambda\circ \Enc_\triangle(m_0, m_1, m_2)$ denote the distribution over the recovered shares. Then for both $i\in \{ 1, 2\}$, 
    \begin{equation}
        \mathbb{P}[{M}'_i\notin \{m_i, \bot\}]\leq  \epsilon_{\share},
    \end{equation}
    \noindent where $\epsilon_{\share}\leq \epsilon_\NM + 2^{4-\lambda}$.
\end{lemma}

By a union bound, we have $\mathbb{P}[({M}'_1,{M}'_2) \notin \{(m_1,m_2),\bot\}]\leq 2\cdot \epsilon_{\share}$. However, here we already exhibit the selective bot problem. The event ${M}'_1=\bot$ may correlate with ${M}'_2$, breaking the tamper detection guarantee. 

\begin{proof} 

[of \cref{lemma:share-weak-td}]
    Fix $m_0, m_1, m_2$, the tampering channel $\Lambda\in \lo^3$, and $i\in \{ 1, 2\}$. In \Cref{alg:triangle-gadget}, note that the $i$th share of $\Enc_\triangle$ is comprised of registers $(Q_i, Y_{i-1}\hat{E}_{i-1}, X_{i+1})$, where 
    \begin{enumerate}
        \item $Q_i$ is a share of $\Enc_{\td}^{\lambda_i}(m_i)$
        \item The only other quantum register contains $|\hat{E}_{i-1}|\leq \lambda_{(i-1)\text{mod} 3}$ qubits. 
    \end{enumerate}

    We can now consider the marginal distribution over the $i$th message $m_i$ after the tampering channel. The marginal distribution over the recovered share ${m}'_i$ in the gadget $\Enc_\triangle$ can be simulated using a channel $\Lambda'$ directly on $\Enc_\td(m_i)$:
    \begin{equation}
        \Tr_{\neq i}\Dec_\triangle\circ\Lambda\circ \Enc_\triangle(m_1, m_2, m_3) = \Dec_\td\circ\Lambda'\circ \Enc_\td(m_i),
    \end{equation}
where $\Lambda'$ consists of a local tampering channel on $(Q_i, Y_{i}\hat{E}_{i}, X_{i})$ aided by pre-shared entanglement in a particular form: the adversary who receives share $Q_i$ holds at most $a=\lambda_{i-1}$ qubits of pre-shared entanglement (but the other two may hold unbounded-size registers). This is precisely the model $\lo^3_{(a, *, *)}$ of \Cref{def:lotsharedentanglement}, and we recall by \Cref{theorem:tdc-loa} that $\Enc_\td^{e}$ detects tampering against  $\lo^3_{(a, *, *)}$ with error $2\epsilon_\NM + 2^{4+a-e}$. Since parties $i\in \{1, 2\}$ hold more EPR pairs in their tamper-detection codes than the corresponding adversaries pre-shared, i.e $\lambda_i - \lambda_{i-1}= \lambda>0$ when $i\in \{1, 2\}$, we conclude their shares are tamper-detected:
    \begin{equation}
       i\in \{1, 2\}: \mathbb{P}[{M}'_i\notin \{m_i, \bot\}]\leq 2\epsilon_\NM + 2^{4-\lambda}.
    \end{equation}
\qed\end{proof}

\subsection{Code Construction}
\label{subsection:tdss-construction}

%\vspace{-0.7cm}
We describe our encoding algorithm in \Cref{alg:td-ss} below:

\begin{algorithm}[H]
    \setstretch{1.35}
    \caption{$\Enc$: A Tamper-Detecting Ramp Secret Sharing Scheme}
    \label{alg:td-ss}
    \KwInput{A $k$ bit message $m$.}
    %\KwOutput{The product state $\rho^m_L\otimes \rho^r_R$ for some $m\in \{0, 1\}$ correlated with $b$}

    \begin{algorithmic}[1]

    \State Encode $m$ into the $\lrss$, $(M_1, \cdots, M_p)\leftarrow \share_\lrss(m)$.
    
    \State For every triplet of parties $a< b < c\in [p]$, 
    \begin{enumerate}
        \item Encode the shares $(M_a, M_b, M_c)$ into the ``triangle gadget" described in \Cref{alg:triangle-gadget} supported on triplets of quantum registers
        \begin{equation}
            A_{(a, b, c)}, B_{(a, b, c)}, C_{(a, b, c)}\leftarrow \Enc_{\triangle}(M_a, M_b, M_c)
        \end{equation}
        \item Hand the $A_{(a, b, c)}$ register to party $a$, $B_{(a, b, c)}$ to $b$, and $C_{(a, b, c)}$ to $c$.
    \end{enumerate}

    \State Let the resulting $i$th share be the collection of quantum registers
    \begin{gather}
        S_i = \{A_{(i, b, c)}:\forall b, c \text{ s.t. } i<b<c\}\bigcup \\ \{B_{(a, i, c)}:\forall a, c\text{ s.t. } a<i<c\}\bigcup \\\{C_{(a, b, i)}: \forall a, b \text{ s.t. }a<b< i\}
    \end{gather}
    \end{algorithmic}

\end{algorithm}

We are now in a position to describe our decoding algorithm. Upon receiving the locally tampered shares of any authorized subset $T$ of size equal to $t+2$, our decoder partitions $T$ into two un-authorized subsets, and decodes the gadgets only within each partition. 
\newpage
\begin{algorithm}[H]
    \setstretch{1.35}
    \caption{$\Dec$: A Bipartite Decoding Algorithm.}
    \label{alg:td-ss-dec}
    \KwInput{An (authorized) subset $T$ of size $t+2$, and a collection of tampered quantum registers ${S}'_i:i\in T$.}
    \KwOutput{A $k$ bit message $M'$}

    \begin{algorithmic}[1]

    \State Partition $T$ into a subset $U$ of the three smallest indexed shares, and $T\setminus U$. 

    \State For every triplet of parties $a<b<c$ contained entirely in $T\setminus U$ or entirely in $U$:

    \begin{enumerate}
        \item Apply the triangle gadget decoding algorithm $\Dec_\triangle$ on registers $A_{(a, b, c)}, B_{(a, b, c)}, C_{(a, b, c)}$.
        \item Output $\bot$ if so does the decoder. Otherwise, let $M_{(a, bc)}', M_{(b, ac)}', M_{(c, ab)}'$ be the recovered shares.
    \end{enumerate}

    \State If $M_{(a, bc)}'\neq M_{(a, de)}'$ for any two triangles $(a, b, c), (a, d, e)$ on the same side of the partition, output $\bot$. Otherwise, let $M_U'$, $M'_{T\setminus U}$ be the recovered shares from either side of the partition.

    \item Remove the lowest index share of $U$ to obtain $V\subset U$ and $T\setminus U$ to obtain $W\subset T\setminus U$. Note $|V\cup W|\geq t$.
    
    \State Run the decoder $M'\leftarrow \rec_\lrss(M'_V, M'_W)$ on the classical shares of $V\cup W$.
    \end{algorithmic}

\end{algorithm}

The main result of this section proves that the secret sharing scheme described in \Cref{alg:td-ss} and \Cref{alg:td-ss-dec} above detects unentangled tampering, when handed at least $t+2$ shares.

\begin{lemma}\label{lemma:tdss}
    Assuming the $\lrss$ is $\epsilon_{lr}$-resilient to $\mu \geq 10\cdot p^2\cdot \lambda$ qubits of leakage, then $(\Enc, \Dec)$ described above is a $(p, t, 0, 0)$-secret sharing scheme which is $(\lo^p, t+2, \epsilon)$-tamper-detecting with error $\epsilon = O(\epsilon_{lr} + p^4\cdot (\epsilon_{\NM}^{1/2} + 2^{-\lambda/2}))$.
\end{lemma}

We dedicate the next subsection to its proof. By instantiating \Cref{lemma:tdss} above,

\begin{theorem}
    There exists an efficient $(p, t, 0, 0)$ secret sharing scheme for $k$ bit messages which is $(\lo^p, t+3, 2^{-\max(k, p)^{\Omega(1)}})$-tamper-detecting. 
\end{theorem}

\begin{proof}
    We use the construction of the $(p, t,0, 0)$ $\lrss$ from \Cref{theorem:lrss-local-leakage-prelim} with message length $k$ bits, shares of length $s=k+ p^2\cdot \lambda$ bits, and error $2^{-\Omega(\lambda)}$; together with the tamper-detection codes $\Enc^\lambda$ from \Cref{theorem:tdc-loa} with message size $s$, error $2^{-\lambda} + 2^{-s^{\Omega(1)}}$, and selecting $\lambda = \max(k, p)$, we conclude the corollary.
\qed \end{proof}

\subsection{Analysis}
\label{subsection:tdss-analysis}

Our proof is comprised of two key lemmas, which we state and use to prove \Cref{lemma:tdss} and then analyze in the subsequent sections. We begin by arguing that our secret sharing scheme inherits a weak form of tamper detection from the ``outer" split-state tamper-detection codes.

\begin{lemma}[Weak Tamper Detection]\label{lemma:tdss-weak-td}
    For every tampering channel $\Lambda\in \lo^p$, authorized subset $T\subset [p]$ of size $\geq t+2$ and message $m$, the distribution over the recovered message $M'\leftarrow \Dec_T\circ \Lambda\circ \Enc(m)$ satisfies
    \begin{equation}
        \mathbb{P}[M'\notin \{m, \bot\}]\leq \eta_\td  = O(p^3\cdot (\epsilon_\share +  \delta_{\pair}\cdot p)) = O(p^4(\epsilon_{\NM}+2^{-\lambda}))
    \end{equation}

    \noindent In fact, all the shares of $V\cup W$ are recovered whp: $\mathbb{P}[\exists i\in V\cup W: M'_i\notin \{m_i, \bot\}]\leq \eta_\td$
\end{lemma}

In other words, $(\Enc, \Dec)$ is a secret sharing scheme which is weak tamper-detecting against $\lo^p$, so long as the decoder receives at least $t+2$ shares. As previously discussed, the reason this is only \textit{weak} tamper detecting is due to the selective bot problem - the probability with which we reject may apriori depend on $m$. We leverage the leakage resilience of the underlying secret sharing scheme to ensure this cannot occur:

\begin{lemma}[The Selective Bot Problem]\label{lemma:tdss-selective-bot}
    For every tampering channel $\Lambda\in \lo^p$ and authorized subset $T\subset [p]$ of size $\geq t+2$, the probability the decoding algorithm $\Dec$ in \Cref{alg:td-ss-dec} rejects is near independent of the message. That is,
    \begin{equation}
       \forall \Lambda\in \lo^p, m_0, m_1: \big|\mathbb{P}[\Dec_T\circ \Lambda\circ \Enc(m_0) = \bot] - \mathbb{P}[\Dec_T\circ \Lambda\circ \Enc(m_1) = \bot]\big|\leq \eta_{lr},
    \end{equation}
    \noindent where $\eta_{lr} = O(\epsilon_{lr}+p^2\cdot \sqrt{\delta_{\pair}+\epsilon_\share}) = O(\epsilon_{lr} + p^2(\epsilon_\NM^{1/2} + 2^{-\lambda/2}))$.
\end{lemma}

Put together, we now conclude the proof:

\begin{proof}

[of \Cref{lemma:tdss}] Let the random variable $M$ denote the uniform distribution over messages, and let $m$ be any fixed value of $M$. From \Cref{lemma:tdss-weak-td}, for a fixed $\Lambda\in \lo^p$ and  $T\subset [p]$ the distribution over the recovered message $M'$ is near a convex combination of $M$ and rejection, with bias $\gamma_m$ (dependent on $m$):
\begin{equation}
    M'|_{M=m} \approx_{\eta_\td} \gamma_m \cdot m+(1-\gamma_m) \cdot \bot.
\end{equation}

However, by \Cref{lemma:tdss-selective-bot}, the probability of $\bot$ can barely depend on $M$: 
\begin{equation}
    1-\gamma_m \approx_{\eta_\td} \mathbb{P}[M' = \bot|M=m] \approx_{\eta_{lr}} \mathbb{P}[M' = \bot] \equiv p_{ \mathsf{Adv}}
\end{equation}

By the triangle inequality, $M'|_{M=m}\approx_{3\cdot \eta_\td + 2\eta_{lr}} (1-p_{ \mathsf{Adv}})\cdot m+ p_{ \mathsf{Adv}}\cdot \bot$, for some fixed $p_{ \mathsf{Adv}}$ dependent on $\Lambda, T$. With $\eta_\td, \eta_{lr}$ as in the lemmas above we conclude the theorem. 
\qed \end{proof}

\subsubsection{Proof of Weak Tamper Detection (\Cref{lemma:tdss-weak-td})}

\begin{proof}

[of \Cref{lemma:tdss-weak-td}] Consider any triangle $a<b<c\in [p]$ contained on the same side of the partition $(U, T\setminus U)$, and let $(m_a, m_b, m_c)$ be any fixing of their classical shares. Note that the reduced density matrix on their quantum registers $(S_a, S_b, S_c) \vert ((M_a, M_b, M_c) = (m_a, m_b ,m_c))$ contains the triangle $\Enc_\triangle(m_a, m_b, m_c)$, as well as the registers of all the triangles $(a', b', c')$ which intersect $(a, b, c)$. However, since each individual register of $\Enc_{\triangle}$ is maximally mixed, and all pairs of registers of $\Enc_{\triangle}$ are $\delta_\pair$ approximately pairwise independent (\Cref{lemma:triangle-pairwise-indep}), we have 
\begin{equation}
    (S_a, S_b, S_c)\vert_{((M_a, M_b, M_c) = (m_a, m_b ,m_c))} \approx_{3\cdot \delta_{\pair}\cdot p} \Enc_\triangle(m_a, m_b, m_c) \otimes \frac{\mathbb{I}}{\Tr\mathbb{I}} 
\end{equation}

The copies $\sigma = \frac{\mathbb{I}}{\Tr\mathbb{I}} $ of the maximally mixed state are separable across the shares $a, b, c$, and independent of $m$, and act as ancillas. Therefore, any channel $\Lambda\in \lo^3$ which tampers with the registers $(S_a, S_b, S_c)$ using $\sigma$ can be simulated in $\lo^3$ without them: 
\begin{gather*}
    \Dec_\triangle\circ \Lambda\circ (S_a, S_b, S_c)\vert_{((M_a, M_b, M_c) = (m_a, m_b ,m_c))} \\  \approx_{3\cdot \delta_{\pair}\cdot p} \Dec_\triangle\circ \Lambda (\Enc_\triangle(m_a, m_b, m_c) \otimes \sigma) \\= \Dec_\triangle\circ \Lambda' \circ \Enc_\triangle(m_a, m_b, m_c) 
\end{gather*}

By \Cref{lemma:share-weak-td} and the triangle inequality, this implies that two of the three shares in each triangle $(a, b, c)$ are weak tamper detected: 
\begin{equation}
    \mathbb{P}[({M}'_a, {M}'_b) \notin \{(m_a, m_b), \bot\}]\leq 2\epsilon_\share + 3\cdot \delta_{\pair}\cdot p
\end{equation}

By a union bound over all the triangles, either all the shares in $V\subset U, W\subset (T\setminus U)$ are recovered, or we reject with probability
\begin{equation}
    \mathbb{P}[\forall i\in V\cup W: {M}'_i\notin \{m_i, \bot\}]\leq t^3\cdot (2\epsilon_\share + 3\cdot \delta_{\pair}\cdot p).
\end{equation}

Since $|V\cup W|$ contains $t$ (honest) shares, $\mathbb{P}[M'\notin \{m, \bot\}]\leq t^3\cdot (2\epsilon_\share + 3\cdot \delta_{\pair}\cdot p)$.
\qed \end{proof}

\subsubsection{Proof of Leakage-Resilience (\Cref{lemma:tdss-selective-bot})}

The decoder $\Dec_T$ of \Cref{alg:td-ss-dec} rejects a tampered message $\Lambda\circ \Enc(m)$ if either subset $U$ or $(T\setminus U)$ reject (up to step 3), or if the decoder for the $\lrss$ rejects. 

We begin by analyzing the first of these two possibilities. In \Cref{lemma:tdss-sb}, we use the leakage resilience of the underlying secret sharing scheme to argue that the event that either subset $U$ or $(T\setminus U)$ rejects cannot depend on the message $m$. To do so, we show that the measurement of whether either $U$ or $(T\setminus U)$ reject can be simulated by a leakage channel from $U$ to $(T\setminus U)$ on the shares of the classical $\lrss$. 

\begin{lemma}\label{lemma:tdss-sb}
    Fix a tampering channel $\Lambda$ and subsets $T\subset [n]$, $|T|\geq t+2$ and $U\subset T$ of the smallest 3 shares in $T$. The probability $\Dec_T$ \Cref{alg:td-ss-dec} has rejected before step 4, does not depend on the message: $\forall m_0, m_1: $
    \begin{equation}
       |\mathbb{P}[\Dec_T \text{ aborts before Step 4 } |M=m_0] - \mathbb{P}[\Dec_T \text{ aborts before Step 4 }|M=m_1]|
    \end{equation}
    $\leq p^3\cdot \delta_{\pair} + \epsilon_{lr}.$
\end{lemma}

\begin{proof}
    Fix a message $m$, and let us condition on the classical $\lrss$ shares $M_T$ of $m$ in $T\subset [p]$ being a fixed value $M_T=m_T = \{m_i:i\in T\}$. Now, consider the reduced density matrix $(\Enc(m)_T|_{M_T=m_T})$ on the shares of $T\subset [p]$ conditioned on the classical shares. It is  comprised of multiple shares of $\Enc_\triangle$, for every triangle $a<b<c$ which intersects $T$. However, similarly to the proof of \Cref{lemma:tdss-weak-td}, the pairwise independence of $\Enc_\triangle$ (\Cref{lemma:triangle-pairwise-indep}) ensures any less than three shares are near maximally mixed. It suffices to then consider the triangles contained entirely in $T$:
    \begin{equation}
        \bigg\|(\Enc(m)_T|_{M_T=m_T}) -\bigg(\bigotimes_{a, b, c\in T} \Enc_\triangle(m_a, m_b, m_c)\bigg)\otimes \frac{\mathbb{I}}{\Tr\mathbb{I}}\bigg\|_1\leq p^3\cdot \delta_{\pair}
    \end{equation}

    Now, recall that $\Enc_\triangle^\lambda$ can be implemented using local operations between the three parties and at most $O(\lambda)$ qubits of shared entanglement. This implies that the parties in $U$ and $T\setminus U$ could prepare the reduced density matrix $T$, $(\Enc(m)_T|_{M_T=m_T} )$, using only their own shares $m_U, m_{T\setminus U}$ and
    \begin{enumerate}
        \item Joint operations within the subsets $U, T\setminus U$.
        \item At most $\mu=O(\lambda\cdot p^2)$ qubits of one-way communication from $U$ to $ T\setminus U$. 
    \end{enumerate}

    If we denote $L$ as the $\mu$ qubit ``leakage" register, then one can formalize the above by defining two channels $\mathcal{E}_{L, U}$ and $\mathcal{E}_{T\setminus U, L}$. $\mathcal{E}_{L, U}$ acts on the shares of $U$ producing the leakage register $L$, and $\mathcal{E}_{T\setminus U, L}$ acts on $L, T\setminus U$ respectively:
    \begin{equation}
       (\Enc(m)_T|_{M_T=m_T}) \approx_{p^3\delta_{\pair}}  (\mathcal{E}_{T\setminus U, L}\otimes \mathbb{I}_U)\circ (\mathbb{I}_{T\setminus U}\otimes \mathcal{E}_{L, U})\circ (m_{T\setminus U}, m_U)
    \end{equation}
    
    Moreover, before step $4$ in \Cref{alg:td-ss-dec} the decoding channel $\Dec_T$ factorizes as a tensor product of channels $\Dec_U\otimes \Dec_{T\setminus U}$, as does the tampering channel $\Lambda_U\otimes \Lambda_{T\setminus U}$. The output of the decoder (up to step $4$), conditioned on the message $m$ and its shares $M_T=m_T$ can thereby be simulated using 
    \begin{gather}
        (\Dec_T\circ \Lambda_T\circ \Enc(m)_T\big|_{M_T=m_T}) \\ \approx_{p^3\delta_{\pair}} (\mathcal{N}_{T\setminus U, L}\otimes \mathbb{I}_U)\circ (\mathbb{I}_{T\setminus U}\otimes \mathcal{N}_{L, U})\circ (m_{T\setminus U}, m_U)\\
        \text{ where }\mathcal{N}_{U, L} = \Dec_{U}\circ \Lambda_{U}\circ \mathcal{E}_{U, L} \text{ and }\mathcal{N}_{L, T\setminus U} = \Dec_{T\setminus U}\circ \Lambda_{T\setminus U}\circ \mathcal{E}_{{T\setminus U}, L}
    \end{gather}

    Let $\mathcal{N}_{U, L}'$ denote the channel which applies $\mathcal{N}_{U, L}$ on $m_U$ and then checks whether the decoder rejected (a bit $b$), leaks $b$ and $L$ to $T\setminus U$, and traces out $U$. Note that $\mathcal{N}_{U, L}'$ is in the leakage family $\mathcal{F}_{3, \mu+1}^{p, t}$, since $|U|=3$ leaks $\mu+1$ qubits to a subset of size $|T\setminus U|<t$. We conclude that for any two messages $m_0, m_1$, by monotonicity of trace distance under CPTP maps:
    \begin{gather}
        \big|\mathbb{P}[\Dec_T \text{ aborts before Step 4 }  |M=m_0]-\mathbb{P}[\Dec_T \text{ aborts before Step 4 }  |M=m_1]\big| \\ \leq 2\cdot p^3\cdot \delta_{\pair} + \|(\mathcal{N}_{U, L}'\otimes \mathbb{I}_{T\setminus U})\circ \share(m_0)_{T} - (\mathcal{N}_{U, L}'\otimes \mathbb{I}_{T\setminus U})\circ \share(m_1)_{T}\|_1 
    \end{gather}

\noindent $\leq 2\cdot p^3\cdot \delta_{\pair} + \epsilon_{lr}$, where last we leverage the leakage resilience of the $\lrss$ against $\mathcal{F}_{3, \mu}^{p, t}$. 
\qed \end{proof}

Observe that in \Cref{alg:td-ss-dec} the $\lrss$ can only reject if the subsets $U$ or $(T\setminus U)$ did not already reject. However, from \Cref{lemma:tdss-weak-td}, we know that if $U$ or $(T\setminus U)$ didn't reject, then (roughly) with high probability we must have recovered honest shares from the subsets $V\subset U, W\subset (T\setminus U)$ - which the $\lrss$ must accept. We can now conclude the proof of \Cref{lemma:tdss-selective-bot}:

\begin{proof}
    
    [of \Cref{lemma:tdss-selective-bot}] The decoder $\Dec_T$ rejects a tampered message $\Lambda\circ \Enc(m)$ with probability 
    \begin{gather}
       \mathbb{P}[\Dec\circ \Lambda\circ \Enc(M) = \bot|M=m] = \\ = \mathbb{P}[\Dec_T \text{ aborts before Step 4 } |M=m] + \\ + \mathbb{P}[\Dec_T \text{ doesn't abort before Step 4 }  \text{ and }\rec_\lrss(M'_V, M'_{W})= \bot|M=m] .
    \end{gather}

    From \Cref{lemma:tdss-sb}, the probability either half of the partition $U \text{ or } T\setminus U$ rejects is $\gamma = 2p^3\cdot \delta_{\pair} + \epsilon_{lr}$ close to independent of the message:
    \begin{gather}
        \mathbb{P}[\Dec_T \text{ aborts before Step 4 }|M=m_0] \approx_{\gamma} \\ \mathbb{P}[\Dec_T \text{ aborts before Step 4 }|M=m_1]
    \end{gather}

    \noindent Fix a real parameter $\delta$. We divide into two cases on $\mathbb{P}[\Dec_T \text{ aborts before Step 4 }|M=m_0]$:
    
    \begin{enumerate}
        \item $\mathbb{P}[\Dec_T \text{ doesn't abort before Step 4 }|M=m_0]\leq \delta$, then $\forall m$,
        \begin{gather}
        \mathbb{P}[\Dec_T \text{ doesn't abort before Step 4 } \text{ and }\rec_\lrss(M'_V, M'_{W})= \bot|M=m] \leq\\ \leq \mathbb{P}[\Dec_T \text{ doesn't abort before Step 4 }|M=m] \leq \delta + \gamma.
    \end{gather}

    By the triangle inequality,
    \begin{equation}
        \mathbb{P}[\Dec\circ\Lambda\circ \Enc(m_0)= \bot] \approx_{\delta+\gamma} \mathbb{P}[\Dec\circ\Lambda\circ \Enc(m_1)=\bot].
    \end{equation}

    \item $\mathbb{P}[\Dec_T \text{ doesn't abort before Step 4 }|M=m_0]\geq \delta$. From the weak tamper detection guarantee in \Cref{lemma:tdss-weak-td}, $\forall m$:
    \begin{equation}
        \mathbb{P}[\Dec_T \text{ doesn't abort before Step 4, and} (M'_V, M'_W)\neq (m_V, m_W)|M=m] \leq \eta_\td
    \end{equation}

    Since under message $m_0$, $\Dec_T$ doesn't abort before Step 4 with non-negligible probability, then conditioned on that event, the shares of $V, W$ are untampered with high probability:
     \begin{equation}
        \mathbb{P}[(M'_V, M'_W)= (m_V, m_W) |\Dec_T \text{ doesn't abort before Step 4}, M=m_0] 
    \end{equation}

    \noindent $\geq 1-\frac{\eta_\td}{\delta}.$ A similar statement holds for $m_1$ since $\mathbb{P}[\Dec_T \text{ doesn't abort before Step 4 }|M=m_1]\geq \delta-\gamma$, and therefore
    \begin{equation}
        \mathbb{P}[(M'_V, M'_W)= (m_V, m_W) |\Dec_T \text{ doesn't abort before Step 4}, M=m_1] 
    \end{equation}

    \noindent $\geq 1-\frac{\eta_\td}{\delta-\gamma}.$ Finally, since $\rec_\lrss$ has perfect correctness, if we condition on $(M'_V, M'_W)= (m_V, m_W)$, $\rec_\lrss$ recovers the message with probability $1$. Therefore, $\forall m: $
    \begin{gather}
       \mathbb{P}[ \Dec_T \text{ doesn't abort before Step 4, and }\rec_\lrss(M'_V, M'_{W})\neq \bot|M=m]  \\
        \approx_{\frac{\eta_\td}{\delta-\gamma}} \mathbb{P}[ \Dec_T \text{ doesn't abort before Step 4} |M=m].
    \end{gather}
    \end{enumerate}

    \noindent We conclude $|\mathbb{P}[\Dec\circ\Lambda\circ \Enc(m_0)= \bot] - \mathbb{P}[\Dec\circ\Lambda\circ \Enc(m_1)=\bot]|\leq \eta_{lr}$, where by the triangle inequality and optimizing over $\delta:$

    \begin{equation}
        \eta_{lr}\leq \gamma + \max(\delta, \frac{\eta_\td}{\delta-\gamma})\leq  2\gamma+\sqrt{\eta_\td} \leq O(\epsilon_{lr}+p^2\cdot \sqrt{\delta_{\pair}+\epsilon_\share}),
    \end{equation}

    \noindent where we used the values of $\gamma$ and $\eta_\td$ from \Cref{lemma:tdss-sb} and \Cref{lemma:tdss-weak-td}.
\qed \end{proof}

\section{Connections to Quantum Encryption}
\label{section:connections-encryption}

A remarkable property of quantum authentication schemes \cite{Barnum2001AuthenticationOQ} is that they encrypt the quantum message.\footnote{That is, any adversary which is oblivious to the internal randomness shared by the encoder \& decoder, can learn nothing about the encoded state. This naturally is in contrast to classical authentication, where oftentimes a short authentication tag suffices, leaving the message ``in the clear".} Since quantum non-malleable and tamper-detection codes are relaxations of quantum authentication codes, it may not seem too surprising that they inherit similar properties. In this section, we show they satisfy a related notion of encryption:

\begin{definition}
    We refer to a quantum code $\Enc$ on $t$ registers as \emph{single share encrypting} with error $\delta$ if the reduced density matrix on each register is independent of the message: $\forall \psi_0, \psi_1$ and $i\in [t]$, 
    \begin{equation}
        \|\Tr_{\neg i}\Enc(\psi_0)-\Tr_{\neg i}\Enc(\psi_1)\|_1\leq \delta
    \end{equation}
\end{definition}

We prove that each share of a split-state quantum tamper-detection code must be encrypted:

\begin{theorem}[\cref{theorem:results-encryption}, restatement]\label{theorem:ttd-ss}
    Any $t$-split quantum tamper-detection code against $\lo^t$ with error $\epsilon$ must be single share encrypting with error $\Delta \leq 
4 \cdot \epsilon^{1/2}$.
\end{theorem}

As previously discussed, classical non-malleable codes with $3$ or more shares do \emph{not} satisfy single share encryption. However, it still remains open whether quantum non-malleable codes do. Nevertheless, by combining \Cref{theorem:ttd-ss} with our ``reduction" described in \cref{subsection:overview}, it turns out \cref{theorem:ttd-ss} does tell us something about quantum non-malleability. In particular, one can always convert a quantum non-malleable code into one which encrypts its shares, with arbitrarily small changes to its rate and error, by ``padding" the message state using $\lambda$ random bits, which are ignored during decoding.

To prove \Cref{theorem:ttd-ss}, we follow a proof approach is analogous to that of \cite{Barnum2001AuthenticationOQ}. At a high level, we leverage an equivalence between distinguishing two states $\ket{a}, \ket{b}$ (even with some tiny bias) and mapping $\frac{1}{\sqrt{2}}(\ket{a}+\ket{b})$ to $\frac{1}{\sqrt{2}}(\ket{a}-\ket{b})$. We make fundamental use of a lemma by \cite{Gunn2022CommitmentsTQ} (and shown similarly in \cite{Aaronson2020OnTH}):

\begin{lemma}[\cite{Gunn2022CommitmentsTQ}, Lemma 6.6.ii\label{lemma:map-to-distinguish}]
    Suppose a distinguisher $D$ implemented using a binary projective measurement $(\Pi, \mathbb{I}-\Pi)$ distinguishes $\ket{x}, \ket{y}$ with bias $\Delta$, i.e. $\vert \bra{x}\Pi\ket{x} - \bra{y}\Pi\ket{y} \vert \geq \Delta$. Then, the unitary (reflection) $U = \mathbb{I}-2\Pi$ maps the between the following superpositions with fidelity bias:

    \begin{equation}
       \bigg|\bigg(\frac{\bra{x}-\bra{y}}{\sqrt{2}}\bigg)U\bigg(\frac{\ket{x}+\ket{y}}{\sqrt{2}}\bigg)\bigg|\geq \Delta
    \end{equation}
\end{lemma}

As a consequence, we reason that if one of the $t$ adversaries is able to distinguish information about the message using just \textit{their share}, then they can also map between messages with some non-negligible advantage. 

\begin{proof} 

[of \Cref{theorem:ttd-ss}] For the purpose of contradiction, suppose two messages $\ket{a'}, \ket{b'}$ are distinguishable on one of their shares, with statistical distance $\geq 4\Delta$. Then, by the triangle inequality and an averaging argument, there exists a pair of orthogonal states $\ket{a}, \ket{b}$ with statistical distance $\geq 2\Delta$ on that share. We thus restrict our attention to the distiguishability of orthogonal states. 

If $\|\Tr_{\neg i} \Enc( \ket{a}) - \Tr_{\neg i} \Enc( \ket{b})\|_1\geq 2\Delta$ for some $i\in [t]$, then let $(\Pi_D, \mathbb{I}-\Pi_D)$ be the measurement which optimally distinguishes the two reduced density matrices on the $i$th share with bias $\geq \Delta$. By \Cref{lemma:map-to-distinguish}, that implies the $i$th adversary can map between $\Enc(\psi_+), \Enc(\psi_-)$ with fidelity $\geq \Delta$ using $U_i = \mathbb{I}-2\Pi$, where $\ket{\psi_\pm} = \frac{1}{\sqrt{2}}(\ket{a}\pm\ket{b})$. That is, 
\begin{equation}
\fid(\Enc(\psi_-), (\mathbb{I}_{\neg i}\otimes U_i)\circ \Enc(\psi_+))) \geq \Delta
\end{equation}

\noindent By monotonicity of fidelity, we conclude that after applying the decoding channel, $i$ is able to swap the messages with non-trivial bias:
\begin{gather}
   \fid(\psi_-, \Dec\circ (\mathbb{I}_{\neg i}\otimes U_i)\circ \Enc(\psi_+))) = \\ =   \fid(\Dec\circ \Enc(\psi_-), \Dec\circ (\mathbb{I}_{\neg i}\otimes U_i)\circ \Enc(\psi_+))) \geq \Delta.
\end{gather}

By the the tamper detection definition, there must exist $p\in [0, 1]$ such that
\begin{gather}
    \Dec\circ (\mathbb{I}_{\neg i}\otimes U_i)\circ \Enc(\psi_+)) \approx_\epsilon p\cdot \psi_+ + (1-p)\cdot \bot\\\Rightarrow \fid^2(\psi_-, \Dec\circ (\mathbb{I}_{\neg i}\otimes U_i)\circ \Enc(\psi_+))) = \\ =\bra{\psi_-}\Dec\circ (\mathbb{I}_{\neg i}\otimes U_i)\circ \Enc(\psi_+)\ket{\psi_-}\leq \epsilon
\end{gather}
\noindent which violates the tamper detection definition unless $\Delta \leq \epsilon^{1/2}$.
\qed \end{proof}

%\section*{Acknowledgements}

%Removed for Submission

%We thank João Ribeiro for insights on his compilers for secret sharing schemes, Venkat Guruswami and Rahul Jain for discussions on the capacity of split-state classical and quantum non-malleable codes, and Marshall Ball for conversations on augmented quantum non-malleable codes during the early stages of this work. Finally, Umesh Vazirani and Nathan Ju for comments on the manuscript and presentation.

%TB acknowledges support by the National Science Foundation Graduate Research Fellowship under Grant No. DGE 2146752.

 \bibliographystyle{splncs04}
 \bibliography{references}

\appendix
\section{Background on Pauli and Clifford Operators}
\label{section:cliffords}

Here we review Pauli operators and the associated Pauli and Clifford groups.
\begin{definition}[Pauli Operators]\label{def:pauli}
The single-qubit \emph{Pauli operators} are given by
\[ I = \begin{pmatrix} 1 & 0 \\ 0 & 1\end{pmatrix} \quad X = \begin{pmatrix} 0 & 1 \\ 1 & 0\end{pmatrix} \quad Y = \begin{pmatrix} 0 & -i \\ i & 0 \end{pmatrix} \quad Z = \begin{pmatrix} 1 & 0 \\ 0 & -1 \end{pmatrix}.\]

 An $n$-qubit Pauli operator is given by the $n$-fold tensor product of single-qubit Pauli operators. 
 We denote the set of all $\vert A \vert$-qubit Pauli operators on $\cH_A$ by  $\cP(\cH_A)$, where $\vert \cP(\cH_A)\vert =4^{\vert A \vert}$. Any linear operator $L \in \cL(\cH_A)$ can be written as a linear combination of $\vert A \vert$-qubit Pauli operators with complex coefficients as $L = \sum_{P \in \mathcal{P}(\cH_A)} \alpha_P P$. This is called the \emph{Pauli decomposition} of a linear operator.
\end{definition}

We remark that for $a\in \mathbb{F}_2^n$, we refer to the $n$-qubit Pauli operator $X^a = \otimes_{i\in [n]} X^{a_i}$ (respectively $Z^a$).

\begin{definition}[Pauli Group]
The single-qubit \emph{Pauli group} is given by
\begin{equation*}
    \{ +P, -P, \ iP, \ -iP : P \in \{ I, X, Y, Z\} \}.
\end{equation*}
The Pauli group on $\vert A \vert$-qubits is the group generated by the operators described above applied to each of $\vert A \vert$-qubits in the tensor product. We denote the $\vert A \vert$-qubit Pauli group on $\cH_A$ by  $\tilde{\cP}(\cH_A)$.
    
\end{definition}
\begin{definition}[Clifford Group]\label{def:clifford}
The \emph{Clifford group} $\mathcal{C}(\cH_A)$ is defined as the group of unitaries that normalize the Pauli group $\tilde{\cP}(\cH_A)$, i.e.,
\begin{equation*}
    \mathcal{C}(\cH_A) = \{ V \in \mathcal{U}(\cH_A) : V \tilde{\cP}(\cH_A) V^\dagger =\tilde{\cP}(\cH_A)\}.
\end{equation*}
The \emph{Clifford unitaries} are the elements of the Clifford group.

\end{definition}

We will also need to work with subgroups of the Clifford group with certain special properties.
The following fact describes these properties and guarantees the existence of such subgroups.
\begin{fact}[Restatement of \cref{lem:subclifford}~\cite{CLLW16}]\label{lem:subclifford-appendix}
There exists a subgroup $\mathcal{SC}(\cH_A)$ of the Clifford group $\mathcal{C}(\cH_A)$ such that given any non-identity Pauli operators $P ,Q \in \cP(\cH_A)$ we have that
\begin{equation*}
    \vert \{ C \in \cSC(\cH_A) \vert C^\dagger P C =Q \} \vert = \frac{\vert \cSC(\cH_A)  \vert}{\vert\cP(\cH_A) \vert -1} \quad \textrm{and} \quad \vert \cSC(\cH_A)  \vert = 2^{5 \vert A \vert }-2^{3 \vert A \vert}.
\end{equation*} 
Informally, applying a random Clifford operator from  $\mathcal{SC}(\cH_A)$ (by conjugation) maps $P$ to a Pauli operator chosen uniformly at random over all non-identity Pauli operators. 
Furthermore, we have that $\cP(\cH_A)\subset \cSC(\cH_A)$.

Additionally, there exists a procedure $\samp$ which given as input a uniformly random string $R\leftarrow \bits^{5|A|}$ outputs in time $\poly(|A|)$ a Clifford operator $C_R\in\mathcal{SC}(\cH_A)$ such that
\begin{equation}\label{eq:approxsample}
    C_R \approx_{2^{-2|A|}} U_{\mathcal{SC}(\cH_A)},
\end{equation}
where $U_{\mathcal{SC}(\cH_A)}$ denotes the uniform distribution over $\cSC(\cH_A)$.
\end{fact}

\paragraph{Twirling and related facts.}

The analysis of our construction will require the use of several facts related to Pauli and Clifford twirling. We collect them below, beginning with the usual version of the Pauli twirl.

\begin{fact}[Pauli $1$-Design]\label{fact:bellbasis}
    Let $\rho_{AB}$ be a state. Then,
$$\frac{1}{\vert \cP(\cH_A) \vert} \sum_{Q \in \cP(\cH_A) } (Q \otimes \id)   \rho_{A B}  ( Q^\dagger \otimes \id  )  = U_{A } \otimes  \rho_B.$$
\end{fact}

\begin{fact}[Pauli Twirl~\cite{DCEL09}]\label{lem:paulitwirl}
 Let $\rho \in \cD( \cH_A)$ be a state and $P , P' \in \cP(\cH_A)$ be Pauli operators such that $P \ne P'$. Then,
$$\sum_{Q \in \cP(\cH_A)}   Q^\dagger P Q \rho Q^\dagger P'^\dagger Q  =0 .$$
\end{fact}

\begin{fact}[$1$-Design]\label{fact:notequal-appendix}     Let $\rho_{AB}$ be a state. Let $\cSC(\cH_A)$ be the subgroup of Clifford group as defined in \cref{lem:subclifford-appendix}. Then,
$$\frac{1}{\vert \cSC(\cH_A) \vert} \sum_{C \in \cSC(\cH_A) } (C \otimes \id)   \rho_{A B}  ( C^\dagger \otimes \id  )  = U_{A } \otimes  \rho_B.$$
\end{fact}

\begin{fact}[Clifford Subgroup Twirl, Lemma 1 in~\cite{BBJ23}]\label{lem:equalreq} Let state $\rho_{AB}$ be a state. Let $P, Q \in \cP(\cH_A)$ be any two Pauli operators. Let $\cSC(\cH_A)$ be the sub-group of Clifford group as defined in \cref{lem:subclifford-appendix}. 
\begin{enumerate}
    \item If $P  \ne Q$, then 
\begin{equation}
   \frac{1}{\vert \cSC(\cH_A)\vert } \sum_{C \in \cSC(\cH_A)} (\id  \otimes C^\dagger P C)   \rho_{BA }  ( \id \otimes C^\dagger Q^\dagger C  ) = 0. 
\end{equation}
\item If $P =Q \ne \id_A$, then  
\begin{equation}
    \frac{1}{\vert \cSC(\cH_A)\vert } \sum_{C \in \cSC(\cH_A)} (\id  \otimes C^\dagger P C)   \rho_{BA }  ( \id \otimes C^\dagger Q^\dagger C  ) = \frac{\vert \cP(\cH_A)\vert (  \rho_{B} \otimes U_A ) -\rho_{BA} }{\vert \cP(\cH_A)\vert-1}.
\end{equation}
\end{enumerate}

\end{fact}

\suppress{

\begin{remark}
    \cref{lem:cliffordtwirl1} and \cref{lem:equalreq} are stated in \cite{BBJ23} when $\rho_{AB}$ is a pure state and $B$ is the canonical purification register of $A$. However, their proofs hold for a general state $\rho_{AB}$.
\end{remark}

\begin{fact}[Uniform Pauli conjugation]\label{lem:paulitwirl2}
  Let $P \in \cP(\cH_A)$ be a Pauli operator. If $P=\id_A$, then 
  $$\frac{1}{\vert \cP(\cH_A) \vert} \sum_{Q \in \cP(\cH_A)}    Q P Q^\dagger  =\id_A ,$$
  else if $P \ne \id_A$
 $$\frac{1}{\vert \cP(\cH_A) \vert} \sum_{Q \in \cP(\cH_A)}    Q P Q^\dagger  =0 .$$  
\end{fact}

\begin{fact}[Uniform conjugation]\label{lem:cliffordtwirl2}
  Let $P \in \cP(\cH_A)$ be a Pauli operator. If $P=\id_A$, then 
  $$\frac{1}{\vert \cSC(\cH_A) \vert} \sum_{C \in \cSC(\cH_A)}    C P C^\dagger  =\id_A .$$
  Else, if $P \ne \id_A$, then
 $$\frac{1}{\vert \cSC(\cH_A) \vert} \sum_{C \in \cSC(\cH_A)}    C P C^\dagger  =0 .$$  
\end{fact}

}

% \suppress{
% \begin{fact}[Clifford randomization~\cite{ABE08}]\label{lem:cliffordrand}Let $P ,Q \in \cP(\cH_A)$ be non identity Pauli operators. Then, 
%     $$ \vert \{ C \in \cC(\cH_A) \vert C^\dagger P C =Q \} \vert = \frac{\vert \cC(\cH_A)  \vert}{\vert\cP(\cH_A) \vert -1}.$$Informally, applying a random Clifford operator (by conjugation) maps $P$ to a Pauli operator chosen uniformly over all non-identity Pauli operators.
% \end{fact}
% }

\begin{lemma}[Restatement of~\Cref{lem:twirl-wsi}]
 Consider \cref{fig:splitstate21-background}.  Let $\psi_R=U_R$ be a state independent of ${\psi}_{A \hat{A}E}$ and $\vert R \vert =5 \vert A \vert$. Let  $ \Lambda:\cD(\cH_A\otimes \cH_{E})\to \cD(\cH_{A}\otimes\cH_{E})$ be any CPTP map. Let $\cSC(\cH_A)$ be the sub-group of Clifford group as defined in \cref{lem:subclifford-appendix}. Let \[ \rho_{\hat{A}AE} = \frac{1}{\vert \cSC(\cH_A)\vert } \sum_{C \in \cSC(\cH_A)}  (   C^\dagger \Lambda ( C  (\psi_{ \hat{A}AE}   ) C^\dagger )   C) .\]Then,
 \[  \rho_{\hat{A}AE}  \approx_{\frac{2}{2^{2 \vert A \vert}-1}}  \Phi_1 ( \psi_{\hat{A}AE}) + (\Phi_2( \psi_{\hat{A}E}) \otimes U_A),\]
 where $\Phi_1, \Phi_2 :\cL(\cH_E)\to \cL(\cH_{E}) $ are CP (completely positive) maps acting only on register $E$, depending only on $\Lambda$, and $\Phi_1 (.)+\Phi_2 (.)$ is a CPTP map.
\end{lemma}

\begin{figure}[h]
\centering
\resizebox{10cm}{5cm}{
\begin{tikzpicture}
%\draw (1,1.2) ellipse (0.3cm and 2.5cm);
  %        \node at (1,3.3) {$R$};
  %          \node at (5.2,3.7) {$R$};
   %      \node at (14.5,3.7) {$R$};
   %      \draw (1.2,3.5) -- (3,3.5);
   %       \draw (5,3.5) -- (15,3.5);
%\draw (10,3.5) -- (15,3.5);
%\node at (1,-0.8) {$\hat{R}$};
%\node at (14.5,-0.8) {$\hat{R}$};
%\draw (1.2,-1) -- (15,-1);

%\draw (1,5) ellipse (0.2cm and 1cm);
\node at (1,6.5) {$\hat{A}$};
\node at (13.2,6.7) {$\hat{A}$};
\draw (1.2,6.5) -- (13,6.5);
\draw (1,4.8) ellipse (0.3cm and 2.2cm);
\draw (2.6,3.9) rectangle (4,5.4);
\draw (9.3,3.9) rectangle (10.7,5.4);
\node at (1,4.5) {$A$};

%\draw (0.8,0.5) ellipse (0.3cm and 1.5cm);
%\node at (0.8,-0.7) {$\hat{R}$};
\node at (3.3,4.6) {$C_R$};
\node at (10.1,4.6) {$C^\dagger_{R}$};
%\node at (1,6.5) {$\hat{M}$};
%\node at (14.8,6.5) {$\hat{M}$};

%\draw (1.2,6.5) -- (14.6,6.5);

%\node at (14.5,4.3) {$S'$};
\draw (1.2,4.5) -- (2.6,4.5);
%\draw (2.4,4.2) -- (2.6,4.2);
%\node at (2.1,4) {$\ket{0}^d$};
\draw (4,4.5) -- (6,4.5);
\draw (7,4.5) -- (9.3,4.5);
\draw (10.7,4.9) -- (13.2,4.9);
%\draw (10.7,4.5) -- (11,4.5);

%\draw (10.7,3) rectangle (11.5,4.4);

%\node at (11.7,4.5) {$\{  \ketbra{0^d}{0^d}, $};
%\node at (11.6,4) {$ \id -  \ketbra{0^d}{0^d}  \}$};

%\draw (10,4.5) -- (15,4.5);
%\node at (5.2,4.7) {$S$};

\node at (0.8,0.5) {$R$};
%\node at (1.5,1.5) {$S$};
\draw (1,0.5) -- (10,0.5);
%\node at (2.0,0) {$\ketbra{0^{n}}$};
%\draw (2.8,0) -- (3,0);
%\draw (3,-0.1) rectangle (4.5,3.3);
%\draw (1.5,-0.5) rectangle (4.8,5.9);
%\node at (3.3,6.1) {$\enc$};
%\node at (3.8,1.5) {$\nmcenc$};

%\draw (9.4,-0.5) rectangle (14.2,5.9);
%\node at (11.5,6.1) {$\dec$};

\draw  (3.2,0.5) -- (3.2,3.9);

\draw [dashed] (1.4,-1) -- (1.4,7.2);
%\draw [dashed] (3.5,-1.4) -- (3.5,7.2);
%\draw [dashed] (9.3,-1.4) -- (9.3,7.2);
%\draw [dashed] (10,-1.4) -- (10,7.2);
\draw [dashed] (12.8,-1) -- (12.8,7.2);
%\draw [dashed] (14.3,-0.8) -- (14.3,5.5);

%\draw (0.9,-1) -- (13.2,-1);

%\node at (6.2,1.6) {$\ket{\psi}_{W_1W_2}$};
\node at (1.25,-1.4) {$\psi$};
%\node at (3.35,-1.4) {$\rho$};
%\node at (9.15,-1.4) {$\tau$};
%\node at (10.2,-1.4) {$\sigma_3$};
\node at (12.6,-1.4) {$\rho$};
%\node at (14.13,-0.8) {$\rho$};

%\node at (4.6,3) {$E$};
\node at (1,2.8) {$E$};

\draw (1.2,2.8) -- (6,2.8);
%\draw (7,2.8) -- (8.5,2.8);
%\draw (9.5,2.8) -- (10.5,2.8);
\draw (7,2.8) -- (13.1,2.8);

%\node at (12.7,2.6) {$R'$};
%\draw (12.5,2.4) -- (13.5,2.4);
\draw (10,0.5) -- (10,3.9);

%\node at (4.6,0.4) {$Y$};
%\node at (8,0.0) {$Y'$};
%\draw (4.5,0.2) -- (6,0.2);
%\draw (7,0.2) -- (8.5,0.2);

%\draw (9.5,0.2) -- (10.5,0.2);
%\draw (6,2.5) -- (5.7,2.5);
%\draw (7,2.5) -- (7.3,2.5);
%\node at (5.6,2.3) {$E'$};
%\node at (7.4,2.3) {$E'$};

\draw (6,2.3) rectangle (7,5);
\node at (6.5,3.5) {$ \Lambda$};
%\draw (6,0) rectangle (7,1);
%\node at (6.5,0.5) {$V$};

\node at (6.5,5.5) {$\mathcal{A}$};
\draw (5.2,2) rectangle (7.8,6);

%\draw (5.5,1.5) ellipse (0.2cm and 1cm);
%\node at (5.5,2) {$W_1$};
%\draw (5.7,2.2) -- (6,2.2);
%\node at (7.4,2.1) {$W'_1$};
%\draw (7,2.2) -- (7.2,2.2);
%\node at (5.5,1) {$W_2$};
%\draw (5.7,0.7) -- (6,0.7);
%\node at (7.4,0.9) {$W'_2$};
%\draw (7.0,0.7) -- (7.2,0.7);

%\draw (9,2.8) circle (0.5);
%\node at (9,2.8) {$\mathcal{M}$};
%\draw (9,0.2) circle (0.5);
%\node at (9,0.2) {$\mathcal{M}$};

%\draw (8,-0.5) rectangle (10,1.2);
%\draw (8,1.5) rectangle (10,3.2);
%\node at (9,1.5) {$2\mhyphen\nmext$};

%\draw (10.5,-0.2) rectangle (12.5,3.0);
%\node at (11.5,1.5) {${\nmcdec}$};
%\node at (12.5,1.6) {$ \mathcal{M}= \{ \ketbra{0^{2n}},$};
%\node at (12.5,1) {$ \id- \ketbra{0^{2n}} \}$};
\node at (13.2,2.8) {$E$};
\end{tikzpicture}}

\caption{Clifford Twirling with Side Information}\label{fig:splitstate21-background}
\end{figure}

\newpage

\begin{proof} 
Let $\Lambda  : \cL(\cH_A \otimes \cH_{E}) \to  \cL(\cH_A \otimes \cH_{E})$ be the CPTP map. Let $\{ M_i\}_i$ be the set of Kraus operators corresponding to $\Lambda$, and its Pauli basis decomposition $M_i=\sum_j \alpha_{ij}P^{ij}$ for $P^{ij}\in \cP(\cH_{AE})$. We denote $P^{ij}_A$ to use Pauli operator corresponding to register $A$ of $P^{ij}$. 
Let $M^Q_i \defeq  \sum_{j :  P^{ij}_A =Q } \alpha_{ij} P^{ij}_E$  for every $Q \in \cP(\cH_A)$. Note that since $\Lambda$ is CPTP, we have $\sum_iM_i^\dagger M_i=\id_{AE}$. 

We begin by showing that if we restrict the $M_i$ to their operation on the register $E$, they still form a CPTP map. For this we need to show $
\sum_{i,Q} (M^Q_i)^\dagger M^Q_i =\id_E$. Consider, 
\begin{align*}
&2^{\vert A \vert }\id_{E} = \Tr_A(\id_{AE})\\
& = \Tr_A \left( \sum_{i} M_i^\dagger M_i \right)  \\
&= \Tr_A \left( \sum_{i} \left(\sum_{j'    } \alpha^*_{ij'} P^{ij'} \right) \left(\sum_{j   } \alpha_{ij} P^{ij} \right) \right) \\&=  \Tr_A \left(\sum_{i,j,j'  } \alpha_{ij} \alpha^*_{ij'} (  P^{ij'}_A  P^{ij}_A ) \otimes ( P^{ij'}_E  P^{ij}_E ) \right) \\
& = \left(\sum_{i,j,j' : P^{ij}_A= P^{ij'}_A  } \alpha_{ij} \alpha^*_{ij'} 
 \Tr(\id_A) \otimes P^{ij'}_E  P^{ij}_E +\sum_{i,j,j' : P^{ij}_A\ne P^{ij'}_A  } \alpha_{ij} \alpha^*_{ij'} 
  \Tr( P^{ij'}_A  P^{ij}_A ) \otimes P^{ij'}_E  P^{ij}_E  \right) \\
  & = 2^{\vert A \vert }\left(\sum_{i,j,j' : P^{ij}_A= P^{ij'}_A  } \alpha_{ij} \alpha^*_{ij'} 
 P^{ij'}_E  P^{ij}_E  \right) \\
 & =  2^{\vert A \vert } \sum_{i,Q} \left(\sum_{j,j' : P^{ij}_A= P^{ij'}_A =Q } \alpha_{ij} \alpha^*_{ij'} P^{ij'}_E  P^{ij}_E  \right)   \\
 & = 2^{\vert A \vert } \sum_{i,Q} \left(\sum_{j' :  P^{ij'}_A =Q } \alpha^*_{ij'} P^{ij'}_E \right) \left(\sum_{j :  P^{ij}_A =Q } \alpha_{ij} P^{ij}_E \right)  \\
 & =  2^{\vert A \vert } \sum_{i,Q} (M^Q_i)^\dagger (M^Q_i) 
\end{align*}

We can now turn our attention to $\rho_{\hat{A}AE}$. Consider,

\begin{align*}
   &\rho_{\hat{A}AE}= \frac{1}{\vert \cSC(\cH_{A})\vert } \sum_{C \in \cSC(\cH_{A})}     C^\dagger \Lambda(C\psi_{ \hat{A}AE}C^\dagger )  C  \\ 
   &=  \sum_{i}  \left(\frac{1}{\vert \cSC(\cH_A)\vert } \sum_{C \in \cSC(\cH_A)}  (   C^\dagger M_iC)  (\psi_{ \hat{A}AE} ) (C^\dagger  M_i^\dagger C)  \right)  
  \\ &=  \sum_{i}  \left(\frac{1}{\vert \cSC(\cH_A)\vert }\sum_{j,j'}\alpha_{ij}\alpha^*_{ij'} \sum_{C \in \cSC(\cH_A)}   (   C^\dagger P^{ij}C)  (\psi_{ \hat{A}AE} ) (C^\dagger  P^{ij'} C)  \right)   \\ 
  &=  \sum_{i,j,j' : P^{ij}_A =P^{ij'}_A }  \left(\alpha_{ij}\alpha^*_{ij'}\frac{1}{\vert \cSC(\cH_A)\vert }  P^{ij}_E \left(\sum_{C \in \cSC(\cH_A)}   (   C^\dagger P^{ij}_A C)  (\psi_{ \hat{A}AE} ) (C^\dagger  P^{ij'}_A C) \right)  P^{ij'}_E \right)  \\ &+  \sum_{i,j, j' :  P^{ij}_A \ne P^{ij'}_A}  \left(\alpha_{ij}\alpha^*_{ij'}\frac{1}{\vert \cSC(\cH_A)\vert }  P^{ij}_E\left( \sum_{C \in \cSC(\cH_A)}   (   C^\dagger P^{ij}_A C)  (\psi_{ \hat{A}AE} ) (C^\dagger  P^{ij'}_A C) \right)  P^{ij'}_E \right)  \\
&=   \sum_{i,j,j' : P^{ij}_A =P^{ij'}_A }  \left(\alpha_{ij}\alpha^*_{ij'}\frac{1}{\vert \cSC(\cH_A)\vert }  P^{ij}_E \left(\sum_{C \in \cSC(\cH_A)}   (   C^\dagger P^{ij}_A C)  (\psi_{ \hat{A}AE} ) (C^\dagger  P^{ij'}_A C) \right)  P^{ij'}_E \right) &\mbox{(\cref{lem:equalreq}.1)}\\
    &=  \sum_{i,j,j' : P^{ij}_A =P^{ij'}_A =\id_A }  \left(\alpha_{ij}\alpha^*_{ij'}\frac{1}{\vert \cSC(\cH_A)\vert }  P^{ij}_E \left(\sum_{C \in \cSC(\cH_A)}   (   C^\dagger P^{ij}_A C)  (\psi_{ \hat{A}AE} ) (C^\dagger  P^{ij'}_A C) \right)  P^{ij'}_E \right)  \\ &+  \sum_{i,j,j' : P^{ij}_A =P^{ij'}_A \ne \id_A }  \left(\alpha_{ij}\alpha^*_{ij'}\frac{1}{\vert \cSC(\cH_A)\vert }  P^{ij}_E \left(\sum_{C \in \cSC(\cH_A)}   (   C^\dagger P^{ij}_A C)  (\psi_{ \hat{A}AE} ) (C^\dagger  P^{ij'}_A C) \right)  P^{ij'}_E \right) \\
    &=  \sum_{i,j,j' : P^{ij}_A =P^{ij'}_A =\id_A }  \left(\alpha_{ij}\alpha^*_{ij'} P^{ij}_E  (\psi_{ \hat{A}AE} )  P^{ij'}_E \right)   \\ &+  \sum_{i,j,j' : P^{ij}_A =P^{ij'}_A \ne \id_A }  \left(\alpha_{ij}\alpha^*_{ij'} P^{ij}_E \left(  \frac{\vert \cP(\cH_A)\vert ( \psi_{\hat{A}E} \otimes U_A ) -\psi_{\hat{A}EA} }{\vert \cP(\cH_A)\vert-1}  \right) P^{ij'}_E \right) &\mbox{(\cref{lem:equalreq}.2)} \\
    &  \approx_{\frac{2}{ \vert \cP(\cH_A) \vert-1}}  \sum_{i}  M^{\id_A}_i \psi_{A\hat{A}E} (M^{\id_A}_i)^\dagger +   \sum_{i}\sum_{Q \in \cP(\cH_A) \land Q \ne \id_A}  M^{Q}_i ( U_A \otimes \psi_{\hat{A}E}) (M^{Q}_i)^\dagger .
\end{align*}

 The approximation in the result follows from the fact that
\begin{equation*}
\left\|\frac{\vert \cP(\cH_A)\vert ( U_{A} \otimes \psi_{\hat{A}E} ) -\psi_{A\hat{A}E} }{\vert \cP(\cH_A)\vert-1}   - ( U_{A} \otimes \psi_{\hat{A}E} )   \right\|_1 \leq \frac{2}{ \vert \cP(\cH_A) \vert-1}.
\end{equation*}

\noindent and the monotonicity of trace distance \cref{fact:data}, since the operators $\{M_i^Q\}_{i, Q}$ define a CPTP set of Krauss operators. Finally, we let $\{ M^{\id}_i\}_{i, \id}$ define the CP map $\Phi_1$ and $\{ M^{Q}_i\}_{i,Q\ne  \id}$ define $\Phi_2$.

\qed \end{proof}

\newpage

\section{Secret Sharing Schemes Resilient to Joint Quantum Leakage}
\label{section:lrss}

In this section, we show that simple modifications to a recent construction of leakage resilient secret sharing schemes by \cite{CKOS22} can be made secure against quantum leakage, even when the leakage adversaries are allowed to jointly leak a quantum state from an unauthorized subset (of size $k$) to another (of size $< t$). We refer to this leakage model as $\mathcal{F}_{k, \mu}^{n, t}$. The main result of this section 

\begin{theorem}\label{theorem:lrss-local-leakage}
    For every $k<t\leq p <l, \mu\in \mathbb{N}$ there exists an $(p, t, 0, 0)$ threshold secret sharing scheme on messages of $l$ bits and shares of size $l+\mu +o(l, \mu)$ bits, which is perfectly private and $p\cdot 2^{-\Tilde{\Omega}(\sqrt[3]{\frac{l+\mu}{p}})}$ leakage resilient against the $k$ local $\mu$ qubit leakage family $\mathcal{F}_{k, \mu}^{p, t}$.
\end{theorem}

We organize the rest of this section as follows. In \Cref{subsubsection:lrss-prelim}, we present the relevant secret sharing definitions, and recall the relevant background on quantum secure extractors. In \Cref{subsection:lrss-construction}, we present the code construction, and in \cref{subsection:lrss-analysis} its proof of security. Finally, we instantiate our construction using specific secret sharing schemes and extractors in \Cref{subsubsection:lrss-params}.

\subsection{Preliminaries}
\label{subsubsection:lrss-prelim}

\subsubsection{Leakage Resilient Secret Sharing}
\label{subsubsection:lrss-prelim-local}

We refer the reader to \Cref{subsubsection:prelim-lrss} for a more comprehensive background on secret sharing. 

\begin{definition}
[Leakage-Resilient Secret Sharing] Let $(\share, \rec)$ be a secret sharing scheme with randomized sharing function $\Share:\mathcal{M}\rightarrow \{\{0, 1\}^{l'}\}^p$, and let $\mathcal{F}$ be a family of leakage channels. Then $\share$ is said to be $(\mathcal{F}, \epsilon_{lr})$ leakage-resilient if, for every channel $\Lambda\in \mathcal{F}$, 
    \begin{equation}
       \forall m_0, m_1\in \mathcal{M}: \text{ } \Lambda(\Share(m_0)) \approx_{\epsilon_{lr}}  \Lambda(\Share(m_1))
    \end{equation}
\end{definition}

\begin{definition}
[Quantum $k$ local leakage model]
    For any integer sizes $p, t, k$ and leakage length (in qubits) $\mu$, we define the $(p, t, k, \mu)$-local leakage model to be the collection of channels specified by
    \begin{equation}
    \mathcal{F}_{k, \mu}^{p, t} =\bigg\{ (T, K, \Lambda): T, K\subset [p],|T|< t, |K|\leq k, \text{ and }\Lambda:\{0, 1\}^{l'\cdot |K|}\rightarrow \cL(\mathcal{H}_\mu)\bigg\}
\end{equation}

\noindent Where $\log \dim(\mathcal{H}_\mu) = \mu$. A leakage query $(T, K, \Lambda)\in \mathcal{F}_{k, \mu}^{p, t}$ on a secret $m$ is the density matrix:
\begin{equation}
    (\mathbb{I}_{T}\otimes \Lambda_K)(\Share(m)_{T \cup K})
\end{equation}
\end{definition}

In their constructions \cite{CKOS22} leverage secret sharing schemes augmented to satisfy a ``local uniformity" property, where individual shares given out by the Share function are statistically close to the uniform distribution over the share space. To extend their construction to $k$ local tampering, we require $k$ wise independence:

\begin{definition}
A randomized sharing function $\Share:\{0, 1\}^l\rightarrow \{\{0, 1\}^{l'}\}^p$ is $\epsilon_u$-approximately $k$ wise independent if for every message $m\in \{0, 1\}^l$ and subset $S\subset [p]$ of size $k$:
    \begin{equation}
        \Share(m)_S \approx_{\epsilon_u}  U_{l'}^{\otimes k}
    \end{equation}
\end{definition}

We note that $(p, t)$ Shamir Secret Sharing \cite{Sha79} is exactly $(t-1)$ wise independent.

\subsubsection{Quantum Min Entropy and Randomness Extractors}
\label{subsubsection:lrss-prelim-extractors}

\begin{definition}
    [Quantum conditional min-entropy] Let $X, Y$ be registers with state space $\mathcal{X} , \mathcal{Y}$
and joint state $\rho$. We define the conditional min-entropy of $X$ given $Y$ as

\begin{equation}
    H_\infty (X|Y )_\rho = - \log \min_{\sigma\in \mathcal{Y}} \{\min_{\lambda\in \mathbb{R}} \lambda\cdot \mathbb{I}\otimes \sigma \geq \rho\}.
\end{equation}
\end{definition}

When $\rho$ is a cq-state, $H_\infty (X|Y )_\rho$ has an operational meaning in terms of the optimal guessing probability for $X$ given $Y$. We remark that product states $\rho=\tau_X \otimes \sigma_Y$ have conditional min entropy $H_\infty (X|Y )_\rho = H_\infty (X)_\rho =-\log \lambda_{max}(\tau)$ equal to the log of the largest eigenvalue of $\tau_X$. When $\rho$ is separable, it satisfies a Chain rule:

\begin{lemma}
    [Separable chain rule for quantum min-entropy \cite{Desrosiers2007QuantumES}, Lemma 7]\label{lemma:sep-chain-rule} Let $A, B, C$ be registers with some joint, separable state $\rho = \sum_i \tau^{AB}_i\otimes \sigma^C_i$. Then,
    \begin{equation}
        H_\infty (A|B, C )_\rho\geq H_\infty (A|B )_\rho -  \log |C|
    \end{equation}
\end{lemma}

\begin{definition}
    [Quantum-proof seeded extractor \cite{DPVR09}]\label{definition:quantum-proof-extractor} A function $\Ext :\{0, 1\}^\eta\times \{0, 1\}^d\rightarrow \{0, 1\}^l$ is said to be a $(\eta, \tau, d, l, \epsilon_{\Ext})$-strong quantum-proof seeded extractor if for any cq-state $\rho\in \mathcal{H}^{\otimes n}\otimes \mathcal{Y}$ of the registers $X, Y$ with $H_\infty (X|Y )_\rho\geq \tau$, we have 

    \begin{equation}
        \Ext(X, S), Y, S\approx_\epsilon U_l, Y, S\text{ where } S \leftarrow \{0, 1\}^d
    \end{equation}

    Morover, if $\Ext(\cdot, s)$ is a linear function for all $s\in \{0, 1\}^d$, then $\Ext$ is called a linear seeded extractor.
\end{definition}

\begin{lemma}
    [\cite{Trevisan01,DPVR09}]\label{lemma:quantum-proof-trevisan} There is an explicit $(\eta, \tau, d, l, \epsilon)$-strong quantum-proof linear seeded extractor with $d = O(log^3(\eta/\epsilon)$ and $l = \tau - O(d)$.
\end{lemma}

We require the extractor to support efficient \textit{pre-image sampling}. Given a seed $s$ and some $y \in \{0, 1\}^l$, the inverting function $\iExt$ needs to sample an element uniformly from the set $\Ext(\cdot, s)^{-1}(y) = {w : \Ext(w; s) = y}$. \cite{CKOS22} showed that linear extractors always admit such sampling:

\begin{lemma}
    [\cite{CKOS22}]\label{lemma:ckos-inverter} For every efficient linear extractor $\Ext$, there exists an efficient randomized function $\iExt : \{0, 1\}^l \times \{0, 1\}^d \rightarrow \{0, 1\}^\eta \cup {\bot}$ (termed inverter) such that
    \begin{enumerate}
        \item $U_\eta, U_d, \Ext(U_\eta;U_d) \equiv \iExt(\Ext(U_\eta;U_d), U_d), U_d, \Ext(U_\eta;U_d)$
        \item For each $(s, y) \in \{0, 1\}^d \times \{0, 1\}^l$:
        \begin{enumerate}
            \item $\mathbb{P}[\iExt(y, s) = \bot] = 1$, if and only if there exists no $w \in \{0, 1\}^\eta$ such that $\Ext(w; s) = y$.
            \item $\mathbb{P}[\Ext(\iExt(y, s), s) = y] = 1$, if there exists some $w \in \{0, 1\}^\eta$ such that $\Ext(w; s) = y$
        \end{enumerate}
    \end{enumerate}
\end{lemma}

\subsection{Code Construction}
\label{subsection:lrss-construction}

Our code construction uses essentially the same ingredients as \cite{CKOS22}, with small modifications to the locality of their privacy parameters and to their compiler.

\begin{enumerate}
    \item $(\mshare, \mrec)$, an $(p, t,\epspriv, 0)$ threshold secret sharing scheme which is $(k, \epsilon_u)$-locally uniform over the message space $\{0, 1\}^l$ and with share size $l'$.
    \item $(\sdshare, \sdrec)$, an $(p, k+1, \epspriv', 0)$ threshold secret sharing scheme over the message space $\{0, 1\}^d$ and share size $d'$.
    \item $\Ext$, a quantum-proof $(\eta, \tau\leq \eta-\mu, d, l', \epsilon_{ext})$-strong linear extractor. Let $\iExt$ be the inverter function corresponding to $\Ext$ given by \Cref{lemma:ckos-inverter} .
\end{enumerate}

\noindent \textbf{Share} To share a message $m$, we begin by encoding it into $(M_1,\cdots, M_p)\leftarrow \mshare(m)$. Then, for each party $i\in [p]$, we sample a random seed $R_i\in \{0, 1\}^d$, and then use $\iExt$ to get the source $W_i \leftarrow \iExt(M_i, R_i)$. If any of the $W_i=\bot$ rejects, output each of the share to be $(\bot, M_i)$. Else, concatenate the randomness $R = (R_1, \cdots, R_p)$ and share it using $\sdshare(R)$ to get $(S_1, \cdots, S_p)$, and finally set the $i$-th share to be $(W_i, S_i)$. \\

\noindent \textbf{Rec} Assuming the encoding doesn't reject, an authorized party $T\subset [p]$ of size $\geq t> k$ begins by recovering the randomness $R = (R_1, \cdots, R_p)$ using $\sdrec$ on any $k+1$ (honest) shares of $T$. Then, they recover the message shares $M_T =\{M_i:i\in T\}$ by running the extractor on $(W_i, R_i)$ for each $i\in T$. Finally, using $\mrec$ on $M_T$ it decodes the message $m$.

\begin{theorem}\label{theorem:lrss-klocal}
    $(\share, \rec)$ defines a $(p, t, \epspriv, 0)$ secret sharing scheme, which is $\leq 2(\epspriv +\epspriv') +2 p\cdot (\epsilon_\Ext+\epsilon_u)$ leakage resilient against $\mathcal{F}_{ k, \mu}^{p, t}$.
\end{theorem}

\subsection{Analysis}
\label{subsection:lrss-analysis}

The correctness and privacy of the scheme are inherited from that of $\mshare, \mrec$. We analyze its rate in the next subsection \Cref{subsubsection:lrss-params}, and dedicate the rest of this subsection to a proof of security. 

Fix a leakage channel $(T, K, \Lambda)$ and a message $m$. We assume $|T|\leq t-1$ and $|K|\leq k$ are both unauthorized subsets. We proceed in a sequence of hybrids, where within the encoding map $\Share$ we replace the shares of $K$, $W_i\leftarrow \iExt(M_i, R_i)$, by a uniformly random source $W_i\leftarrow U$:\\

\noindent $\share^0(m)$:  To share a message $m$, we simply encode it into $\share(m)$.\\

\noindent $\share^1(m)$: To share a message $m$, we begin by encoding it into $(M_1,\cdots, M_p)\leftarrow \mshare(m)$. For each $i\in [n]$, sample a random seed $R_i\in \{0, 1\}^d$ and use $\iExt$ to get the source $W_i \leftarrow \iExt(M_i, R_i)$. Then, concatenate the randomness $R = (R_1, \cdots, R_p)$ and share it using $\sdshare(R)$ to get $(S_1, \cdots, S_p)$, and finally set the $i$-th share to be $(W_i, S_i)$.\\

\noindent $\share^2(m)$: To share a message $m$, we begin by encoding it into $(M_1,\cdots, M_p)\leftarrow \mshare(m)$. For each $i\in K$, sample a random seed $R_i\in \{0, 1\}^d$ and source $W_i$. For each $i\in [p]\setminus K$, sample a random seed $R_i\in \{0, 1\}^d$ and use $\iExt$ to get the source $W_i \leftarrow \iExt(M_i, R_i)$. Then, concatenate the randomness $R = (R_1, \cdots, R_p)$ and share it using $\sdshare(R)$ to get $(S_1, \cdots, S_p)$, and finally set the $i$-th share to be $(W_i, S_i)$.\\

Note that $\share^0$ differs from $\share^1$ in that it conditions on the extractor pre-image sampling succeeding. \cite{CKOS22} begin by proving that it succeeds with high probability, and thus $\share(m)\approx \share^1(m)$.

\begin{lemma}
    [\cite{CKOS22}]\label{lemma:lrss-enc-bot} For any message $m$, $\share(m) = ((\bot, M_1), \cdots,(\bot, M_p))$ with probability $\leq p(\epsilon_\Ext+\epsilon_u)$.
\end{lemma}

Moreover, note that $\share^2(m)$ is completely independent of the shares $M_i:i\in K$, and thus the reduced density matrix $\share^{2}_{T\cup K}$ only depends on the shares of $M_i:i\in T$ - where $|T|\leq t-1$. By the privacy of $\mshare$,

\begin{lemma}
    [\cite{CKOS22}]\label{lemma:lrss-privacy} For any pair of messages $m, m'$, 
    \begin{equation}
        \share^2_{T\cup K}(m)\approx_{\epspriv}\share^2_{T\cup K}(m')
    \end{equation}
\end{lemma}

It remains to argue that $\share^1(m)$ and $\share^2(m)$ are indistinguishable, given the shares in the unauthorized subset $T$ and the leakage $L$. 

\begin{lemma}
    \label{lemma:lrss-hybrid} For any message $m$,
    \begin{equation}
        (\mathbb{I}_T\otimes \Lambda_K)(\share^1_{T\cup K}(m))\approx_{\delta}(\mathbb{I}_T\otimes \Lambda_K)(\share_{T\cup K}^{2}(m))
    \end{equation}

    \noindent where $\delta \leq 2(\epspriv'+\epsilon_u+k\cdot \epsilon_\Ext)$.
\end{lemma}

By the triangle inequality and the claims above, we conclude that for all $m, m'$:
\begin{equation}
        (\mathbb{I}_T\otimes \Lambda_K)(\share_{T\cup K}(m))\approx_{\delta'} (\mathbb{I}_T\otimes \Lambda_K)(\share_{T\cup K}(m'))
\end{equation}

\noindent For some $\delta' \leq \delta + 2\epspriv + 2\cdot p (\epsilon_\Ext+\epsilon_u)\leq 2(\epspriv +\epspriv') +2 p\cdot (\epsilon_\Ext+\epsilon_u)$, which is exactly \Cref{theorem:lrss-klocal}. Now we prove \Cref{lemma:lrss-hybrid}:

\begin{proof}

    [of \Cref{lemma:lrss-hybrid}] Fix a message $m$. Consider the quantum-classical mixed state comprised of the classical shares $M_K = \{M_i:i\in K\}$ of $\mshare$, the seed $R$ and its shares $S_K = \{S_i:i\in K\}$, and the quantum leakage register $L$. Our goal will be to show that this cq density matrix is nearly independent of the ``source" register in the shares of $K$. That is, if we denote as $W_i, i\in K$ as uniformly random sources on $\eta$ bits, then it suffices to show that for some $\delta$,
    \begin{equation}\label{equation:main-ckos-lemma}
        \Lambda\big(\{S_i, \iExt(M_i, R)\}_{i\in K}\big), R, S_K, M_K \approx_{\delta} \Lambda\big(\{S_i, W_i\}_{i\in K}\big), R, S_K, M_K.
    \end{equation}

    This is since there is a CPTP map $\mathcal{N}^m$ (dependent on the message) which given $R, S_K, M_K$ and the leakage $L$ produces 
    \begin{gather}
        \mathcal{N}^m(\Lambda\big(\{S_i, \iExt(M_i, R)\}_{i\in K}\big), R, S_K, M_K) = (\mathbb{I}_T\otimes \Lambda_K)(\share^1_{T\cup K}(m)) \text{ and } \\ \mathcal{N}^m(\Lambda\big(\{S_i, W_i\}_{i\in K}\big), R, S_K, M_K) = (\mathbb{I}_T\otimes \Lambda_K)(\share^2_{T\cup K}(m)).
    \end{gather}

    \noindent Note that $\mathcal{N}^m$ simply samples shares $M_T, S_T$ consistent with $m, R$ and $M_K, S_K$.\footnote{This is also known as ``consistent resampling" \cite{CKOS22}} Thereby by monotonicity of trace distance we obtain the desired 
    \begin{equation}
        (\mathbb{I}_T\otimes \Lambda_K)(\share^1_{T\cup K}(m))\approx_{\delta}(\mathbb{I}_T\otimes \Lambda_K)(\share_{T\cup K}^{2}(m))
    \end{equation}

    It remains to show \cref{equation:main-ckos-lemma}. We begin by replacing the shares in $K$ by fixed shares independent of $R$, $\hat{S}_K\leftarrow \sdshare(0^d)_K$ using the privacy of $\sdshare$, and leveraging the $\epsilon_u$ approximate $k$-wise independence of the shares $M_K$ to replace them by the uniform distribution:
    \begin{gather}
        \Lambda\big(\{S_i, \iExt(M_i, R)\}_{i\in K}\big), R, S_K, M_K \approx_{\epspriv'}\\ \Lambda\big(\{\hat{S}_i, \iExt(M_i, R)\}_{i\in K}\big), R, \hat{S}_K, M_K, \text{ and }\\
         \Lambda\big(\{\hat{S}_i, \iExt(M_i, R)\}_{i\in K}\big), R, \hat{S}_K, M_K \approx_{\epsilon_u}\\ \Lambda\big(\{\hat{S}_i, \iExt(U_{l'}^i, R_i)\}_{i\in K}\big), R, \hat{S}_K, (U_{l'})^{\otimes k}
    \end{gather}

    Recall $W_1, \cdots, W_K$ are uniformly random $\eta$ bit sources. By \Cref{definition:quantum-proof-extractor}, since $\Ext$ is a strong linear seeded extractor, $\Ext(W_i, R_i)\approx_{\epsilon_\Ext} U_{l'}$, and moreover by \Cref{lemma:ckos-inverter}(b) from \cite{CKOS22} we have $W_i = \iExt(\Ext(W_i, R_i), R_i)$. Thus,
    \begin{gather}
        \Lambda\big(\{\hat{S}_i, \iExt(U_{l'}^i, R_i)\}_{i\in K}\big), R, \hat{S}_K, (U_{l'})^{\otimes k}  \approx_{k\cdot \epsilon_{\Ext}}\\ \Lambda\big(\{\hat{S}_i, W_i\}_{i\in K}), R, \hat{S}_K, \{\Ext(W_i, R_i)\}_{i\in K}
    \end{gather}
    
   We now replace each $\Ext(W_j, R_j)$ by $U_l$ in a sequence of hybrids, evoking \textit{quantum-proof} extractor security. This is the main modification to the proof of \cite{CKOS22}: Fix $0\leq j\leq k$, and define the collection of classical registers $Z_j$:
    \begin{equation}
        Z_j = R_{[n]\setminus \{j\}}, \hat{S}_{K}, (U_l)^{\otimes (j-1)}, \{\Ext(W_i, R_i)\}_{j< i\leq k}
    \end{equation}

    Note that $Z_j$ is independent of $W_j$. If $L$ denotes the $\mu$ qubit leakage register, then from the chain rule for the min entropy of separable states, \Cref{lemma:sep-chain-rule} \cite{Desrosiers2007QuantumES}, $H_\infty(W_j|Z_j, L)\geq \eta-\mu\geq \tau$. By \Cref{definition:quantum-proof-extractor}, 
    \begin{equation}
       \Lambda(Z_j, W_j), Z_j, R_j, \Ext(W_i, R_i) \approx_{\epsilon_\Ext} \Lambda(Z_j, W_j), Z_j, R_j, U_{l}
    \end{equation}
    
    \noindent Which implies through the triangle inequality,
    \begin{gather}
        \Lambda\big(\{\hat{S}_i, W_i\}_{i\in K}), R, \hat{S}_K, \{\Ext(W_i, R_i)\}_{i\in K} \approx_{k\cdot \epsilon_\Ext} \\\Lambda\big(\{\hat{S}_i, W_i\}_{i\in K}), R, \hat{S}_K, U^{\otimes k}.
    \end{gather}

    \noindent By once again appealing to $k$ wise independence of $M_K$ and the privacy of $S_K$, we conclude
    \begin{gather}
        \Lambda\big(\{S_i, \iExt(M_i, R)\}_{i\in K}\big), R, S_K, M_K \approx_{2(\epspriv'+\epsilon_u+k\cdot \epsilon_\Ext)} \\\Lambda\big(\{S_i, W_i\}_{i\in K}\big), R, S_K, M_K
    \end{gather}

    \noindent That is, $\delta \leq 2(\epspriv'+\epsilon_u+k\cdot \epsilon_\Ext)$. 
\qed \end{proof}

\subsubsection{Parameters}
\label{subsubsection:lrss-params}

We combine 

\begin{enumerate}
    \item $(\mshare, \mrec)$ is a $(p, t, \epspriv = 0, \epsilon_c = 0)$-Shamir secret sharing scheme for $l$ bit messages and $l$ bit shares, which is perfectly $t-1$ wise independent. 
    
    \item We set $\epsilon_\Ext = 2^{-\lambda^{-1/3}}$, and let $\Ext$ be the $(\eta= l+\mu+O(d), \tau = l+O(d), d, l, \epsilon_\Ext)$ quantum proof strong linear extractor guaranteed by \Cref{lemma:quantum-proof-trevisan}, where $d= O(\log^3 \frac{\eta}{\epsilon_\Ext}) = O(\lambda+\log^3 (l+\mu))$.

    \item $(\sdshare, \sdrec)$ is a $(p, k+1, \epspriv'=0, \epsilon'_c= 0)$-Shamir secret sharing scheme for $p\cdot d$ bit messages and $p\cdot d$ bit shares. 
\end{enumerate}

The resulting share size of $\share$ to handle $\mu$ qubits of leakage is $p\cdot d+\eta = l+\mu+O(p\cdot \lambda + p\cdot \log^3(l+\mu))$, which is $l+\mu+o(l, \mu)$ whenever $\lambda = o(\frac{l+\mu}{p})$ and $p = O(\frac{l+\mu}{\log^3(l+\mu)})$.

\newpage


%
\end{document}